\documentclass[12pt]{article} 

\title{A scattering construction for nonlinear wave equations on Kerr--Anti-de Sitter spacetimes}
\author{Gemma L. Hood\footnote{Universit\"at Leipzig, Institut f\"ur Theoretische Physik, Br\"uderstra{\ss}e 16, 04103 Leipzig, Bundesrepublik Deutschland, gemma\_louise.hood@uni-leipzig.de}}
\date{December 2024}

\usepackage{amsrefs}

\usepackage{hyperref}
\hypersetup{
    colorlinks,
    citecolor=black,
    filecolor=black,
    linkcolor=black,
    urlcolor=blue
}

\usepackage{fullpage}
\usepackage[a4paper, margin={1in}]{geometry}

\DeclareMathAlphabet{\mathpzc}{OT1}{pzc}{m}{it}

\usepackage{amssymb}
\usepackage{amsmath}
\usepackage{mdframed}
\usepackage{slashed}
\usepackage{comment}
\usepackage{cancel}    
\usepackage{amsthm}
\usepackage{subcaption}
\usepackage{graphicx}\usepackage{epstopdf}
\usepackage{upgreek}
\usepackage[scr]{rsfso}
\usepackage{gensymb}
\usepackage{enumitem}
\usepackage{enumerate}
\usepackage{breqn}
\usepackage{tensor}
\usepackage[dvipsnames]{xcolor}

\usepackage[toc,page]{appendix}

\theoremstyle{definition}
\newtheorem{definition}{Definition}[section]

\newtheorem*{remark}{Remark}

\newtheorem{theorem}[definition]{Theorem}
\newtheorem{corollary}[definition]{Corollary}
\newtheorem{lemma}[definition]{Lemma}
\newtheorem{proposition}[definition]{Proposition}

\newtheorem{assumption}[definition]{Condition}

\allowdisplaybreaks

\numberwithin{equation}{section}

\begin{document}

\maketitle

\begin{abstract} 
Existence of a large class of exponentially decaying solutions of the nonlinear massive wave equation $\Box_g\psi+\alpha\psi = \mathcal{F}(\psi,\partial\psi)$ on a Kerr--Anti-de Sitter exterior is established via a backwards scattering construction. Exponentially decaying data is prescribed on the future event horizon, and Dirichlet data on the timelike conformal boundary. The corresponding solutions exhibit the full functional degrees of freedom of the problem, but are exceptional in the sense that (even) general solutions of the forward, linear ($\mathcal{F}=0$) problem are known to decay at best inverse logarithmically \cite{SharpDecay}. Our construction even applies outside of the \textit{Hawking-Reall bound} on the spacetime angular momentum, in which case, there exist exponentially growing mode solutions of the forward problem \cite{Dold}. As for the analogous construction in the asymptotically flat case \cite{SchwzScattering}, the assumed exponential decay of the scattering data on the event horizon is exploited to overcome the gravitational blueshift encountered in the backwards construction.
\end{abstract}

\tableofcontents

\section{Introduction}
This paper concerns solutions of the massive scalar wave (or Klein-Gordon) equation\footnote{$\Box_g$ denotes the d'Alembertian $\Box_g\psi=\frac{1}{\sqrt{|\det g|}}\partial_{\mu}(\sqrt{|\det g|}g^{\mu\nu}\partial_{\nu}\psi)$ with respect to the metric $g$.}
\begin{align}
\label{Massive}
    \Box_g\psi + \alpha\psi = \mathcal{F}(\psi,\partial\psi),
\end{align}
on the exterior ($r\geq r_+$) of a fixed $(3+1)$-dimensional Kerr--Anti-de Sitter (Kerr-AdS) background $(\mathcal{M},g)$ with cosmological constant $\Lambda=-\frac{3}{\ell^2}$ ($\ell>0$), mass $M>0$ and future event horizon ($\mathcal{H}^{+}$) radius $r_+$ whose angular momentum $a$ satisfies\footnote{The second requirement ensures that expression (\ref{KAdSMetric}) defines a Lorentzian metric (one can check via direct calculation that $\det g$ is undefined when $|a|=\ell$, due to the vanishing of $\Xi$ and when $|a|>\ell$, there exists $\theta=\theta'$ at which $\det g|_{\theta'}>0$). The first bound ensures that (\ref{KAdSMetric}) describes a black hole spacetime (\textit{extremal} in the case of equality). See \cite{Extremality}.}
\begin{align}
\label{ExtremalityCond}
    |a|\leq\sqrt{\frac{M\ell^2}{r_+}-\ell^2-2r_+^2}\quad\text{ and }\quad |a|<\ell.
\end{align}
\par

Recall that, in \textit{Boyer-Lindquist} coordinates, the metric takes the familiar form
\begin{align}
    \label{KAdSMetric}
    g =& -\frac{\Delta_{-}-\Delta_{\theta}a^2\sin^2\theta}{\Sigma}\mathrm{d}t^2+\frac{\Sigma}{\Delta_{-}}\mathrm{d}r^2+\frac{\Sigma}{\Delta_{\theta}}\mathrm{d}\theta^2+\frac{\Delta_{\theta}(r^2+a^2)^2
    -\Delta_{-}a^2\sin^2\theta}{\Xi^2\Sigma}\sin^2\theta\mathrm{d}\tilde{\phi}^2\nonumber\\
    &-2\frac{\Delta_{\theta}(r^2+a^2)^2-\Delta_{-}}{\Xi\Sigma}a\sin^2\theta\mathrm{d}\tilde{\phi}\mathrm{d}t,
\end{align}
where $\Delta_{-}$, $\Delta_{\theta}$, $\Sigma$ and $\Xi$ are functions of $r$ and $\theta$, defined below in (\ref{MetricQuantities}). In practice, alternative $(t^*,r,\theta,\phi)$ coordinates will be used, as these remedy that Boyer-Lindquist coordinates cease to be regular at the event horizon. The spacetime metric in these coordinates, as well as the relation between the two coordinate systems is stated in Section \ref{sec:KAdS}. One may consult Figure \ref{fig:Penrose} for examples of hypersurfaces of constant $t$ and $t^*$ time respectively.\par 

In (\ref{Massive}), the mass $\alpha$ is assumed to satisfy the \textit{Breitenlohner-Freedman bound}, $\alpha<\frac{9}{4}$, within which positivity of the energy was first derived in the ``pure" Anti-de Sitter (AdS) context \cite{BF}. The range $\alpha>0$ corresponds to ``negative mass."  The class of nonlinearity $\mathcal{F}$ treated here will be defined in Section \ref{sec:Eqn}. For now, one should think of $\mathcal{F}$ as polynomial in first-order ``unit" derivatives $D\psi$ (see (\ref{UnitDerivatives})) of the solution and consider $\mathcal{F}=\frac{1}{r^2}(\partial_{t^*}\psi)^2$ as one explicit example. In particular, in contrast to the asymptotically flat case, no algebraic condition on $\mathcal{F}$ is required. Specifically, one need not impose the \textit{null condition} which is crucial in the asymptotically flat case \cite{Christodoulou, FritzJohn, Klainerman}. This is due to the fact that derivatives in the null directions decay at the same rate in $r$, meaning that there are no ``good" and ``bad" derivatives. However, some general $r$-weighted condition (\ref{FBound}) on $\mathcal{F}$ is required. Note already that condition (\ref{FBound}) allows a broad class of nonlinearities. In particular, we permit terms of the form $g^{\mu\nu}\partial_{\mu}\psi\partial_{\nu}\psi$ which mimic the first-order terms appearing when one writes the Einstein equation in harmonic gauge. This is discussed further in Section \ref{sec:F}.

\subsection{The main result}
Since asymptotically-AdS spacetimes are not globally hyperbolic, the pertinent setting in which to study (\ref{Massive}) is an \textit{initial boundary value problem}. Here, we restrict our attention to the Dirichlet problem
\begin{equation}
    \begin{cases}
    \label{Problematic}
        &\Box_g\psi + \alpha\psi = \mathcal{F}(\psi,\partial\psi),\\
        &(\psi,\partial_{t^*}\psi)|_{\{t^*=\tau\}}=(\psi_0,\psi_1),\quad r^{\frac{3}{2}-s}\psi|_{\mathcal{I}}=0,\quad s=\sqrt{\frac{9}{4}-\alpha},
    \end{cases}
\end{equation}
with initial conditions prescribed on a spacelike hypersurface $\{t^*=\tau\}$ and boundary conditions at the timelike conformal boundary $\mathcal{I}$, which can formally be viewed as the limiting hypersurface $\{r=\infty\}$ (see Figure \ref{fig:Penrose}). However, one may easily extend the discussion that follows to Neumann or Robin boundary conditions in the manner of Holzegel-Warnick \cite{Twisted}.\par

\begin{figure}
    \centering
    \includegraphics[width=0.15\textwidth]{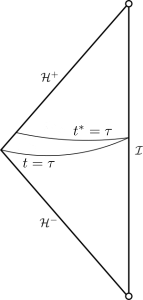}
    \caption{The Penrose diagram of the Kerr-AdS \textit{exterior} region $\{r\geq r_+\}$ depicting examples of constant $t$ and $t^*$ spacelike hypersurfaces.}
    \label{fig:Penrose}
\end{figure}

The objective of this paper is to construct and prove quantitative estimates for an expectedly non-generic, but nonetheless large, class of exponentially decaying solutions of (\ref{Problematic}). No symmetry assumptions will be made on the class of solutions which, moreover, exhibit the full functional degrees of freedom of the problem (\ref{Problematic}).  The main result can be stated informally as follows: 

\begin{theorem}[Exponentially decaying nonlinear waves on the Kerr-AdS exterior]
    \label{InformalMe}
    Under appropriate assumptions on $\mathcal{F}$ (see Section \ref{sec:F}), there exist a large class of (classical) solutions $\psi$ of (\ref{Massive}) with vanishing Dirichlet boundary conditions at $\mathcal{I}$, parameterised by appropriate \textit{scattering data} on $\mathcal{H}^+$, on the (sub)extremal\footnote{Note that this result does not assume the Hawking-Reall bound (\ref{HawkingReall}), outside of which exponentially \textit{growing} mode solutions are known to exist (see Theorem \ref{thm:Dold} below).} (\ref{ExtremalityCond}) Kerr-AdS exterior defined globally to the future of some constant $t^*$ hypersurface. These solutions, moreover, decay exponentially in time towards the future.
\end{theorem}
Theorem \ref{InformalMe} follows from recasting the initial boundary value problem (\ref{Problematic}) as a \textit{backwards} problem with exponentially decaying scattering data on the horizon, and constructing solutions via an iteration scheme. It is the exponentially decaying data which makes the solutions non-generic. However, it seems an extremely intricate problem to a priori characterise the initial data giving rise to such solutions as a set in the space of all Cauchy data for (\ref{Problematic}).\par 

\begin{remark}[Solutions outside of the span of quasinormal modes]
    As one immediate application of the construction of Theorem \ref{InformalMe}, it follows that, even for the linear ($\mathcal{F}=0$) equation (\ref{Massive}), there are solutions which do not lie in the span of so-called quasinormal modes. Indeed, one can prescribe the scattering data to vanish after finite time to construct non-trivial solutions which vanish to the future of some late $t^*=\text{constant}$ slice. These solutions do not lie in the span of quasinormal modes. See \cite{ClaudeQNM} where such solutions and their relation to quasinormal modes are discussed in the Schwarzschild-de Sitter case.
\end{remark}

Before discussing our construction to establish Theorem \ref{InformalMe} in more detail, the following section provides a brief review of what is known about the \textit{forward} problem (\ref{Problematic}).

\subsection{Comparison with the forward problem}
Understanding the behaviour of general solutions of the nonlinear initial boundary value problem (\ref{Problematic}) is key to the wider \textit{black hole stability problem} (see the discussion of Section \ref{sec:Motivation}) and is a difficult open question. Holzegel and Smulevici \cite{KGDecay} derived sharp logarithmic decay of general solutions of (\ref{Problematic}) for the class of \textit{linear} ($\mathcal{F}=0$) problems, in the setting where the \textit{Hawking-Reall bound}
\begin{align}
    \label{HawkingReall}
    |a|\ell<r_+^2
\end{align}
is satisfied. Informally, their result can be stated as follows.

\begin{theorem}[Logarithmic decay of general linear waves on the Kerr-AdS exterior \cite{KGDecay}]
    \label{InformalGustav}
    Consider a fixed Kerr-AdS background $(\mathcal{M},g)$ satisfying the Hawking-Reall bound (\ref{HawkingReall}). General solutions of the initial boundary value problem for
    \begin{equation}
    \label{LinEq}
        \Box_g\psi+\alpha\psi=0
    \end{equation}
    with Dirichlet boundary conditions on the subextremal Kerr-AdS exterior decay logarithmically and no faster.
\end{theorem}
See also the work of Gannot who constructs quasinormal mode solutions of (\ref{LinEq}) on Kerr-AdS \cite{Gannot1,Gannot2}.\par

When the Hawking-Reall bound is violated, behaviour of solutions of (\ref{Massive}) can generally be much worse. This was exhibited by Dold \cite{Dold} who produced solutions of (\ref{LinEq}) which grow exponentially in time, for parameters violating (\ref{HawkingReall}).

\begin{theorem}(Exponentially growing mode solutions of (\ref{LinEq}) \cite{Dold})
\label{thm:Dold}
    For any $\ell$, there is a Kerr-AdS exterior (violating the Hawking-Reall bound (\ref{HawkingReall})) for which equation (\ref{LinEq}) with Dirichlet boundary conditions at $\mathcal{I}$ admits an exponentially growing mode solution.
\end{theorem}

The proof of Theorem \ref{InformalGustav} relies on frequency space analysis of solutions of the forwards problem (\ref{Problematic}). This allows the authors to deal with geometric features of the Kerr black-hole exterior, such as superradiance and stable trapping, which cannot be fully understood in physical space. On the other hand, as we shall see, solving backwards from (sufficiently fast) \textit{exponentially decaying} scattering data for (\ref{Massive}) not only makes it possible to treat the more general \textit{nonlinear} $(\mathcal{F}\neq 0)$ equations of Theorem \ref{InformalMe}, but also allows one to use purely physical space methods since the exponential decay estimates render the procedure insensitive to both trapping and superradiance.\par

One may expect that, for slower than exponentially decaying scattering data on $\mathcal{H}^+$, the construction of Theorem \ref{InformalMe} can be generalised to again construct future global solutions on the exterior. Such solutions would now generically be singular at the event horizon, however. Characterising the scattering data which give rise to regular solutions at $\mathcal{H}^+$, which thus arise as solutions of the forward problem (\ref{Problematic}) with regular Cauchy data, is an interesting and difficult open problem.

\subsection{Outline of the proof}
As noted, Theorem \ref{InformalMe} is proven by way of a scattering construction, inspired by that of Dafermos, Holzegel and Rodnianski \cite{SchwzScattering} in the $\Lambda=0$ setting. This involves prescribing the exponentially fast decay of the solution to be constructed as data $h_{\mathcal{H}^+}$ along the event horizon of the Kerr-AdS black hole and evolving backwards to obtain a global towards the future solution of (\ref{Massive}).\par

It is worth noting that, conceptually, the construction on Kerr-AdS is not significantly more difficult than that on the spherically symmetric Schwarzschild-AdS spacetime.\par

The key difficulty to overcome, the \textit{blueshift effect} (see Section \ref{sec:Redshift} below), already arises at the level of treating a \textit{linear} problem (\ref{Limiting}). To extend to the full nonlinear problem, one must control an additional bulk term in the energy estimates. This is straightforward, provided $\mathcal{F}$ satisfies an appropriate assumption (\ref{FBound}).\par

The main step in the proof consists of obtaining a priori energy estimates for a sequence of approximating problems. We give an outline via a discussion of the linear and nonlinear aspects of the problem.\par

\subsubsection{Linear aspects: The redshift effect}
\label{sec:Redshift}
In studying the forward problem (\ref{Problematic}), one encounters \textit{gravitational redshift} (a feature of black holes with positive \textit{surface gravity}\footnote{The surface gravity of a Killing horizon is the constant $\kappa$ which satisfies $\nabla^{a}(T^b T_b)=-2\kappa T^a$, where $T$ is the Killing field which generates the horizon.}) at the event horizon, giving rise to an exponential decay mechanism. In the backwards problem, instead of this key stability mechanism, one sees a \textit{gravitational blueshift}.
Due to the lack of a globally uniformly timelike Killing vector field on Kerr-AdS, this means that \textit{any} energy estimate generated by a timelike vector field will produce a non-vanishing bulk term which could, potentially, drive exponential growth. This should be contrasted with forward-in-time propagation where a (uniformly timelike) \textit{redshift vector field} $N$ can be carefully constructed to ensure all derivatives in the energy estimate appear with a good sign near the horizon. In order to close energy estimates and prove Theorem \ref{InformalMe}, we impose a rate of exponential decay on the scattering data which is stronger than the blueshift-generated exponential growth. Consequently, the methods used in this paper do not allow for scattering data which decay any slower. It is unclear how one might address backwards evolution of (\ref{Massive}) in the setting of polynomially decaying data, for instance\footnote{In the $\Lambda=0$ setting, it is conjectured that polynomially decaying perturbations of the Kerr event horizon give rise to \textit{weak null singularities} \cite{SchwzScattering}. One can expect similar behaviour in the $\Lambda<0$ setting.}. However, in the \textit{extremal} setting (the setting of equality in the first condition appearing in (\ref{ExtremalityCond})), the redshift effect degenerates ($\kappa=0$), and so one may be able to close estimates under weaker decay assumptions.\par

\begin{figure}
    \centering
    \includegraphics[width=0.25\textwidth]{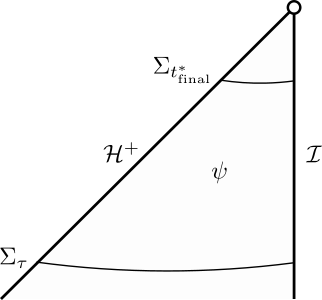}
    \caption{The region of integration for estimate (\ref{LayDownTheLaw}).}
    \label{fig:IntroRegion}
\end{figure}

To see how this works in practice, let us consider a backwards, finite-in-time, linear toy scattering problem which should be thought of as approximating, in the limit $t^*\to\infty$, the full problem
\begin{align}
\label{Limiting}
    \begin{cases}
        \Box_g\psi+\alpha\psi= 0,\quad\alpha<\frac{9}{4}\\
        \psi|_{\mathcal{H}^+}=h_{\mathcal{H}^+},\quad r^{\frac{3}{2}-s}\psi|_{\mathcal{I}}=0,\quad s=\sqrt{\frac{9}{4}-\alpha}.
    \end{cases}
\end{align}
In (\ref{Limiting}), the smooth scattering data $h_{\mathcal{H}^+}$ satisfies
\begin{align}
    \label{ADecay}
    F_{\mathcal{H}^+}[h_{\mathcal{H}^+}]\leq C\exp(-B\cdot\tau),
\end{align}
for some $C>0$ and some constant $B>0$ sufficiently large, depending on the surface gravity $\kappa$. Here, $F_{\mathcal{H}^+}$ is an $H^1$-norm of $h_{\mathcal{H}^+}$, a truncated form of which will appear as the horizon flux term in a subsequent energy inequality (\ref{LayDownTheLaw}). One must also assume higher-order analogues of (\ref{ADecay}) but, for now, the estimate will only be demonstrated to first order.\par

The toy problem we consider is one with mixed spacelike and null scattering data: 
\begin{align}
\label{ODEProblem}
    \begin{cases}
        \Box_g\psi+\alpha\psi= 0,\quad\alpha<\frac{9}{4}\\
        \psi|_{\mathcal{H}^+\cap\{t^*\leq t^*_{\text{final}}\}}=\tilde{\chi}(t^*)h_{\mathcal{H}^+},\quad(\psi,\partial_{t^*}\psi)|_{\Sigma_{t^*_{\text{final}}}}=0,\quad r^{\frac{3}{2}-s}\psi|_{\mathcal{I}}=0,\quad s=\sqrt{\frac{9}{4}-\alpha}.
    \end{cases}
\end{align}
Here, the scattering data on the horizon is truncated via an appropriate cut-off function $\tilde{\chi}$ which vanishes with all derivatives at $t^*_{\text{final}}$. Note that the energy $F_{\mathcal{H}^+\cap[\tau,t^*_{\text{final}}]}=F_{\mathcal{H}^+}[\tilde{\chi}(t^*)h_{\mathcal{H}^+}]$ of the cut-off scattering data is controlled by that without the cut-off (and in turn controlled by (\ref{ADecay})). Local well-posedness of problem (\ref{ODEProblem}) follows from a characteristic data analogue of \cite{WP}. That problem (\ref{ODEProblem}) indeed approximates (\ref{Limiting}) follows from the iteration procedure in Section \ref{sec:Proof} with $\mathcal{F}=0$. \par

We define the non-degenerate energy (expressed in regular coordinates $(t^*,r,\theta,\phi)$)
\begin{align*}
    E[\psi](t^*)=\int_{\Sigma_{t^*}}\bigg[\frac{1}{r^2}(\partial_{t^*}\psi)^2+r^2(\partial_r\psi)^2+|\slashed\nabla\psi|^2+\psi^2\bigg]\mathrm{vol}_{\Sigma_{t^*}}.
\end{align*}
Multiplying the equation in (\ref{ODEProblem}) by $N\psi$, where $N$ is a uniformly timelike vector field equal to the stationary Killing field $T=\partial_{t^*}$ at infinity, and integrating by parts over the spacetime strip $\mathcal{D}_{[\tau,t^*_{\text{final}}]}=\mathcal{M}\cap[\tau,t^*_{\text{final}}]$ yields the energy estimate\footnote{In fact, the zeroth-order term arising from using $N$ (or $T$) as a multiplier appears with a bad sign and, as such, some additional work is required to show that this multiplier controls $E[\psi]$. This is carried out in Section \ref{sec:Coercivity}.}
\begin{align}
\label{LayDownTheLaw}
    E[\psi](\tau)\leq F_{\mathcal{H}^+\cap[\tau,t^*_{\text{final}}]}+F_{\mathcal{I}\cap[\tau,t^*_{\text{final}}]}+C_{\kappa}\int_{\tau}^{t^*_{\text{final}}}E[\psi](t^{*}{'})\mathrm{d}t^{*}{'}.
\end{align}
Here, $F_{\mathcal{H}^+}$ and $F_{\mathcal{I}}$ are the flux terms on the event horizon and future null infinity, respectively. The constant $C_{\kappa}>0$ depends linearly on the surface gravity $\kappa$.\par

Applying the Dirichlet boundary conditions and assumed decay (\ref{ADecay}) along the horizon, (\ref{LayDownTheLaw}) becomes
\begin{align}
    \label{EgIneq}
    E[\psi](\tau)\leq C\exp(-B\cdot t^*)+C_{\kappa}\int_{\tau}^{t^*_{\text{final}}}E[\psi](t^{*}{'})\mathrm{d}t^{*}{'}.
\end{align}
Via a Gr\"onwall inequality, (\ref{EgIneq}) implies
\begin{align}
    E[\psi](\tau)\leq& C\exp(-B\cdot \tau)+\int_{\tau}^{t^*_{\text{final}}}C_k C\exp(-B\cdot t^*)\exp\bigg(\int_{\tau}^{t^*}C_{\kappa}\mathrm{d}s\bigg)\mathrm{d}t^*\nonumber\\
    =&C\exp(-B\cdot \tau)+\int_{\tau}^{t^*_{\text{final}}}C_k C\exp\Big((C_{\kappa}-B) t^*\Big)\mathrm{d}t^*\nonumber\\
    \leq&C\bigg(1+\frac{C_{\kappa}}{B-C_{\kappa}}\bigg)\exp(-B\cdot\tau)\nonumber\\
    =&\tilde{C}\exp(-B\cdot\tau),\label{PostGronwall}
\end{align}
with
\begin{align*}
    \tilde{C}=C\bigg(1+\frac{C_{\kappa}}{B-C_{\kappa}}\bigg)>0,
\end{align*}
provided that the rate $B>0$ of decay (\ref{ADecay}) along the horizon is strictly greater than the constant $C_{\kappa}$.

\subsubsection{Linear aspects: Controlling higher-order derivatives}
In order to prove Theorem \ref{InformalMe}, one requires estimates analogous to (\ref{PostGronwall}) for higher order energies. Commuting with $T$, $N$ and $\Phi=\partial_{\phi}$ (the Killing vector field associated with the axial symmetry of the background) gives non-degenerate control of higher order time (and mixed time) derivatives. This is carried out in Section \ref{sec:EgEst}, for a linear inhomogeneous problem. Commuting with $T$ and $\Phi$ is trivial, however, commuting with $N$ generates non-vanishing commutator terms. On Schwarzschild-AdS, one could also commute with the angular momentum operators (without generating further additional bulk terms) to control the missing angular derivatives. However, outwith spherical symmetry, one must take more care. There are many ways to resolve this (and the terms generated by commuting with $N$) but, for simplicity, here this is achieved in Section \ref{sec:Elliptic} by constructing estimates from the equation (\ref{Massive}). Note that the estimates obtained in this manner are not optimal. In particular the $L^{\infty}$ Sobolev embedding (Theorem \ref{SliceSobolev}) one can prove using these estimates gives only $r^{-\frac{1}{2}}$ decay of $\psi$, $|\slashed\nabla\psi|$. One could improve these estimates by commuting with the angular momentum operators and treating the resulting bulk contributions as error terms, yielding $r^{-\frac{3}{2}}$ decay of $\psi$ and angular derivatives with a stronger weight $|r\slashed\nabla\psi|$, plus further improvements, depending on the Breitenlohner-Freedman bound. However, it will become clear that the weaker decay is sufficient to close the argument, even in the nonlinear case.\par

To control all missing spatial derivatives at general order, we require a four level hierarchy of $L^2$ estimates, which appear in Section \ref{sec:Elliptic}. The first level (\textbf{I}) in the hierarchy controls general derivatives $D^{\sigma}\psi$ on the ``far" region (where $T$ is uniformly timelike). These elliptic estimates are stated in Section \ref{sec:FarElliptic}. Level (\textbf{II}) controls derivatives of the form $N^{\sigma_1}T^{\sigma_2}\Phi^{\sigma_3}D_x^2\psi$, i.e., those which involve only two derivatives in general spatial directions, on the entire exterior region. These estimates are derived in Section \ref{sec:Level1}. Level (\textbf{III}) contains the ``near" region analogues of those estimates at level (\textbf{I}), and are derived in Section \ref{sec:GenNear}. Finally, the fourth level (\textbf{IV}) of $L^2$ estimates, appearing in Section \ref{sec:Level2}, control general derivatives $D^{\sigma}\psi$ on the entire exterior region. This four-level hierarchy of estimates must already be carried out at the linear level. In the following, we include terms arising from the nonlinearity which must eventually be dealt with by nonlinear estimates, but note that the estimates hold at the linear level without these terms being present.\par

\addtocontents{toc}{\protect\setcounter{tocdepth}{2}}
\subsubsection*{I Estimating $D^{\sigma}\psi$ on the ``far" region}
\addtocontents{toc}{\protect\setcounter{tocdepth}{3}}
To construct the elliptic estimate of Proposition \ref{prop:FarKthOrder} (appearing in Section \ref{sec:FarElliptic}), we begin by estimating all second order spatial derivatives on the ``far" region (captured by Proposition \ref{prop:Far2ndOrder}). To do so, we rewrite the equation (\ref{Massive}) (in regular $(t^*,r,\theta,\phi)$ coordinates) in the form
\begin{align}
    \label{TElliptic}
    (\slashed\Delta + g^{rr}\partial_r^2+2g^{\phi r}\partial_{\phi}\partial_r)\psi = &\Box_g\psi+\alpha\psi+\ldots,
\end{align}
where $\slashed\Delta$ denotes the Laplace-Beltrami operator of the spheres of constant $(t^*,r)$. The $\ldots$ represent terms controlled by the first-order energy and the energy after one commutation with the Killing vector field $T$. On the ``far" region, where $T$ is timelike, the operator in brackets on the left-hand side of (\ref{TElliptic}) is elliptic. Via a standard argument, this allows one to control all spatial derivatives on the ``far" region by commuting with $T$ (and eventually an estimate for the nonlinearity $\mathcal{F}=\Box_g\psi+\alpha\psi$. See Proposition \ref{prop:Far2ndOrder}):
\begin{align}
    &\int_{\Sigma_{t^*}}\tilde{\xi}(r)\bigg(|\slashed\nabla^2\psi|^2+r^2|\slashed\nabla\partial_r\psi|^2+r^4(\partial_r^2\psi)^2\bigg)r^2\mathrm{d}r\mathrm{d}\omega\nonumber\\
        \leq& C\bigg(E[\psi](t^*)+E[T\psi](t^*)+\int_{\Sigma_{t^*}}(\Box_g\psi+\alpha\psi)^2r^2\mathrm{d}r\mathrm{d}\omega\bigg),\label{SchemFar}
\end{align}
for some $C>0$. Here, $\tilde{\xi}(r)$ is supported on the ``far" region where $T$ is timelike. Successively applying $\Box_g$ to (\ref{TElliptic}) and commuting sufficiently many times with $T$ then allows one to control all derivatives on the ``far" region in this manner (Proposition \ref{prop:FarKthOrder}):
\begin{align}
         &\sum_{|\sigma|= k}\int_{\Sigma_{t^*}}\tilde{\xi}(r)r^{2\sigma_2}|\slashed\nabla^{\sigma_1}\partial_{t^*}^{\sigma_2}\partial_r^{\sigma_3}\psi|^2 r^{2}\mathrm{d}r\mathrm{d}\omega\nonumber\\
        \leq &C\sum_{\substack{|\alpha|\leq k-1,\\ |\beta|\leq k- 2,\\|\gamma|\leq\lceil\frac{k-2}{2}\rceil}}\bigg[E[T^{\alpha}\psi](t^*)+\int_{\Sigma_{t^*}}\bigg(\Big(T^{\beta}(\Box_g\psi+\alpha\psi)\Big)^2+\Big(\Box_{g}^{\gamma}(\Box_g\psi+\alpha\psi)\Big)^2 \bigg)r^2\mathrm{d}r\mathrm{d}\omega\bigg],\label{SchemKFar}
     \end{align}
for some $C>0$, $k\geq 2$. The last two terms on the right-hand side of (\ref{SchemKFar}) will eventually be controlled by an estimate for $\mathcal{F}$.

\addtocontents{toc}{\protect\setcounter{tocdepth}{2}}
\subsubsection*{II(A) Estimating $D_x^2\psi$ on the entire exterior}
\addtocontents{toc}{\protect\setcounter{tocdepth}{3}}
To construct the estimate of Proposition \ref{prop:2ofk} (appearing in Section \ref{sec:Level1}) for derivatives $N^{\sigma_1}T^{\sigma_2}\Phi^{\sigma_3}D_x^{2}\psi$, we begin by estimating all second order spatial derivatives $D_x^2\psi$. To do so, we first note that the second term on the left-hand side of (\ref{TElliptic}) degenerates at the event horizon (where $g^{rr}=0$). Furthermore, the third term carries a potentially bad sign which plays a role outside of the large-$r$ regime. To deal with the latter, we rewrite (\ref{TElliptic}) as
\begin{align}
    \label{TElliptic2}
    \slashed\Delta\psi + g^{rr}\partial_r^2\psi = &\Box_g\psi+\alpha\psi+\ldots,
\end{align}
where the $\ldots$ represent terms controlled by the first-order energy and the energy after one commutation with each of the Killing vector fields $T$, $\Phi$. Away from the horizon (where $g^{rr}\neq 0$), (\ref{TElliptic2}) allows one to estimate all second order spatial derivatives by the energy after one commutation with each of $T$, $\Phi$ (and eventually estimates for the nonlinearity $\mathcal{F}=\Box_g\psi+\alpha\psi$). To counteract the degeneration at the horizon, we rewrite the $T^2\psi$ term on the right-hand side of (\ref{TElliptic2}) in terms of the redshift vector field $N$, which (close to the horizon) carries a non-degenerate $\partial_r$ term:
\begin{align*}
    N = \Big(1+\sqrt{-g^{t^*t^*}}\Big)\partial_{t^*}-\frac{g^{t^*r}}{\sqrt{-g^{t^*t^*}}}\partial_r-\frac{g^{t^*\phi}}{\sqrt{-g^{t^*t^*}}}\partial_{\phi}.
\end{align*}
This gives
\begin{align}
    \label{SchematicElliptic}
    \slashed\Delta\psi + \Big(g^{rr}+\xi(r)\Big)\partial_r^2\psi= &\Box_g\psi+\alpha\psi+\mathcal{O}\bigg(\frac{1}{r^2}\bigg)N^2\psi-\ldots,
\end{align}
where $\xi(r)$ is a cut-off function supported on the horizon and up to a finite radius. The $\ldots$, as before, represent terms controlled by the first-order energy and the energy after one commutation with each of $T$, $\Phi$. Multiplying (\ref{SchematicElliptic}) by $\slashed\Delta\psi$, integrating by parts on the left and applying Cauchy-Schwarz on the right allows one to obtain control of $\slashed\nabla^2\psi$, $\slashed\nabla\partial_r\psi$, provided care is taken with the $r$-weights which appear. The first two terms on the right-hand side of (\ref{SchematicElliptic}) will eventually be controlled by estimates for the nonlinearity $\mathcal{F}$. This gives
\begin{align}
    &\int_{\Sigma_{t*}}\bigg((\slashed\Delta\psi)^2+r^2|\slashed\nabla\partial_r\psi|^2\bigg)r^2\mathrm{d}r\mathrm{d}\omega\nonumber\\
    \leq& C\bigg[E[\psi](t^*)+E[T\psi](t^*)+E[N\psi](t^*)+E[\Phi\psi](t^*)+\int_{\Sigma_{t^*}}(\Box_g\psi+\alpha\psi)^2r^2\mathrm{d}r\mathrm{d}\omega\bigg],\label{PreClose}
\end{align}
for some $C>0$. The integrand of the final term on the right-hand side of (\ref{PreClose}) will eventually be replaced by $\mathcal{F}^2$ and estimated. Estimate (\ref{PreClose}) does not close immediately due to the presence of the $E[N\psi]$ term. However, after commuting sufficiently many times with $N$, this is eventually controlled (via a Gr\"onwall inequality) by estimates for the (higher order) $N$-commuted energy, which involve commutator terms supported on a compact $r$ region (see Proposition \ref{prop:LeaveFAlone}). The second order radial derivative is controlled on the near region by $E[N\psi]$ and on the far region by estimate (\ref{SchemFar}), so combining (\ref{SchemFar}) and (\ref{PreClose}) yields (see Proposition \ref{prop:SpacelikeElliptic})
\begin{align}
    &\int_{\Sigma_{t*}}\bigg((\slashed\Delta\psi)^2+r^2|\slashed\nabla\partial_r\psi|^2+r^2(\partial_r^2\psi)^2\bigg)r^2\mathrm{d}r\mathrm{d}\omega\nonumber\\
    \leq& C\bigg[E[\psi](t^*)+E[T\psi](t^*)+E[N\psi](t^*)+E[\Phi\psi](t^*)+\int_{\Sigma_{t^*}}(\Box_g\psi+\alpha\psi)^2r^2\mathrm{d}r\mathrm{d}\omega\bigg],
\end{align}
for some $C>0$.

\addtocontents{toc}{\protect\setcounter{tocdepth}{2}}
\subsubsection*{II(B) Estimating $N^{\sigma_1}T^{\sigma_2}\Phi^{\sigma_3}D_x^{2}\psi$ on the entire exterior}
\addtocontents{toc}{\protect\setcounter{tocdepth}{3}}

To produce the higher order estimates in this level of the hierarchy (Proposition \ref{prop:2ofk}), we successively commute (\ref{SchematicElliptic}) with $N$, $T$ and $\Phi$ and repeat the same procedure. Since $N$ does not commute with $\Box_g$, doing so generates additional terms at each step. However, these terms are always of lower order in $N$ and involve at most two general derivatives, so are immediately controlled by lower order estimates. This gives (see Proposition \ref{prop:2ofk})
\begin{align}
\label{SchemCom1Ellip}
    &\sum_{|\sigma|=k-2}\int_{\Sigma_{t*}}\bigg((\slashed\Delta \Gamma^{\sigma}\psi)^2+r^2|\slashed\nabla \partial_r\Gamma^{\sigma}\psi|^2+r^2(\partial_r^2\Gamma^{\sigma}\psi)^2\bigg)r^2\mathrm{d}r\mathrm{d}\omega\nonumber\\
    \leq& C\sum_{\substack{|\alpha|\leq k-1,\\|\beta|\leq k-2}}\bigg[E[\Gamma^{\alpha}\psi](t^*)+\int_{\Sigma_{t^*}}\Big(\Gamma^{\beta}(\Box_g\psi+\alpha\psi)\Big)^2 r^2\mathrm{d}r\mathrm{d}\omega\bigg],
\end{align}
for some $C>0$, $k\geq 2$, where $\Gamma^{\sigma}=N^{\sigma_1}T^{\sigma_2}\Phi^{\sigma_3}$.

\addtocontents{toc}{\protect\setcounter{tocdepth}{2}}
\subsubsection*{III(A) Estimating $N^{\sigma_1}T^{\sigma_2}\Phi^{\sigma_3}\partial_{\theta}^2\psi$ on the ``near" region}
\addtocontents{toc}{\protect\setcounter{tocdepth}{3}}
Before estimating all general derivatives at level \textbf{III} (Theorem \ref{thm:kNear}), as is carried out in Section \ref{sec:GenNear}, we note that (\ref{SchemCom1Ellip}) implies the estimate
 \begin{align}
        &\int_{\Sigma_{t^*}}\sum_{|\sigma|= k-2}\xi(r)(\partial_{t^*}^{\sigma_1}\partial_r^{\sigma_2}\partial_{\phi}^{\sigma_3}\partial_{\theta}^{2}\psi)^2r^2\mathrm{d}r\mathrm{d}\omega\nonumber\\
        \leq&C\sum_{\substack{|\alpha|\leq k-1,\\ |\beta|\leq k-2}}\bigg[E[\Gamma^{\alpha}\psi](t^*)+\int_{\Sigma_{t^*}}\Big(\Gamma^{\beta}(\Box_g\psi+\alpha\psi)\Big)^2 r^2\mathrm{d}r\mathrm{d}\omega\bigg],\label{Schem2Theta}
    \end{align}
for derivatives of second order in $\partial_{\theta}$. Here, $\xi(r)$ (as above) is supported on the ``near" region, where $N$ has a non-vanishing $r$-component. As such, it only remains to control derivatives $N^{\sigma_1}T^{\sigma_2}\Phi^{\sigma_3}\partial_{\theta}^{\sigma_4}\psi$ for $\sigma_4>2$. 

\addtocontents{toc}{\protect\setcounter{tocdepth}{2}}
\subsubsection*{III(B) Estimating $D^{\sigma}\psi$ on the ``near" region}
\addtocontents{toc}{\protect\setcounter{tocdepth}{3}}
Let us demonstrate how to estimate $\partial^3_{\theta}\psi$, $\partial^4_{\theta}\psi$.\par

Rearranging (\ref{TElliptic}) so that only terms involving $\partial_{\theta}$ derivatives remain on the left-hand side, one has
\begin{align}
    \frac{g^{\theta\theta}}{\sqrt{|\det g|}}\partial_{\theta}(\sqrt{|\det g|})\partial_{\theta}\psi+\partial_{\theta}(g^{\theta\theta})\partial_{\theta}\psi+\slashed\Delta\psi=\Box_g\psi+\alpha\psi+\ldots,\label{ApplyLap}
\end{align}
where $\ldots$ represent terms controlled on the ``near" region by the first-order energy and estimate (\ref{Schem2Theta}). Applying $\slashed\Delta$ to (\ref{ApplyLap}) gives
\begin{align}
    \slashed\Delta\bigg(\frac{g^{\theta\theta}}{\sqrt{|\det g|}}\partial_{\theta}(\sqrt{|\det g|})\partial_{\theta}\psi+\partial_{\theta}(g^{\theta\theta})\partial_{\theta}\psi+\slashed\Delta\psi\bigg)=\slashed\Delta(\Box_g\psi+\alpha\psi)+\ldots.\label{ApplyLap2}
\end{align}
The $\ldots$ on the right-hand side of (\ref{ApplyLap2}) involve terms with derivatives of at most second order in $\theta$, and so are controlled by (\ref{Schem2Theta}). The first term on the right-hand side of (\ref{ApplyLap2}) will eventually be controlled via an estimate for $\mathcal{F}$. Multiplying (\ref{ApplyLap2}) by $\slashed\Delta\psi$, applying the Cauchy-Schwarz inequality to the right-hand side and integrating the $\slashed\Delta^2\psi\slashed\Delta\psi$ term on the left-hand side by parts in the angular directions gives control of the $\partial_{\theta}^3\psi$ derivative. The other terms on the left-hand side of (\ref{ApplyLap2}) can be absorbed via $\varepsilon$-Cauchy-Schwarz. This gives the estimate
\begin{align}
    &\int_{\Sigma_{t^*}}\xi(r)|\slashed\nabla^3\psi|^2 r^2\mathrm{d}r\mathrm{d}\omega\nonumber\\
  \leq& C\sum_{\substack{|\alpha|\leq 3,\\ |\beta|\leq 2}}\bigg[E[\Gamma^{\alpha}\psi](t^*)+\int_{\Sigma_{t^*}}\bigg(\Big(\Gamma^{\beta}(\Box_g\psi+\alpha\psi)\Big)^2+\Big(\slashed\Delta(\Box_g\psi+\alpha\psi)\Big)^2\bigg) r^2\mathrm{d}r\mathrm{d}\omega\bigg]\nonumber,
\end{align}
for some $C>0$.

Repeating with the higher order multiplier $\slashed\Delta^2\psi$ yields control of $\partial^4_{\theta}\psi$:
\begin{align*}
    &\int_{\Sigma_{t^*}}\xi(r)|\slashed\nabla^4\psi|^2r^2\mathrm{d}r\mathrm{d}\omega\nonumber\\
    \leq& C\sum_{\substack{|\alpha|\leq 3,\\ |\beta|\leq 2}}\bigg[E[\Gamma^{\alpha}\psi](t^*)+\int_{\Sigma_{t^*}}\bigg(\Big(\Gamma^{\beta}(\Box_g\psi+\alpha\psi)\Big)^2+\Big(\slashed\Delta(\Box_g\psi+\alpha\psi)\Big)^2\bigg) r^2\mathrm{d}r\mathrm{d}\omega\bigg]\nonumber,
\end{align*}
for some $C>0$.

By successively applying $\slashed\Delta$ and higher-order $\slashed\Delta\psi$ multipliers, one controls all $\partial_{\theta}^{\sigma}$ derivatives. Commuting with $T$, $N$ and $\Phi$ then gives (Theorem \ref{thm:kNear})
\begin{align}
        &\sum_{|\sigma|=k}\int_{\Sigma_{t^*}}\xi(r)r^{2\sigma_2}|\slashed\nabla^{\sigma_1}\partial_{t^*}^{\sigma_2}\partial_r^{\sigma_3}\psi|^2 r^{2}\mathrm{d}r\mathrm{d}\omega\nonumber\\
        \leq &C\sum_{\substack{|\alpha|\leq k-1,\\ |\beta|\leq k- 2,\\|\gamma|\leq\big\lceil\frac{k-2}{2}\big\rceil}}\bigg[E[\Gamma^{\alpha}\psi](t^*)+\int_{\Sigma_{t^*}}\bigg(\Big(\Gamma^{\beta}(\Box_g\psi+\alpha\psi)\Big)^2+\Big(\slashed\Delta^{\gamma}(\Box_g\psi+\alpha\psi)\Big)^2\bigg) r^2\mathrm{d}r\mathrm{d}\omega\bigg]\label{SchemNearK}
    \end{align}
for some $C>0$, $k\geq 2$.

\addtocontents{toc}{\protect\setcounter{tocdepth}{2}}
\subsubsection*{IV Estimating $D^{\sigma}\psi$ on the entire exterior}
\addtocontents{toc}{\protect\setcounter{tocdepth}{3}}
Finally, combining estimates (\ref{SchemKFar}) and (\ref{SchemNearK}) gives control of general derivatives on the entire exterior, as appears in Section \ref{sec:Level2} (Theorem \ref{thm:KElliptic}):
\begin{align}
        &\sum_{|\sigma|=k}\int_{\Sigma_{t^*}}r^{2\sigma_2}|\slashed\nabla^{\sigma_1}\partial_{t^*}^{\sigma_2}\partial_r^{\sigma_3}\psi|^2 r^{2}\mathrm{d}r\mathrm{d}\omega\nonumber\\
        \leq &C\sum_{\substack{|\alpha|\leq k-1,\\ |\beta|\leq k- 2,\\|\gamma|\leq\big\lceil\frac{k-2}{2}\big\rceil}}\bigg[E[\Gamma^{\alpha}\psi](t^*)+\int_{\Sigma_{t^*}}\bigg(\Big(\Gamma^{\beta}(\Box_g\psi+\alpha\psi)\Big)^2+\Big(\Box_g^{\gamma}(\Box_g\psi+\alpha\psi)\Big)^2\bigg) r^2\mathrm{d}r\mathrm{d}\omega\bigg]\nonumber,
    \end{align}
    for some $C>0$, $k\geq 2$.

\subsubsection{The nonlinear problem}
Let us now turn to the full nonlinear problem. We once again consider a finite-in-time approximation of the full nonlinear problem, analogous to (\ref{ODEProblem}) with a non-zero nonlinearity $\mathcal{F}$ on the right-hand side of the equation. This forms part of an iteration scheme (see Section \ref{sec:FiniteProblems}) whose limit will provide the desired result for the full problem. Local well-posedness for the finite-in-time nonlinear problems will not be proven here. However, it follows from an adaptation of the linear well-posedness result \cite{WP}.\par

The key ingredient of the proof of Theorem \ref{InformalMe} is an energy estimate, in the spirit of (\ref{LayDownTheLaw}), but now with an additional term arising from the nonlinearity $\mathcal{F}$:
\begin{align}
     E[\psi](\tau)\leq F_{\mathcal{H}^+\cap[\tau,t^*_{\text{final}}]}+\int_{\tau}^{t^*_{\text{final}}}\bigg(C_{\kappa}E[\psi](t^{*}{'})+\int_{\Sigma_{t^*}}|\mathcal{F}(\psi,\partial\psi)\cdot N\psi|\mathrm{vol}_{\Sigma_{t^*}}\bigg)\mathrm{d}t^{*}{'}.\label{SchematicNonlinear}
\end{align}
In order to proceed as before, using the assumed exponential decay and a Gr\"onwall inequality to derive exponential decay of $E[\psi](\tau)$, one must be able to bound this new term appropriately. Given that no null condition is required, this is easily insured, provided the general Condition \ref{NonlinAssump} is satisfied. In particular, one must ensure that an appropriately $r$-weighted spatial $L^2$ norm of $\mathcal{F}$ and its derivatives is well-controlled by the energies (\ref{FBound}). For simplicity it is assumed (assumption (\ref{Quadratic})) that $\mathcal{F}$ is quadratic in the unit derivatives $D\psi$ (\ref{UnitDerivatives}). It is shown in Proposition \ref{prop:AbstractEnergyEst} that this class of $\mathcal{F}$ indeed satisfies Condition \ref{NonlinAssump}. The restriction on $\mathcal{F}$ could be weakened to allow for additional growth in $r$ by commuting with angular momentum operators to deduce stronger radial decay, as discussed earlier. However, the class satisfying (\ref{Quadratic}) can be readily adapted without additional decay to the key application of interest, the Einstein vacuum equation in an appropriate gauge. This is discussed further in Section \ref{sec:F}. Condition \ref{NonlinAssump} easily allows for more general nonlinearities than those of the assumed form (\ref{Quadratic}), including higher order polynomials of $D\psi$, as well as additional linear terms.\par

Provided $\mathcal{F}$ satisfies appropriate assumptions (see Condition \ref{NonlinAssump}), one may then proceed as in the linear setting, commuting (\ref{SchematicNonlinear}) with $T$ and the redshift vector field $N$ (Section \ref{sec:EgEst}) and applying a Gr\"onwall inequality to derive exponential decay of the energies $E[T^{\alpha}N^{\beta}\psi]$ to sufficiently high order (Section \ref{sec:FiniteEstimates}). As was the case in the linear setting, one must employ elliptic estimates and the equation to control the ``missing" angular and radial derivatives (as is done in Section \ref{sec:Elliptic}).

\subsection{Relation to the Einstein equation}
\label{sec:Motivation}
It remains to present the problem in its wider context. The decay behaviour of solutions of (\ref{Problematic}) is closely related to the question of stability of the Kerr-AdS spacetimes as solutions of the \textit{Einstein vacuum equation} (EVE) 
\begin{align}
\label{Einstein}
    R_{\mu\nu}-\frac{1}{2}Rg_{\mu\nu}+\Lambda g_{\mu\nu} = 0
\end{align}
with negative cosmological constant $\Lambda=-\frac{3}{\ell^2}$ and the \textit{black hole stability problem}, more generally. In an appropriate choice of coordinates, (\ref{Einstein}) can be written as a system of wave equations. As such, it is expected that Theorem \ref{InformalMe} will play a central role in deducing existence of a class of exponentially decaying gravitational perturbations of the Kerr-AdS black hole exterior. The technical details of applying Theorem \ref{InformalMe} to (\ref{Einstein}) will be carried out in subsequent work. In order to move from Theorem \ref{InformalMe} to an analogous statement in full gravity, one must deal with two concerns - firstly, the \textit{quasilinear} nature of (\ref{Einstein}); secondly, its \textit{gauge-invariance}. The first concern is straightforward to address - Theorem \ref{InformalMe} is written as a statement for \textit{semilinear} equations, purely for the purpose of simplicity. One may easily extend to quasilinear equations mimicking the structure of (\ref{Einstein}), as one need not capture a null condition in the nonlinearity. The insensitivity of the scattering construction to trapping means, moreover, that no additional difficulties arise from treating quasilinear nonlinearities. The key difficulty of working with (\ref{Einstein}) is rather the second concern - that of fixing a gauge. There are many possible choices of gauge in which to work with (\ref{Einstein}), \textit{harmonic} and \textit{double-null gauge} being two commonly used in the literature. Let us address each of these in turn, highlighting their advantages and disadvantages, and the degree to which they translate (\ref{Einstein}) into a system of equations treatable via the methods appearing in this work.\par

In \textit{harmonic gauge}, (\ref{Einstein}) takes the form of a system of quasilinear massive wave equations, as derived by Enciso and Kamran \cite{EncisoKamran} in their well-posedness result for (\ref{Einstein}) in the asymptotically AdS setting. There, one sees that the system carries little in the way of geometric structure and admits multiple different masses $\alpha$. Each of these masses lies in the range to which Theorem \ref{InformalMe} applies (see the discussion in Section \ref{sec:F}).\footnote{Furthermore, note that the Sobolev spaces used herein are mass independent, so treating multiple masses is straightforward.} A drawback of harmonic gauge particular to the scattering problem is that it is not well-adapted to characteristic hypersurfaces on which one would like to pose data.\par

On the other hand, a \textit{double-null gauge} is, by design, suited to characteristic hypersurfaces. Under this choice of gauge, (\ref{Einstein}) becomes a system of wave and transport equations for spacetime null curvature components. Although the system of equations in double-null gauge is considerably more cumbersome than in the harmonic gauge case, one has the aesthetic advantage of wave equations admitting only the conformal mass $\alpha=2$ in $(3+1)$-dimensions. In linearised gravity with double-null gauge, two gauge-independent curvature components, $\tilde{\alpha}^{[s]}$ for $s=\pm 2$, satisfy the Teukolsky equations\footnote{One may see \cite{Teukolsky}, for example, for work in this direction in the Kerr-AdS setting.} - wave equations of the form (\ref{Massive}) with masses $\alpha=2$.\par

Given that one expects, via an appropriate choice of gauge, that Theorem \ref{InformalMe} can be applied to (\ref{Einstein}) to obtain a class of exponentially decaying gravitational perturbations of Kerr-AdS spacetimes, let us now address the wider context of the black hole stability problem. Moschidis \cites{NullDust,EinsteinVlasov} showed instability of pure AdS in the context of (\ref{Einstein}) in spherical symmetry coupled to appropriate matter models. Previously, Biz\`on and Rostworowski \cite{Numerics} derived an instability mechanism based on resonant frequencies and found instability for the spherically symmetric scalar field using numerical methods. Moreover, Theorem \ref{InformalGustav} suggests that general perturbations of Kerr-AdS spacetimes decay insufficiently fast to allow one to prove stability. Thus the authors of \cite{KGDecay} conjecture instability\footnote{Note, however, that \textit{linear} stability of the Schwarzschild-AdS subfamily has been proven \cite{Olivier1,Olivier2,Olivier3}.} of the Kerr-AdS exterior as a solution of (\ref{Einstein}) (see also, however, the discussion in \cite{Physicists}). As such, the class of exponentially decaying perturbations suggested by Theorem \ref{InformalMe} are expectedly non-generic. 

\subsection{Outline of the paper}
The paper has the following structure. Section \ref{Prelim} introduces the Kerr-AdS spacetime metric and associated quantities, as well as some vector fields which will be key to proving Theorem \ref{InformalMe}. The precise class of wave equations (\ref{Massive}) considered and associated boundary conditions are then specified. Finally, the energy spaces in which the analysis is carried out are defined, and some important energy quantities recalled. Section \ref{sec:Estimates} presents several physical space estimates which will be key for the proof. In Section \ref{sec:Result}, the main result (formulated informally in Theorem \ref{InformalMe}) is presented rigorously as Theorem \ref{MainResult}. Section \ref{sec:Proof} is concerned with the proof of Theorem \ref{MainResult}.

\section{Preliminaries}
\label{Prelim}
Before presenting the main result, the set up and key quantities are noted here. Firstly, in Section \ref{sec:KAdS} the Kerr-AdS metric and relevant spacetime vector fields are introduced. The equation, associated boundary conditions and scattering problem formulation follow in Section \ref{sec:Eqn}. In particular, the class of nonlinearity $\mathcal{F}$ in (\ref{Massive}) is specified at this stage. Section \ref{sec:Energies} presents the energy spaces in which solutions of the scattering problem will be analysed. Finally, in Section \ref{sec:EMTensor}, the energy-momentum tensor, associated quantities and energy identity appear.

\subsection{The Kerr-AdS family of spacetimes}
\label{sec:KAdS}
\subsubsection{The spacetime manifold and metric}
For fixed negative cosmological constant $\Lambda=-\frac{3}{\ell^2}$, the \textit{Kerr--Anti-de Sitter} (Kerr-AdS) spacetimes are a two-parameter family of solutions $(\mathcal{M}',g')$ of the Einstein vacuum equation (\ref{Einstein}). Each fixed choice of mass $M>0$ and each angular momentum $a$ such that (\ref{ExtremalityCond}) holds, where $r_+$ (the largest real root of $\Delta_-$, see (\ref{MetricQuantities}) below) is the event horizon radius, give a stationary, axisymmetric, rotating black hole with metric
\begin{align}
    \label{KAdS*Metric}
    g=&-\frac{\Delta_{-}-\Delta_{\theta}a^2\sin^2\theta}{\Sigma}\mathrm{d}{t^*}^2+\frac{1}{\Sigma(1+\frac{r^2}{\ell^2})^2}\bigg(\Delta_+-a^2\sin^2\theta\bigg(\bigg(1+\frac{r^2}{\ell^2}\bigg)+\frac{\Sigma}{\ell^2}\bigg)\bigg)\mathrm{d}r^2+\frac{\Sigma}{\Delta_{\theta}}\mathrm{d}\theta^2\nonumber\\
    &+\frac{\Delta_{\theta}(r^2+a^2)^2-\Delta_{-}a^2\sin^2\theta}{\Xi^2\Sigma}\sin^2\theta\mathrm{d}\phi^2-2\frac{\Delta_{\theta}(r^2+a^2)^2-\Delta_{-}}{\Xi\Sigma}a\sin^2\theta\mathrm{d}\phi\mathrm{d}t^*\nonumber\\
    &+\frac{2}{(1+\frac{r^2}{\ell^2})}\bigg(\frac{2Mr}{\Sigma}-\frac{a^2}{\ell^2}\sin^2\theta\bigg)\mathrm{d}t^*\mathrm{d}r
    -\frac{2a\sin^2\theta}{(1+\frac{r^2}{\ell^2})}\bigg(1+\frac{2Mr}{\Xi\Sigma}\bigg)\mathrm{d}\phi\mathrm{d}r
\end{align}
in coordinates $(t^*,r,\theta,\phi)$ on the \textit{exterior region} $\mathcal{M}=\mathbb{R}\times[r_+,\infty)\times S^2$. Here,
\begin{align}
    \Delta_{\pm} &=(r^2+a^2)\bigg(1+\frac{r^2}{\ell^2}\bigg)\pm 2Mr \qquad &\Sigma &= r^2+a^2\cos^2\theta  \nonumber\\
    \Delta_{\theta} &= 1-\frac{a^2}{\ell^2}\cos^2\theta &\Xi &= 1-\frac{a^2}{\ell^2},\label{MetricQuantities}
\end{align}
One can check that \textit{Anti-de Sitter} space and the one-parameter \textit{Schwarzschild-Anti-de Sitter} family arise as special cases for the choices $a=M=0$, or $a=0$, $M>0$, respectively.\par

We shall denote by $S^2_{t^*r}$ the sphere at a given time $t^*$ and radius $r$, with induced metric
\begin{align*}
    \slashed g = \frac{\Sigma}{\Delta_{\theta}}\mathrm{d}\theta^2+\frac{\Delta_{\theta}(r^2+a^2)^2-\Delta_{-}a^2\sin^2\theta}{\Xi^2\Sigma}\sin^2\theta\mathrm{d}\phi^2.
\end{align*}
In coordinates, the corresponding gradient and Laplace operators take the form
\begin{align*}
    \slashed\nabla f = \slashed g^{AB}\slashed\nabla_{A}f\cdot\slashed\nabla_{B},\quad
    \slashed\Delta = \slashed g^{AB}\slashed\nabla_{A}\slashed\nabla_{B}.
\end{align*}
Here, $f$ is a sufficiently smooth function.\par

The spacelike hypersurfaces $\Sigma_{t^*}$ of constant $t^*$ foliate the exterior region $\{r\geq r_+\}$, depicted in Figure \ref{fig:Penrose}. We shall denote by $n_{{\Sigma}_{t^*}}$ their future-directed unit normal vector
\begin{align}
\label{Norman}
    n_{\Sigma_{t^*}}=\sqrt{-g^{t^*t^*}}\partial_{t^*}-\frac{g^{t^*r}}{\sqrt{-g^{t^*t^*}}}\partial_r-\frac{g^{t^*\phi}}{\sqrt{-g^{t^*t^*}}}\partial_{\phi}
\end{align}
and by $\mathrm{vol}_{\Sigma_{t^*}}$ their induced volume form
\begin{align*}
    \mathrm{vol}_{\Sigma_{t^*}} = \frac{\Sigma}{\Xi}\sin\theta\sqrt{-g^{t^*t^*}}\mathrm{d}r\mathrm{d}\theta\mathrm{d}\phi.
\end{align*}
We denote by $\mathrm{d}\omega$ the volume form on the spheres $S^2_{t^*,r}$
\begin{align*}
   \mathrm{d}\omega = \sin\theta\mathrm{d}\theta\mathrm{d}\phi. 
\end{align*}
We collect here the asymptotic (in $r$) behaviour of the inverse metric components $g^{\mu\nu}$
\begin{gather}
    \label{InverseAsymptotics}
    g^{t^*t^*}=\mathcal{O}\bigg(\frac{1}{r^2}\bigg),\quad g^{rr}=\mathcal{O}(r^2),\quad
    g^{\theta\theta} =\mathcal{O}\bigg(\frac{1}{r^2}\bigg), \quad g^{\phi\phi}=\mathcal{O}\bigg(\frac{1}{r^2}\bigg)\nonumber\\
    g^{t^*r}=\mathcal{O}\bigg(\frac{1}{r^3}\bigg),\quad g^{t^*\phi}= \mathcal{O}\bigg(\frac{1}{r^2}\bigg),\quad
    g^{r\phi}= \mathcal{O}\bigg(\frac{1}{r^2}\bigg)
\end{gather}
and determinant
\begin{align}
    \label{DetAsymptotics}
    \det g=\mathcal{O}(r^4)
\end{align}
which will be applied in later estimates where it is key that $r$-weights do not become too large at infinity.\par

One may alternatively employ \textit{Boyer-Lindquist} coordinates $(t,r,\theta,\tilde{\phi})$, defined implicitly by the relations
\begin{align*}
    \frac{\mathrm{d}t^*}{\mathrm{d}r}=\frac{2Mr}{\Delta_{-}(1+\frac{r^2}{\ell^2})}\qquad    \frac{\mathrm{d}\phi}{\mathrm{d}r}=\frac{a\Xi}{\Delta_{-}}\qquad\frac{\mathrm{d}t^*}{dt}=1\qquad\frac{\mathrm{d}\phi}{\mathrm{d}\tilde{\phi}}=1.
\end{align*}
Then the metric (\ref{KAdS*Metric}) takes its, perhaps, more familiar form (\ref{KAdSMetric}). However, since Boyer-Lindquist coordinates become irregular at the event horizon $\mathcal{H}^+$, we shall work in $(t^*,r,\theta,\phi)$ from this point onwards.

\subsubsection{Vector fields}
\label{sec:VFs}
One immediately obtains from (\ref{KAdS*Metric}) that the coordinate vector fields $T=\partial_{t^*}$, $\Phi=\partial_{\phi}$ are Killing. However, $T$ is not uniformly timelike in the black hole exterior. As such, we will use a uniformly timelike vector field
\begin{align}
\label{RedshiftVF}
    N = T+\xi(r)Y,\qquad    \xi(r) =
    \begin{cases}
        1\text{ on }r\leq r_0\\
        0\text{ on }r\geq r_1
    \end{cases},
    \qquad r_+<r_0<r_1<\infty
\end{align}
to obtain coercive energy estimates. Otherwise, control of derivatives transverse to the horizon degenerates. Throughout, $r_0$ is a fixed radius such that $T$ is uniformly timelike on $r\geq r_0$, $r_1>r_0$ is a sufficiently large (depending on $a$) fixed radius (see Lemma \ref{lem:FarRegion}) and $\xi(r)$ is a positive function, interpolating smoothly between its two cases. For simplicity, we choose $Y$ to be the future-directed unit normal (\ref{Norman}) to the spacelike hypersurfaces $\Sigma_{t^*}$.\par

Let us now define the set of \textit{commuting vector fields} 
\begin{align*}
    \Gamma \in \{N, T, \Phi\}.
\end{align*}
We shall denote by $\sigma=(\sigma_1,\sigma_2,\sigma_3)$ a multi-index with entries $\sigma_i$, $i\in\mathbb{N}$, and length $|\sigma|=\sum_i\sigma_i$. A general concatenation of the $\Gamma$-vector fields will be written as 
\begin{align*}
    \Gamma^{\sigma}=N^{\sigma_1}T^{\sigma_2}\Phi^{\sigma_3}.
\end{align*}
These will be applied to the equation (\ref{Massive}) in order to generate estimates for the solution $\psi$ to sufficiently high-order.\par

Moreover, define the sets of normalised spacetime derivatives
\begin{align}
\label{UnitDerivatives}
    D \in \bigg\{\frac{1}{r}\partial_{t^*}, r\partial_r, \frac{1}{r}\partial_{\theta}\psi, \frac{1}{r}\partial_{\phi}\psi\bigg\},\qquad\overline D \in \bigg\{\partial_{t^*}, r\partial_r, \frac{1}{r}\partial_{\theta}\psi, \frac{1}{r}\partial_{\phi}\psi\bigg\}
\end{align}
in terms of which the definition of the class of nonlinearity (Section \ref{sec:F}), the Sobolev embedding Theorem \ref{SliceSobolev} and Theorems \ref{MainResult}-\ref{BootstrapThm} will be stated. General concatenations of these derivatives will be denoted in a similar manner as for the commuting vector fields, with multi-index $\sigma=(\sigma_0,...,\sigma_3)$. For example,
\begin{align*}
    D^{\sigma}=\bigg(\frac{1}{r}\partial_{t^*}\bigg)^{\sigma_0}(r\partial_r)^{\sigma_1}\bigg(\frac{1}{r}\partial_{\theta}\bigg)^{\sigma_2}\bigg(\frac{1}{r}\partial_{\phi}\bigg)^{\sigma_3}.
\end{align*}
We denote the purely spatial normalised derivatives $r\partial_r$, $\frac{1}{r}\partial_{\theta}$ and $\frac{1}{r}\partial_{\phi}$ by $D_x$.

\subsection{The energies}
\label{sec:Energies}
We define the following weighted, non-degenerate energies on Kerr-AdS.\footnote{These are equivalent to the Cartesian ``twisted" energies appearing in \cite{EncisoKamran} when one restricts to Dirichlet boundary conditions.} We shall see that these arise from using the vector field $N$ as a multiplier. The gradient and its norm with respect to the induced metric on $S^2_{t^*,r}$ are denoted $\slashed\nabla$, $|\slashed\nabla\cdots\slashed\nabla\psi|^2=\slashed g^{A A'}\cdots\slashed g^{B B'}\slashed\nabla_A\cdots\slashed\nabla_B\psi\cdot\slashed\nabla_{A'}\cdots\slashed\nabla_{B'}\psi$ respectively.

\begin{definition}[Sobolev norms on spacelike slices]
Let $t^*$ be a fixed time. Then, for each $k$ $\in$ $\mathbb{N}$ we define the Sobolev space
\begin{align*}
    H^{k}_{KAdS}(\Sigma_{t^*})=\{\psi:[r_+,\infty)\times S^2_{t^*,r}\to\mathbb{R}\text{ }|\text{ }||\psi||_{H^k_{KAdS}(\Sigma_{t^*})}^2<\infty\}
\end{align*}
of scalar fields $\psi$ on $\Sigma_{t^*}$ associated with the energy
\begin{align}
\label{1stOrderEg}
    ||\psi||_{H^1_{KAdS}(\Sigma_{t^*})}^2 &= \int_{\Sigma_{t^*}}\bigg[\frac{1}{r^2}(\partial_{t^*}\psi)^2+r^2(\partial_r\psi)^2+|\slashed\nabla\psi|^2+\psi^2\bigg]r^2\mathrm{d}r\mathrm{d}\omega
\end{align}
and, for $k\geq 2$,
\begin{align}
    \label{kthEnergy}
    ||\psi||_{H^k_{KAdS}(\Sigma_{t^*})}^2 =& \int_{\Sigma_{t^*}}\bigg[\sum_{\substack{|\sigma|\leq k-1,\\|\rho|\leq k,\text{ }\rho_1\leq k-1}}\bigg(r^{2\sigma_3}|\slashed\nabla^{1+\sigma_1}\partial_{t^*}^{\sigma_2}\partial_r^{\sigma_3}\psi|^2+r^{2\rho_2}(\partial^{\rho_1}_{t^*}\partial^{\rho_2}_r\psi)^2\bigg)+\frac{1}{r^2}(\partial_{t^*}^k\psi)^2\bigg] r^2\mathrm{d}r\mathrm{d}\omega.
\end{align}
\end{definition}
We also define the following $L^p$-norms on spacelike slices of Kerr-AdS.
\begin{definition}[$L^p$-norms on spacelike slices]
    Let $t^*$ $\in$ $\mathbb{R}^+$ be a fixed time. Then for $p\in\{2,\infty\}$ we define the spaces
    \begin{align*}
        L^p_{KAdS}(\Sigma_{t^*})=\{\psi:[r_+,\infty)\times S^2_{t^*,r}\to\mathbb{R}\text{ }|\text{ }||\psi||_{L^p_{KAdS(\Sigma_{t^*})}}^2<\infty\}
    \end{align*}
    of scalar fields $\psi$ on $\Sigma_{t^*}$ associated with the norms
    \begin{align*}
        ||\psi||^2_{L^2_{KAdS}(\Sigma_{t^*})}=\int_{\Sigma_{t^*}}|\psi|^2 r^2\mathrm{d}r\mathrm{d}\omega,\qquad||\psi||_{L^{\infty}_{KAdS}(\Sigma_{t^*})}=\sup_{\Sigma_{t^*}}|\psi|.
    \end{align*}
\end{definition}
In terms of these Sobolev and $L^p$ spaces, we define the following classes of functions of time and space.
\begin{definition}[Function spaces of the solution]
\label{FnSpace}
    By $CH^k_{KAdS}$, we denote the function space
    \begin{align*}
        CH^k_{KAdS}=\bigg\{\psi:\mathcal{M}\to\mathbb{R}\text{ }\bigg|\text{ }\psi\in\bigcap_{q=0}^{k-1}C^q(\mathbb{R}_{t^*};H^{k-q}_{KAdS}(\Sigma_{t^*}))\text{ and } \frac{1}{r}\psi\in C^k(\mathbb{R}_{t^*};L^2_{KAdS}(\Sigma_{t^*}))\bigg\}
    \end{align*}
    of scalar fields $\psi$ on $\mathcal{M}$.
\end{definition}

\subsection{The energy-momentum tensor}
\label{sec:EMTensor}
We recall \cite{PosEng} the \textit{energy-momentum tensor} of a solution $\psi$ of (\ref{Massive})
\begin{align*}
    T_{\mu\nu}[\psi]&=\partial_{\mu}\psi\partial_{\nu}\psi-\frac{1}{2}g_{\mu\nu}(g^{\alpha\beta}\partial_{\alpha}\psi\partial_{\beta}\psi-\alpha\psi^2),
\end{align*}
the \textit{current} of $\psi$ along a vector field $X$
\begin{align}
\label{Current}
    J_{\mu}^X[\psi]=T_{\mu\nu}[\psi] X^{\nu},
\end{align}
and the \textit{bulk} of $\psi$ associated with $X$
\begin{align}
\label{Bulk}
    K^X[\psi]=T_{\mu\nu}[\psi]^{(X)}\pi^{\mu\nu}.
\end{align}
Here, $^{(X)}\pi_{\mu\nu}=\frac{1}{2}(\nabla_{\mu}X_{\nu}+\nabla_{\nu}X_{\mu})$ is the deformation tensor of $X$. Note that, if $X$ is globally uniformly timelike, the bulk satisfies
\begin{align*}
    |K^N[\psi]|&\leq C J_{\mu}^N[\psi]\cdot n^{\mu}_{\Sigma_{t^*}}
\end{align*}
for some $C>0$. Furthermore, recall the relation \cite{PosEng}
\begin{align}
\label{CurrentBulk}
    \nabla^{\mu}(J_{\mu}^X[\psi])=K^X[\psi].
\end{align}
It will become clear that integrating (\ref{CurrentBulk}) with $X=N$ and its commutations over an appropriate spacetime region generates the energies in Section \ref{sec:Energies}

\subsection{A class of nonlinear wave equations}
\label{sec:Eqn}
In this paper, we study the nonlinear wave equation\footnote{The restriction $\alpha<\frac{9}{4}$ is known as the \textit{Breitenlohner-Freedman} bound. This assumption ensures that (\ref{EqnOfChoice}) is well-posed, given appropriate initial and boundary conditions, in the general asymptotically-AdS setting \cite{WP}.}
\begin{align}
\label{EqnOfChoice}
    \Box_g\psi + \alpha\psi = \mathcal{F},\qquad\alpha<\frac{9}{4},
\end{align}
for $\psi$ $\in$ $CH^k_{KAdS}(\Sigma_{t^*})$ with fixed $k\geq 9$ where $\mathcal{F}$ is a smooth function of the spacetime coordinates, the solution $\psi$ and its normalised first-order derivatives $D\psi$ (\ref{UnitDerivatives}). In the following sections, we introduce appropriate boundary conditions before discussing the assumptions $\mathcal{F}$ satisfies and how this relates to the Einstein vacuum equation in harmonic gauge. We then state the scattering problem to be studied.

\subsubsection{Boundary conditions at future null infinity}
Given that the Kerr-AdS spacetime is not globally hyperbolic, equation (\ref{EqnOfChoice}) must be paired with an appropriate boundary condition along the timelike conformal boundary $\mathcal{I}$. In this paper, we restrict our attention to a particular choice in the reflective category, Dirichlet boundary conditions:\footnote{Note, the boundary condition is only really necessary in the range $\frac{5}{4}<\alpha<\frac{9}{4}$ where it excludes the Neumann branch of the solution. Otherwise, an assumption of finite energy suffices.}
\begin{align}
\label{Dirichlet}
    r^{\frac{3}{2}-s}\psi|_{\mathcal{I}} = 0,\quad s=\sqrt{\frac{9}{4}-\alpha},
\end{align}
where we write $r^{\frac{3}{2}-s}\psi|_{\mathcal{I}}$ to mean the limit
\begin{equation*}
    \lim_{r\to\infty}r^{\frac{3}{2}-s}\psi(t^*,r,\theta,\phi)
\end{equation*}
 for each fixed $(t^*,\theta,\phi)$.\par
 
The methods used herein easily extend to include Neumann or Robin boundary conditions. However, one must use the formalism of \textit{twisted} derivatives introduced by Warnick \cite{Claude}. Dirichlet, Neumann and Robin boundary conditions were treated in this manner in the linear setting by Holzegel and Warnick \cite{Twisted}.\par

It is worth noting that, since we do not commute with angular momentum operators (and therefore only show $r^{-\frac{1}{2}}$ decay of $\psi$, not $r^{-\frac{3}{2}-s}$), we do not explicitly prove here that $\psi$ attains this boundary condition. However, this is known from the well-posedness theory where angular momentum operators are utilised \cite{PosEng}.

\subsubsection{A general nonlinearity involving derivatives to first-order}
\label{sec:F}
In order to close the estimates appearing in this paper, we require a general condition on the nonlinearity $\mathcal{F}$. This is encoded in the following  estimates.
\begin{assumption}[The required estimates for $\mathcal{F}$]
    \label{NonlinAssump}
The right-hand side $\mathcal{F}$ of equation (\ref{EqnOfChoice}) must satisfy 
\begin{equation}
\label{FBound}
\begin{gathered}
    \sum_{|\sigma|\leq k-1}||N^{\sigma_1}T^{\sigma_2}\Phi^{\sigma_3}\mathcal{F}||_{L^2_{KAdS}(\Sigma_{t^*})}\leq C_k||\psi||_{H^k_{KAdS}(\Sigma_{t^*})}\bigg(||\psi||_{H^{k-1}_{KAdS}(\Sigma_{t^*})}\bigg)^{n},\\
    \sum_{|\gamma|\leq\lceil\frac{k-2}{2}\rceil}||\Box_g^{\gamma}\mathcal{F}||^2_{L^2_{KAdS}(\Sigma_{t^*})}\leq \tilde{C}_k||\psi||_{H^k_{KAdS}(\Sigma_{t^*})}\bigg(||\psi||_{H^{k-1}_{KAdS}(\Sigma_{t^*})}\bigg)^{p}
\end{gathered}
\end{equation}
for some $k\geq 9$, $C_k,\tilde{C}_k>0$ and $n,p\geq 0$.  
\end{assumption}
Condition \ref{NonlinAssump} allows for a very general class of $\mathcal{F}$. However, for simplicity, we shall restrict our attention to a nonlinearity which is quadratic in the first order normalised derivatives $D\psi$ (\ref{UnitDerivatives}) of the solution. In terms of the frame
\begin{align}
    e_0=\frac{1}{r}\partial_{t^*},\quad e_1=r\partial_r,\quad e_2=\frac{1}{r}\partial_{\theta},\quad e_3=\frac{1}{r}\partial_{\phi},\label{Frame}
\end{align}
we shall write
\begin{align}
\label{Quadratic}
    \mathcal{F}=F^{\mu\nu}(t^*,r,\theta,\phi)\cdot e_{\mu}\psi e_{\nu}\psi.
\end{align}
Here, $F:\mathcal{M}\to T\mathcal{M}\otimes T\mathcal{M}$ is a $(0,2)$-tensor with components $F_{\mu\nu}:\mathcal{M}\to\mathbb{R}$ which are smooth, bounded (with respect to the frame $e_i$) functions\footnote{This is modulo the standard degeneration of the frame $\frac{1}{r}\partial_{\theta}$, $\frac{1}{r}\partial_{\phi}$ on $S^2$.} with uniformly bounded $T$, $N$, $\Phi$, $e_1$ and $e_2$ derivatives up to order $k-1$. By Proposition \ref{prop:AbstractEnergyEst} below, the nonlinearity (\ref{Quadratic}) satisfies Condition \ref{NonlinAssump} with $n=1$, $p=3$.\par

It is worth noting that the restriction to nonlinearities of the form (\ref{Quadratic}) is not necessary to close the argument of this paper. One can readily extend to include quasilinear equations, as well as more general $\mathcal{F}$ satisfying the obvious analogue of Condition \ref{NonlinAssump}. This includes, for example, higher order polynomials in $D\psi$, as well as additional linear terms with suitable asymptotics. Furthermore, one can treat systems
\begin{align*}
    \Box_g\psi_j+\alpha_j\psi_j = \mathcal{F}_j,\quad\alpha_j<\frac{9}{4},\quad j=1,...,\ell
\end{align*}
of equations with the $\mathcal{F}_j$ satisfying the obvious analogue of Condition \ref{NonlinAssump}. Notably, one can also allow for $\mathcal{F}$ to have some additional growth in $r$ by bootstrapping stronger $r$-weighted norms of $\psi$ in the proof. This can be done (although we will not do this here) via commuting with angular momentum operators and performing an additional round of elliptic estimates.

\paragraph{The Einstein equations in harmonic gauge}

The generality of Condition \ref{NonlinAssump} will, in particular, allow for treatment of the Einstein vacuum equation in later work. To see that this is possible, in principle, let us consider the Einstein equation in harmonic gauge. In \cite{EncisoKamran}, the Einstein vacuum equation in harmonic gauge is written symbolically, in $(t,x,\theta,\phi)$ coordinates, in terms of a conformally rescaled metric $\overline{g}=x^2 g$ ($\{x=0\}$ corresponds to $\mathcal{I}$) as
\begin{align*}
    Q(g)=\overline{g}^{-1}\partial^2\overline{g}+a_0(\overline{g})\partial\overline{g}\partial\overline{g}+\frac{a_1(\overline{g})}{x}\partial\overline{g}+\frac{a_2(\overline{g})}{x^2}=0.
\end{align*}
Here, the $a_i(\overline{g})$ are smooth functions of $x$ and $\overline{g}$. Due to the blow-up of some terms in the above towards the conformal boundary, the authors perform a technical process of renormalisation to obtain a system of nonlinear wave equations which can be treated via an earlier result for scalar waves \cite{EncisoKamranWave}. The solution is written as a sum 
\begin{align*}
    g=\gamma+h, \quad h=x^{\frac{3}{2}}u
\end{align*}
of terms $\gamma$ and $h$ which contain the most (respectively least) singular behaviour at the boundary. The authors construct the approximate solution $\gamma$ of $Q(\gamma)=0$ algebraically from the boundary data, iteratively constructing $\gamma_{\ell}$ such that $Q(\gamma_{\ell})=\mathcal{O}(x^{\ell-1})$ and setting $\gamma=\gamma_{\ell}$ for $\ell$ sufficiently large. After fixing $\gamma$, they solve a tensorial equation
\begin{align}
\label{Tensorial}
    P_g u = -\frac{1}{x^{\frac{7}{2}}}Q(\gamma)+\frac{F(\overline g)u}{x}+\int_0^1 x^{\frac{3}{2}}B(u,x\partial u)\mathrm{d}\sigma
\end{align}
for $u$. Here, $F(\overline{g})$ is smooth and $B$ is quadratic, with smooth coefficients depending on $x^2\gamma+\sigma x^{\frac{7}{2}}u$. The principal part of $P_g$ is the wave operator $\Box_{\overline{g}}$ of the rescaled metric. Having fixed $\gamma=\gamma_{\ell}$ for large enough $\ell$, the highly singular weight in the first term is counteracted. This problem reduces to first ignoring the tensorial nature of equation (\ref{Tensorial}) and solving a system of scalar equations
\begin{align}
\label{EKScalarEqn}
    L_{g,\alpha_j} u_j = \mathcal{N}_j,\quad j=0,..,3.
\end{align}
Here, the left-hand side $L_{g,\alpha_j}u_j$ is related to a scalar wave equation with mass\footnote{It is worth noting that equations (\ref{EqnOfChoice}) with ``positive mass" $\alpha\leq 0$ are easier to treat than the more general equations treated herein. In particular, the corresponding energies are automatically coercive, meaning that the additional work appearing in Section \ref{sec:Coercivity} is not required if one restricts to these masses.} $\alpha_j\leq 0$ via
\begin{align}
\label{LWaveReln}
    \Box_g\psi_j+\alpha_j\psi_j=[1+\mathcal{O}(x^2)]x^{\frac{7}{2}}L_{g,\alpha_j}u_j,\quad\psi_j=x^{\frac{3}{2}}u_j
\end{align}
and the right-hand side $\mathcal{N}_j$ schematically looks like
\begin{align}
\label{EKNonlinearity}
    \mathcal{N}_j=x^{\frac{3}{2}}u_j^2+x^{\frac{5}{2}}u_j\partial u_j+x^{\frac{7}{2}}(\partial u_j)^2.
\end{align}
This gives, via relation (\ref{LWaveReln}), a system of scalar wave equations
\begin{align}
    \label{EKBox}\Box_g\psi_j+\alpha_j\psi_j=&x^2\psi_j^2+x^3\psi_j\partial\psi_j+x^4(\partial\psi_j)^2+\text{l.o.t.}
\end{align}
for $\psi_j$.\par
Local well-posedness of scalar wave equations of the form (\ref{EqnOfChoice}) with quadratic nonlinearities $\mathcal{F}$ akin to the quadratic (in $\psi_j$) term on the right-hand side of (\ref{EKBox}) was proven in an earlier paper of the authors \cite{EncisoKamranWave}. In particular, in that setting, equation (\ref{EqnOfChoice}) with
    \begin{align}
        \mathcal{F}=F^{\mu\nu}(t^*,x,\theta,\phi)\cdot \tilde{e}_{\mu}\psi \tilde{e}_{\nu}\psi,\quad F^{\mu\nu}=x^q\tilde{F}^{\mu\nu}(t^*,x,\theta,\phi),\quad q\geq 2+\sqrt{\frac{9}{4}-\alpha},\label{EKF}
    \end{align}
was considered. Here $\tilde{F}^{\mu\nu}:\mathcal{M}\to\mathbb{R}$ are smooth, bounded functions and $\tilde{e}_i$ is the frame analogous to (\ref{Frame}) in $(t^*,x,\theta,\phi)$ coordinates. Since the masses $\alpha_j$ of the system of equations (\ref{EKScalarEqn}) satisfy $\alpha_j\leq 0$, the relevant nonlinearity $\mathcal{F}$ of the form (\ref{EKF}) has weight $x^q$ with $q\geq\frac{7}{2}$, as appears in the final term on the right-hand side of (\ref{EKNonlinearity}).\par

We note that the class (\ref{EKF}) of nonlinearity considered in \cite{EncisoKamranWave} is strictly smaller than that (\ref{Quadratic}) treated herein, due to the requirement of weights which decay in $r$. Most importantly, one can show that Condition \ref{NonlinAssump} applies to all terms appearing on the right-hand side of (\ref{EKBox}).

\subsubsection{The scattering problem for nonlinear waves}
All that remains is to supplement (\ref{EqnOfChoice}), (\ref{Dirichlet}) with appropriate scattering data $h_{\mathcal{H}^+}$ on the event horizon $\mathcal{H}^+$. It will become clear in subsequent sections that this scattering construction requires that $h_{\mathcal{H^+}}$ decay at a sufficiently fast exponential rate in $t^*$. That being said, the scattering problem we will study in this paper is
\begin{align}
\label{FullProblem}
    \begin{cases}
        \Box_g\psi+\alpha\psi = \mathcal{F}\\
        \psi|_{\mathcal{H}^+}=h_{\mathcal{H}^+},\quad r^{\frac{3}{2}-s}\psi|_{\mathcal{I}}=0,\quad s=\sqrt{\frac{9}{4}-\alpha}
    \end{cases}
\end{align}
where we seek a solution $\psi$ $\in$ $CH^k_{KAdS}$ for $k\geq 9$. From Proposition \ref{prop:AbstractEnergyEst} onwards, we will restrict our attention to $\mathcal{F}$ satisfying assumption (\ref{Quadratic}).

\section{Key estimates}
\label{sec:Estimates}
Here, we collect a number of estimates which will be key in proving the main result. Firstly, in Section \ref{sec:Coercivity}, coercivity of the energies is derived. This is followed in Section \ref{sec:Elliptic} by elliptic estimates for spatial derivatives. In Section \ref{sec:Sobolev}, these elliptic estimates are applied to derive an $L^{\infty}$ Sobolev embedding on the spacelike slices $\Sigma_{t^*}$. Finally, in sections \ref{sec:EgEst} and \ref{sec:EstTheNonlin}, the previous estimates are appealed to in proving the $k$\textsuperscript{th}-order energy estimates which are central to the proof of Theorem \ref{InformalMe}. 

\subsection{Coercivity of the energy}
\label{sec:Coercivity}
In order to prove Theorem \ref{InformalMe}, one must be able to prove coercive estimates. By direct calculation, one finds that the (degenerate at $\mathcal{H}^+$) $T$-energy
\begin{align}
\label{Tenergy}
    &\int_{\Sigma_{t^*}}J^T_{\mu}[\psi]\cdot n^{\mu}_{\Sigma_{t^*}}\cdot\mathrm{vol}_{\Sigma_{t^*}}\nonumber\\
    =&\frac{1}{2}\int_{\Sigma_{t^*}}\bigg[-g^{t^*t^*}(\partial_{t^*}\psi)^2+\frac{\Delta_-}{\Sigma}(\partial_r\psi)^2+|\slashed\nabla\psi|^2+2g^{r\phi}\partial_r\psi\partial_{\phi}\psi-\frac{\alpha}{\ell^2}\psi^2\bigg]\frac{\Sigma}{\Xi}\mathrm{d}r\mathrm{d}\omega
\end{align}
carries a zeroth-order term of bad sign.\footnote{One must also deal with the $\partial_r\psi\partial_{\phi}\psi$ term, however, this can be absorbed by other terms via an application of Cauchy-Schwarz away from the horizon.} It will become clear in the proof of Lemma \ref{lem:FarRegion} that ensuring positivity of this energy requires stipulating the Breitenlohner-Freedman bound $\alpha<\frac{9}{4}$. An analogous zeroth-order term also appears in the non-degenerate $N$-energy
\begin{align}
\label{Nrgy}
     \int_{\Sigma_{t^*}}J^N_{\mu}[\psi]\cdot n^{\mu}_{\Sigma_{t^*}}\cdot\mathrm{vol}_{\Sigma_{t^*}},
\end{align}
with which we will work throughout this paper, where now the coefficient is a different function $-C_{\alpha}$ (with $C_{\alpha}\equiv\frac{\alpha}{\ell^2}$ for $r>r_1$). Trivially, for some $C>0$ one has the estimate
\begin{align*}
    &\int_{\Sigma_{\tau}}J^N_{\mu}[\psi]\cdot n^{\mu}_{\Sigma_{\tau}}\cdot\mathrm{vol}_{\Sigma_{\tau}}+C\bigg[\int_{\Sigma_{\tau}\cap\{r\geq r_1\}}\frac{\alpha}{\ell^2}\psi^2 \frac{\Sigma}{\Xi}\mathrm{d}r\mathrm{d}\omega+\int_{\Sigma_{\tau}\cap\{r<r_1\}}C_{\alpha}\psi^2\frac{\Sigma}{\Xi}\mathrm{d}r\mathrm{d}\omega\bigg]\\
    \geq&||\psi||_{H^1_{KAdS}(\Sigma_{\tau})}^2,
\end{align*}
however, we will show in Lemma \ref{lem:FarRegion} that the term on the ``far region" $r\geq r_1$ can be dropped. In Lemma \ref{lem:NearRegion}, it will be shown that the term on the ``near region" $r<r_1$ can be exchanged for a spacetime term on a compact in $r$ region.\par

We prove the following coercivity statement for the $N$-energy.

\begin{proposition}[A coercive estimate for the non-degenerate $N$-energy]
\label{prop:NCoercive}
Let $\psi$ be a solution of (\ref{FullProblem}) in $CH^1_{KAdS}$ with first time derivative in $C([\tau,\tau']; L^2(\Sigma_{t^*}))$. Then the $N$-energy (\ref{Nrgy}) of $\psi$ satisfies the coercive estimate
\begin{align}
    &\int_{\Sigma_{\tau}}J_{\mu}^N[\psi]\cdot n_{\Sigma_{t^*}}^{\mu}\cdot\mathrm{vol}_{\Sigma_{t^*}}\nonumber\\
    &+\tilde{C}_{\alpha}\bigg[\int_{\Sigma_{\tau'}\cap\{r<r_1+\delta\}}\psi^2\frac{\Sigma}{\Xi}\mathrm{d}r\mathrm{d}\omega+\int_{\mathcal{D}_{[\tau,\tau']}\cap\{r<r_1+\delta\}}\Big(\psi^2+(\partial_{t^*}\psi)^2\Big)\frac{\Sigma}{\Xi}\mathrm{d}r\mathrm{d}\omega\mathrm{d}t^*\bigg]\nonumber\\
    \geq&||\psi||^2_{H^1_{KAdS}(\Sigma_{\tau})}\label{CoerciveNEst}
\end{align}
for an appropriately chosen constant $\tilde{C}_{\alpha}$ depending on the mass $\alpha$.
\end{proposition}
Note that the implication
\begin{align*}
    ||\psi||^2_{H^1_{KAdS}}\geq C\int_{\Sigma_{\tau}}J_{\mu}^N[\psi]\cdot n_{\Sigma_{t^*}}^{\mu}\cdot\mathrm{vol}_{\Sigma_{t^*}}
\end{align*}
in the other direction is trivially true for some positive constant $C$.\par

Proposition \ref{prop:NCoercive} will be applied to the finite-in-time problems (\ref{FiniteProblem}) in Section \ref{sec:Proof}. In that setting, the integral on $\Sigma_{\tau'}$ appearing in (\ref{CoerciveNEst}) vanishes and so Proposition \ref{prop:NCoercive} shows that the $N$-current $J^N[\psi]$ controls the $H^1_{KAdS}$ norm of $\psi$ up to a spacetime term:
\begin{align*}
   \int_{\tau}^{\tau'}||\mathbf{1}_{\{r<r_1+\delta\}}\psi||^2_{H^1_{KAdS}(\Sigma_{t^*})}\mathrm{d}t^*,
\end{align*}
supported away from infinity. This term can then be dealt with via a Gr\"onwall inequality.\par

Proposition \ref{prop:NCoercive} follows from two estimates: one on the ``far" region $r\geq r_1$ where $N=T$, and one on the ``near" region. The former is a Hardy inequality, adapted from \cite[Section 5.1]{PosEng}, the proof of which demonstrates the importance of the aforementioned Breitenlohner-Freedman bound $\alpha<\frac{9}{4}$. This result is as follows.

\begin{lemma}[A ``far" region Hardy inequality]
\label{lem:FarRegion}
Let $\chi(r)$ be a smooth, positive cut-off function satisfying
    \begin{align*}
        \chi(r)=\begin{cases}
            0\text{ on }r<r_1,\\
            1\text{ on }r\geq r_1+\delta,
        \end{cases}
    \end{align*}
for some $\delta>0$, and increasing monotonically on $[r_1,r_1+\delta]$, where $r_1$ is sufficiently large with respect to $a$. Then for a solution $\psi$ of (\ref{FullProblem}) in $H^1_{KAdS}(\Sigma_{t^*})$, one has that
\begin{align*}
   \int_{\Sigma_{t^*}}\chi(r)\bigg[\frac{\Delta_-}{\Sigma}(\partial_r\psi)^2-\frac{\alpha}{\ell^2}\psi^2\bigg]\frac{\Sigma}{\Xi}\mathrm{d}r\mathrm{d}\omega\geq 0
\end{align*}
\end{lemma}

On the ``near" region, the zeroth order term of the $N$-current $J^{N}[\psi]$ is estimated via a spacetime term in order to derive the following coercivity statement.

\begin{lemma}[A ``near" region estimate for the non-degenerate $N$-energy]
\label{lem:NearRegion}
Let $\psi$ be a solution of (\ref{FullProblem}) in $CH^1_{KAdS}$ with first time derivative in $C([\tau,\tau']; L^2(\Sigma_{t^*}))$. Then
\begin{align*}
    \int_{\Sigma_{\tau}\cap\{r<r_1+\delta\}}\psi^2\frac{\Sigma}{\Xi}\mathrm{d}r\mathrm{d}\omega\leq\int_{\Sigma_{\tau'}\cap\{r<r_1+\delta\}}\psi^2\frac{\Sigma}{\Xi}\mathrm{d}r\mathrm{d}\omega+C\int_{\mathcal{D}_{[\tau,\tau']}\cap\{r<r_1+\delta\}}\bigg[\psi^2+(\partial_{t^*}\psi)^2\bigg]\frac{\Sigma}{\Xi}\mathrm{d}r\mathrm{d}\omega\mathrm{d}t^*,
\end{align*}
for some $C>0$.
\end{lemma}
We begin by proving Lemma \ref{lem:FarRegion}.

\begin{proof}[Proof of Lemma \ref{lem:FarRegion}]
We adapt the proof of coercivity of the $T$-energy from  \cite[Lemma 4.1]{PosEng} as follows.\par

Firstly, we have that
\begin{align}
\label{FirstAlpha}
    \frac{\alpha}{\ell^2\Xi}\int_{\Sigma_{t^*}}\chi(r)\psi^2\Sigma\mathrm{d}r\mathrm{d}\omega\leq\frac{\alpha}{\ell^2\Xi}\int_{\Sigma_{t^*}}\chi(r)\psi^2(r^2+a^2)\mathrm{d}r\mathrm{d}\omega.
\end{align}

Integrating the right-hand side of (\ref{FirstAlpha}) by parts and applying the Cauchy-Schwarz inequality, one obtains
\begin{align}
    &\int_{\Sigma_{t^*}}\chi(r)\psi^2(r^2+a^2)\mathrm{d}r\mathrm{d}\omega\\
    =& \int_{\Sigma_{t^*}}\frac{1}{3}\chi(r)\psi^2 \partial_r\Big(r^3-r_+^3+3a^2(r-r_+)\Big)\mathrm{d}r\mathrm{d}\omega\nonumber\\
    =&-\int_{\Sigma_{t^*}}\frac{2}{3}\Big(r^3-r_+^3+3a^2(r-r_+)\Big)\bigg[\chi(r)\psi\partial_r\psi+\frac{1}{2}\partial_r(\chi(r))\psi^2\bigg]\mathrm{d}r\mathrm{d}\omega\nonumber\\
    \leq&\bigg(\int_{\Sigma_{t^*}}\frac{4}{9}\chi(r)\psi^2 (r^2+a^2)\mathrm{d}r\mathrm{d}\omega\bigg)^{\frac{1}{2}}\bigg(\int_{\Sigma_{t^*}}(\partial_r\psi)^2\chi(r)\frac{[r^3-r_+^3+3a^2(r-r_+)]^2}{r^2+a^2}\mathrm{d}r\mathrm{d}\omega\bigg)^{\frac{1}{2}}\nonumber\\
    &-\int_{\Sigma_{t^*}}\frac{1}{3}\Big(r^3-r_+^3+3a^2(r-r_+)\Big)\partial_r(\chi(r))\psi^2\mathrm{d}r\mathrm{d}\omega\label{NegativeCutOff}\\
    \leq&\bigg(\int_{\Sigma_{t^*}}\frac{4}{9}\chi(r)\psi^2 (r^2+a^2)\mathrm{d}r\mathrm{d}\omega\bigg)^{\frac{1}{2}}\bigg(\int_{\Sigma_{t^*}}(\partial_r\psi)^2\chi(r)\frac{[r^3-r_+^3+3a^2(r-r_+)]^2}{r^2+a^2}\mathrm{d}r\mathrm{d}\omega\bigg)^{\frac{1}{2}},\nonumber
\end{align}
where the final inequality follows from the fact that the second term of (\ref{NegativeCutOff}) is negative.\par

Thus,
\begin{align*}
    \int_{\Sigma_{t^*}}\frac{9}{4}\chi(r)\psi^2(r^2+a^2)\mathrm{d}r\mathrm{d}\omega &\leq\int_{\Sigma_{t^*}}(\partial_r\psi)^2\chi(r)\frac{[r^3-r_+^3+3a^2(r-r_+)]^2}{r^2+a^2}\mathrm{d}r\mathrm{d}\omega.
\end{align*}
Now, provided that the cut-off radius $r_1$ is chosen sufficiently large (depending on $a$), we have that
\begin{align*}
    \bigg[\Delta_- -\frac{[r^3-r_+^3+3a^2(r-r_+)]^2}{(r^2+a^2)\ell^2}\bigg]\bigg|_{\chi(r)\neq 0}>0
\end{align*}
implying
\begin{align*}
    \int_{\Sigma_{t^*}}\chi(r)\bigg[\frac{\Delta_-}{\Sigma}(\partial_r\psi)^2-\frac{\alpha}{\ell^2}\psi^2\bigg]\frac{\Sigma}{\Xi}\mathrm{d}r\mathrm{d}\omega\geq 0\quad\iff\quad\alpha<\frac{9}{4},
\end{align*}
as required.
\end{proof}
Next, we prove Lemma \ref{lem:NearRegion} which covers the remaining ``near" region, $r<r_1+\delta$.
\begin{proof}[Proof of Lemma \ref{lem:NearRegion}]
    On the ``near" region $r<r_1+\delta$, we estimate the zeroth-order term of (\ref{Nrgy}) by a spacetime term via the Fundamental Theorem of Calculus:
\begin{align}
    &\int_{\Sigma_{\tau}\cap\{r<r_1+\delta\}}\psi^2\frac{\Sigma}{\Xi}\mathrm{d}r\mathrm{d}\omega\nonumber\\
    =&\int_{\Sigma_{\tau'}\cap\{r<r_1+\delta\}}\psi^2\frac{\Sigma}{\Xi}\mathrm{d}r\mathrm{d}\omega-\int_{\tau}^{\tau'}\partial_{t^*}\bigg[\int_{\Sigma_{t^*}\cap\{r<r_1+\delta\}}\psi^2\frac{\Sigma}{\Xi}\mathrm{d}r\mathrm{d}\omega\bigg]\mathrm{d}t^*\nonumber\\
    =&\int_{\Sigma_{\tau'}\cap\{r<r_1+\delta\}}\psi^2\frac{\Sigma}{\Xi}\mathrm{d}r\mathrm{d}\omega-\int_{\mathcal{D}_{[\tau,\tau']}\cap\{r<r_1+\delta\}}\partial_{t^*}(\psi^2)\frac{\Sigma}{\Xi}\mathrm{d}r\mathrm{d}\omega\mathrm{d}t^*\nonumber\\
    =&\int_{\Sigma_{\tau'}\cap\{r<r_1+\delta\}}\psi^2\frac{\Sigma}{\Xi}\mathrm{d}r\mathrm{d}\omega-\int_{\mathcal{D}_{[\tau,\tau']}\cap\{r<r_1+\delta\}}2\psi\cdot\partial_{t^*}\psi\frac{\Sigma}{\Xi}\mathrm{d}r\mathrm{d}\omega\mathrm{d}t^*\nonumber\\
    \leq&\int_{\Sigma_{\tau'}\cap\{r<r_1+\delta\}}\psi^2\frac{\Sigma}{\Xi}\mathrm{d}r\mathrm{d}\omega+\int_{\mathcal{D}_{[\tau,\tau']}\cap\{r<r_1+\delta\}}\bigg[2\psi^2+\frac{1}{2}(\partial_{t^*}\psi)^2\bigg]\frac{\Sigma}{\Xi}\mathrm{d}r\mathrm{d}\omega\mathrm{d}t^*\label{CSz}\\
    \leq&\int_{\Sigma_{\tau'}\cap\{r<r_1+\delta\}}\psi^2\frac{\Sigma}{\Xi}\mathrm{d}r\mathrm{d}\omega+C\int_{\mathcal{D}_{[\tau,\tau']}\cap\{r<r_1+\delta\}}\bigg[\psi^2+(\partial_{t^*}\psi)^2\bigg]\frac{\Sigma}{\Xi}\mathrm{d}r\mathrm{d}\omega\mathrm{d}t^*,\nonumber
\end{align}
for some $C>0$. Here, (\ref{CSz}) follows from the Cauchy-Schwarz inequality.
\end{proof}
Proposition \ref{prop:NCoercive} follows as a straight forward corollary of Lemmas \ref{lem:FarRegion} and \ref{lem:NearRegion}.
\begin{proof}[Proof of Proposition \ref{prop:NCoercive}]
    By Lemma \ref{lem:FarRegion}, the $N$-energy is coercive on the ``far" region $r\geq r_1+\delta$. Then, by adding an appropriate multiple
    \begin{align*}
        \tilde{C}_{\alpha}\bigg[\int_{\Sigma_{\tau'}\cap\{r<r_1+\delta\}}\psi^2\frac{\Sigma}{\Xi}\mathrm{d}r\mathrm{d}\omega+\int_{\mathcal{D}_{[\tau,\tau']}\cap\{r<r_1+\delta\}}\bigg[\psi^2+(\partial_{t^*}\psi)^2\bigg]\frac{\Sigma}{\Xi}\mathrm{d}r\mathrm{d}\omega\mathrm{d}t^*\bigg]
    \end{align*}
    of the right-hand side of the estimate appearing in Lemma \ref{lem:NearRegion} (choose, say, $\tilde{C}_{\alpha}=2C_{\alpha}$, where $C_{\alpha}$ is the coefficient of the zeroth-order term in (\ref{Nrgy})), one obtains the desired result.
\end{proof}
By commuting (\ref{CoerciveNEst}) with $\Gamma^{\sigma}$ and applying the elliptic estimates appearing in the next section, one obtains analogous higher order coercive estimates for $\psi$ sufficiently regular.

\subsection{$L^2$ estimates for derivatives of order $>2$ in space}
\label{sec:Elliptic}
The estimates appearing in this section will allow us to control higher-order angular and radial derivatives of the solution $\psi$. This allows one to avoid bulk terms generated by commuting with the angular momentum operators, as one ordinarily would with no issue in the spherically symmetric setting. We will prove a four-level hierarchy of estimates - the first for general derivatives on the ``far" region, the second for derivatives of the form $\Gamma^{\sigma}D_x^2\psi$ (where only two derivatives are in general directions) on the entire exterior, the third for general derivatives on the ``near" region, and the fourth for general derivatives on the entire exterior.\par

We begin with preliminary second order elliptic estimates for angular derivatives, first on the spheres $S^2_{t^*,r}$ and then on the spacelike slices $\Sigma_{t^*}$. The first is given on the Kerr-AdS spheres $S^2_{t^*,r}$ at fixed time and radial distance, but holds for more general topological $2$-spheres.

\begin{lemma}[An elliptic estimate on the topological $2$-spheres $S^2_{t^*,r}$ {\cite[Prop 4.4.3]{DHR}}]
\label{Lemma:EllipticSphere}
Let $\psi$ be a sufficiently smooth function on $S^2_{t^*,r}$ for $t^*$, $r\geq r^+$ fixed. Then
    \begin{align}
    \label{EllipticLaplace}
        \int_{S^2_{t^*,r}}(\slashed\Delta\psi)^2\mathrm{d}\omega\geq&c\int_{S^2_{t^*,r}}\bigg(|\slashed\nabla^2\psi|^2+2K|\slashed\nabla \psi|^2\bigg)\mathrm{d}\omega,
    \end{align}
for some $c>0$, where $K$ is the Gauss curvature of $S^2_{t^*,r}$.
\end{lemma}
Note that the $K$ term in (\ref{EllipticLaplace}) can be negative, however, when it comes to applying Lemma \ref{Lemma:EllipticSphere} to prove Proposition \ref{prop:SpacelikeElliptic} it will be treated as a lower-order term. As such, the sign does not matter.\par

Under assumption of sufficient regularity of $\psi$ on the entire spacelike hypersurface $\Sigma_{t^*}$, one immediately obtains the following result as a consequence of Lemma \ref{Lemma:EllipticSphere}.

\begin{corollary}[An elliptic estimate for second-order angular derivatives on spacelike slices]
\label{Cor:IntegratedEllipticSphere}
Let $\psi$ be a function in $H^2_{KAdS}(\Sigma_{t^*})$. Then
\begin{align}
\label{IntegratedEllipticSphere}
    \int_{\Sigma_{t^*}}(\slashed\Delta\psi)^2 r^2\mathrm{d}r\mathrm{d}\omega\geq&c\int_{\Sigma_{t^*}}\bigg(|\slashed\nabla^2\psi|^2+2K|\slashed\nabla \psi|^2\bigg)r^2\mathrm{d}r\mathrm{d}\omega,
\end{align}
for some $c>0$, where $K$ is the Gauss curvature of $S^2_{t^*,r}$.
\end{corollary}

Given Corollary \ref{Cor:IntegratedEllipticSphere}, we next recall some elliptic estimates on the ``far" region.

\subsubsection{Elliptic estimates for general derivatives on the ``far" region (\textbf{I})}
\label{sec:FarElliptic}
We first recall the standard (see, for example, \cite{WP}) second-order elliptic estimate one can prove on the ``far" region (outside of the ergoregion) on which $T$ is everywhere timelike. This requires appealing to Corollary \ref{Cor:IntegratedEllipticSphere} and commuting once with $T$.

\begin{proposition}[A ``far" region second-order elliptic estimate]
\label{prop:Far2ndOrder}
    Let $\psi$ be a function in $H^2_{KAdS}(\Sigma_{t^*})$ satisfying Dirichlet boundary conditions (\ref{Dirichlet}) and $\tilde{\xi}(r)$ be a smooth, positive cut-off function satisfying
    \begin{align*}
    \tilde{\xi}(r)=
        \begin{cases}
            0\text{ on }r<r_0,\\
            1\text{ on }r\geq r_0+\tilde{\delta}
        \end{cases}
    \end{align*}
    for some $\tilde{\delta}>0$ such that $r_0+\tilde{\delta}<r_1$, and increasing monotonically on $[r_0,r_0+\tilde{\delta}]$. Here, $r_0,$ $r_1$ are as in Section \ref{sec:VFs}. Then we have that
    \begin{align}
        &\int_{\Sigma_{t^*}}\tilde{\xi}(r)\bigg(|\slashed\nabla^2\psi|^2+r^2|\slashed\nabla\partial_r\psi|^2+r^4(\partial_r^2\psi)^2\bigg)r^2\mathrm{d}r\mathrm{d}\omega\nonumber\\
        \leq& C\bigg(||\psi||^2_{H^1_{KAdS}(\Sigma_{t^*})}+||T\psi||^2_{H^1_{KAdS}(\Sigma_{t^*})}+\int_{\Sigma_{t^*}}(\Box_g\psi+\alpha\psi)^2r^2\mathrm{d}r\mathrm{d}\omega\bigg),\label{TSpacelikeElliptic}
    \end{align}
    for some $C>0$.
\end{proposition}
Proposition \ref{prop:Far2ndOrder} is a result of the fact that one can write the wave equation (\ref{Massive}) as
\begin{align}
    &\frac{1}{\sqrt{|\det g|}}\partial_i(g^{ij}\sqrt{|\det g|}\partial_j\psi)\nonumber\\
    =&\Box\psi-g^{t^*t^*}\partial_{t^*}^2\psi-2g^{t^* i}\partial_{t*}\partial_i\psi-\frac{1}{\sqrt{|\det g|}}\partial_r(g^{t^*r}\sqrt{|\det g|})\partial_t\psi,\label{EllipticDecomp}
\end{align}
for $i=r,\theta,\phi$. The left-hand side of (\ref{EllipticDecomp}) is an elliptic operator on the region where $g_{\Sigma_{t^*}}$ (the induced metric on the slice $\Sigma_{t^*}$) is positive definite. This is precisely where one has
\begin{align*}
    \det g_{\Sigma_{t^*}}=\frac{\Xi^2}{\Sigma^3}\bigg[\frac{\Delta_-}{\sin^2\theta}-a^2\Delta_{\theta}\bigg]>0,
\end{align*}
the same region on which $T$ is timelike:
\begin{align*}
    g_{t^*t^*}=-\frac{\sin^2\theta}{\Sigma}\bigg[\frac{\Delta_-}{\sin^2\theta}-a^2\Delta_{\theta}\bigg]<0.
\end{align*}
Integrating (\ref{EllipticDecomp} by parts gives control of $\slashed\nabla^2\psi$, $\slashed\nabla\partial_r\psi$. The remaining $\partial_r^2\psi$ can then be controlled directly by (\ref{EllipticDecomp}), given that one controls all other terms.\par

By successively applying $\Box_g$ to (\ref{EllipticDecomp}) and repeating these steps, one obtains the following higher-order result.
\begin{proposition}[A ``far" region $k$\textsuperscript{th}-order elliptic estimate]
\label{prop:FarKthOrder}
     Let $\psi$ be a function in $H^k_{KAdS}(\Sigma_{t^*})$ for $k\geq 2$ satisfying Dirichlet boundary conditions (\ref{Dirichlet}) and $\tilde\xi(r)$ be as in Proposition \ref{prop:Far2ndOrder}. Then we have that
     \begin{align*}
         &\sum_{|\sigma|=k}\int_{\Sigma_{t^*}}\tilde{\xi}(r)r^{2\sigma_2}|\slashed\nabla^{\sigma_1}\partial_{t^*}^{\sigma_2}\partial_r^{\sigma_3}\psi|^2 r^{2}\mathrm{d}r\mathrm{d}\omega\nonumber\\
        \leq &C\sum_{\substack{|\alpha|\leq k-1,\\ |\beta|\leq k- 2,\\|\gamma|\leq\lceil\frac{k-2}{2}\rceil}}\bigg[||T^{\alpha}\psi||^2_{H^1_{KAdS}(\Sigma_{t^*})}+\int_{\Sigma_{t^*}}\bigg(\Big(T^{\beta}(\Box_g\psi+\alpha\psi)\Big)^2+\Big(\Box_g^{\gamma}(\Box_g\psi+\alpha\psi)\Big)^2\bigg) r^2\mathrm{d}r\mathrm{d}\omega\bigg],
     \end{align*}
     for some $C>0$.
\end{proposition}
We can now proceed to estimate derivatives $\Gamma^{\sigma}D_x^2\psi$ on the entire exterior, which will subsequently allow us to construct estimates for general derivatives on the ``near"  region. These will eventually be combined with Proposition \ref{prop:FarKthOrder} to control general derivatives on the entire exterior region.

\subsubsection{Estimating spatial derivatives of at most second order in general directions on the entire exterior region (\textbf{II})}
\label{sec:Level1}
We first derive an elliptic estimate for second order spatial derivatives on the spacelike slices $\Sigma_{t^*}$. This requires appealing to Corollary \ref{Cor:IntegratedEllipticSphere}, commuting once with $T$ and commuting once with the redshift vector field $N$ in order to counteract the degenerate (at the horizon) $g^{rr}$ weights appearing in the estimate involving $T$ alone. One commutation with $\Phi$ is also employed, in order to deal treat a mixed angular-radial term of potentially bad sign. This result will then be generalised to a statement for higher order spatial derivatives of the form $\Gamma^{\sigma}D_x^2\psi$ via successive commutation with the commuting vector fields $\Gamma$.

\begin{proposition}[A non-degenerate second-order elliptic estimate on spacelike slices]
\label{prop:SpacelikeElliptic}
Let $\psi$ be a function in $H^2_{KAdS}(\Sigma_{t^*})$ satisfying Dirichlet boundary conditions (\ref{Dirichlet}). Then we have that
    \begin{align}
    \label{SpacelikeElliptic}
        &\int_{\Sigma_{t^*}}\bigg(|\slashed\nabla^2\psi|^2+r^2|\slashed\nabla\partial_r\psi|^2+r^4(\partial_r^2\psi)^2\bigg)r^{2}\mathrm{d}r\mathrm{d}\omega\nonumber\\
        \leq&C\bigg(||\psi||^2_{H^1_{KAdS}(\Sigma_{t^*})}+||T\psi||^2_{H^1_{KAdS}(\Sigma_{t^*})}+||N\psi||^2_{H^1_{KAdS}(\Sigma_{t^*})}+||\Phi\psi||^2_{H^1_{KAdS}(\Sigma_{t^*})}\nonumber\\
        &\qquad+\int_{\Sigma_{t^*}}(\Box_g\psi+\alpha\psi)^2 r^{2}\mathrm{d}r\mathrm{d}\omega\bigg)
    \end{align}
    for some $C>0$.
\end{proposition}

\begin{proof}
     Let us consider the massive wave operator
\begin{align*}
    \Box_g\psi +\alpha\psi
\end{align*}
    appearing on the left-hand side of (\ref{EqnOfChoice}) in regular $(t^*,r,\theta,\phi)$-coordinates:
    \begin{align}
        \Box_g\psi+\alpha\psi=& g^{t^*t^*}\partial_{t^*}^2\psi+g^{rr}\partial_r^2\psi+2g^{t^*\phi}\partial_{t^*}\partial_{\phi}\psi+2g^{t^*r}\partial_{t^*}\partial_r\psi\nonumber\\
        &+\partial_r(g^{rr})\partial_r\psi+\frac{g^{rr}}{\sqrt{|\det g|}}\partial_r(\sqrt{|\det g|})\partial_r\psi+\frac{g^{\theta\theta}}{\sqrt{|\det g|}}\partial_{\theta}(\sqrt{|\det g|})\partial_{\theta}\psi\nonumber\\
        &+\partial_{\theta}(g^{\theta\theta})\partial_{\theta}\psi
        +2g^{\phi r}\partial_{\phi}\partial_r\psi+\partial_r(g^{t^*r})\partial_{t^*}\psi+\frac{g^{t^*r}}{\sqrt{|\det g|}}\partial_r(\sqrt{|\det g|})\partial_{t^*}\psi\nonumber\\
        &+\frac{g^{r\phi}}{\sqrt{|\det g|}}\partial_r(\sqrt{|\det g|})\partial_{\phi}\psi+\partial_r(g^{r\phi})\partial_{\phi}\psi+\slashed\Delta\psi+\alpha\psi.\label{ExpandBox}
    \end{align}
By rewriting the $\partial_{t^*}^2\psi$ term in (\ref{ExpandBox}) in terms of $N^2\psi$, rearranging, multiplying by $\slashed\Delta\psi$ and integrating over a spacelike slice $\Sigma_{t^*}$, one obtains
    \begin{align}
    \label{MultElliptic}
        \bigg|\int_{\Sigma_{t^*}}\bigg[&\frac{1}{r^2}N^2\psi+\bigg(\frac{1}{r^2}+\xi\bigg)\partial_{t^*}\partial_{\phi}\psi
        +\bigg(\frac{1}{r^3}+\xi\bigg)\partial_{t^*}\partial_r\psi+\bigg(\frac{1}{r^4}+\xi\bigg)\partial_{t^*}\psi+(r+\xi)\partial_r\psi\nonumber\\
        &+\frac{1}{r^2}\partial_{\theta}\psi+\bigg(\frac{1}{r^3}+\xi\bigg)\partial_{\phi}\psi+\Box_g\psi+\alpha\psi\bigg]\slashed\Delta\psi\cdot r^{2}\mathrm{d}r\mathrm{d}\omega\bigg|\nonumber\\
        \geq C\bigg|\int_{\Sigma_{t^*}}\bigg[&\Big(g^{rr}+\xi^2(g^{t^* r})^2\Big)\partial_r^2\psi+\xi^2(g^{t^*\phi})^2\partial_{\phi}^2\psi+2\Big(g^{\phi r}+\xi^2 g^{t^*r}g^{t^*\phi}\Big)\partial_{\phi}\partial_r\psi\nonumber\\
        &+\alpha\psi+\slashed\Delta\psi\bigg]\slashed\Delta\psi\cdot r^{2}\mathrm{d}r\mathrm{d}\omega\bigg|,
    \end{align}
    for some $C>0$. Note that by using the redshift vector field, the degenerately weighted ($g^{rr}|_{r^+}=0$) second-order $r$-derivative has been supplemented with a non-degenerate near-horizon term.\par
   
We estimate all terms on the left-hand side of (\ref{MultElliptic}), beginning with the first:
\begin{align*}
    \bigg|\int_{\Sigma_{t^*}}\frac{1}{r^2}N^2\psi\cdot\slashed\Delta\psi \cdot r^{2}\mathrm{d}r\mathrm{d}\omega\bigg|
    \leq&\bigg(\frac{1}{\varepsilon}\int_{\Sigma_{t^*}}(N^2\psi)^2\mathrm{d}r\mathrm{d}\omega\bigg)^{\frac{1}{2}}\bigg(\varepsilon\int_{\Sigma_{t^*}}(\slashed\Delta\psi)^2\mathrm{d}r\mathrm{d}\omega\bigg)^{\frac{1}{2}}\\
    \leq&\frac{1}{2\varepsilon}\int_{\Sigma_{t^*}}(N^2\psi)^2 \mathrm{d}r\mathrm{d}\omega+\frac{\varepsilon}{2}\int_{\Sigma_{t^*}}(\slashed\Delta\psi)^2\mathrm{d}r\mathrm{d}\omega\\
    \leq&\frac{1}{2\varepsilon}||N\psi||^2_{H^1_{KAdS}(\Sigma_{t^*})}+\frac{\varepsilon}{2}\int_{\Sigma_{t^*}}(\slashed\Delta\psi)^2\mathrm{d}r\mathrm{d}\omega.
\end{align*}
Here, $\varepsilon>0$ is chosen sufficiently small to allow for the $\varepsilon$ weighted second term to be absorbed by the $(\slashed\Delta\psi)^2$ term on the right-hand side of (\ref{MultElliptic}).\par

The terms on the left-hand side of (\ref{MultElliptic}) involving second-order $t^*$-derivatives and first-order derivatives respectively are dealt with similarly, yielding 
\begin{align*}
    &\bigg|\int_{\Sigma_{t^*}}\bigg[\bigg(\frac{1}{r^2}+\xi\bigg)\partial_{t^*}\partial_{\phi}\psi\cdot\slashed\Delta\psi+\bigg(\frac{1}{r^3}+\xi\bigg)\partial_{t^*}\partial_r\psi\cdot\slashed\Delta\psi\bigg]r^{2}\mathrm{d}r\mathrm{d}\omega\bigg|\\
    \leq&\frac{1}{2\varepsilon}\int_{\Sigma_{t^*}}\bigg[|\slashed\nabla\partial_{t^*}\psi|^2 +r^2(\partial_r\partial_{t^*}\psi)^2\bigg]r^2\mathrm{d}r\mathrm{d}\omega+\varepsilon\int_{\Sigma_{t^*}}(\slashed\Delta\psi)^2 r^{2}\mathrm{d}r\mathrm{d}\omega\\
    \leq&\frac{1}{2\varepsilon}||T\psi||^2_{H^1_{KAdS}(\Sigma_{t^*})}+\varepsilon\int_{\Sigma_{t^*}}(\slashed\Delta\psi)^2 r^{2}\mathrm{d}r\mathrm{d}\omega
\end{align*}
and
\begin{align*}
    &\bigg|\int_{\Sigma_{t^*}}\bigg[\bigg(\frac{1}{r^4}+\xi\bigg)\partial_{t^*}\psi\cdot\slashed\Delta\psi+(r+\xi)\partial_r\psi\cdot\slashed\Delta\psi+\frac{1}{r^2}\partial_{\theta}\psi\cdot\slashed\Delta\psi+\bigg(\frac{1}{r^3}+\xi\bigg)\partial_{\phi}\psi\cdot\slashed\Delta\psi\bigg]r^{2}\mathrm{d}r\mathrm{d}\omega\bigg|\\
    \leq&\frac{1}{2\varepsilon}\int_{\Sigma_{t^*}}\bigg[\frac{1}{r^2}(\partial_{t^*}\psi)^2+r^2(\partial_r\psi)^2+|\slashed\nabla\psi|^2\bigg]r^2\mathrm{d}r\mathrm{d}\omega+2\varepsilon\int_{\Sigma_{t^*}}(\slashed\Delta\psi)^2 r^{2}\mathrm{d}r\mathrm{d}\omega\\
    \leq&\frac{1}{2\varepsilon}||\psi||^2_{H^1_{KAdS}(\Sigma_{t^*})}+2\varepsilon\int_{\Sigma_{t^*}}(\slashed\Delta\psi)^2 r^{2}\mathrm{d}r\mathrm{d}\omega.
\end{align*}
Analogously,
\begin{align*}
    &\bigg|\int_{\Sigma_{t^*}}\bigg[(\Box_g\psi+\alpha\psi)\slashed\Delta\psi\cdot r^2\mathrm{d}r\mathrm{d}\omega\bigg|
    \leq\frac{1}{2\varepsilon}\int_{\Sigma_{t^*}}(\Box_g\psi+\alpha\psi)^2 r^2\mathrm{d}r\mathrm{d}\omega+\frac{\varepsilon}{2}\int_{\Sigma_{t^*}}(\slashed\Delta\psi)^2 r^{2}\mathrm{d}r\mathrm{d}\omega.
\end{align*}
Next, we determine what the right-hand side of (\ref{MultElliptic}) controls. Integrating the $r$-derivative terms by parts and dropping boundary terms at infinity as a result of the Dirichlet boundary conditions (\ref{Dirichlet}) gives
\begin{align}
    &\bigg|\int_{\Sigma_{t^*}}\Big(g^{rr}+\xi^2(g^{t^* r})^2\Big)\partial_r^2\psi\cdot\slashed\Delta\psi\cdot r^{2}\mathrm{d}r\mathrm{d}\omega\bigg|\nonumber\\
     =&\bigg|\int_{\Sigma_{t^*}}\Big(g^{rr}+\xi^2(g^{t^* r})^2\Big)\bigg[\partial_r(\partial_r\psi\cdot\slashed\Delta\psi)-\partial_r\psi\cdot\partial_r\slashed\Delta\psi\bigg]r^{2}\mathrm{d}r\mathrm{d}\omega\bigg|\nonumber\\
     =&\bigg|\int_{\Sigma_{t^*}}\bigg[-\partial_r\bigg(\Big(g^{rr}+\xi^2(g^{t^* r})^2\Big)r^2\bigg)\partial_r\psi\cdot\slashed\Delta\psi\nonumber\\
     &\qquad\quad-\Big(g^{rr}+\xi^2(g^{t^* r})^2\Big)\bigg(\partial_r\psi\cdot\slashed\Delta\partial_r\psi+\partial_r(g^{\theta\theta})\partial_r\psi\cdot\partial_{\theta}^2\psi+\partial_r(g^{\phi\phi})\partial_r\psi\cdot\partial_{\phi}^2\psi\bigg)r^{2}\bigg]\mathrm{d}r\mathrm{d}\omega\nonumber\\
     &+\int_{S^2_{t^*,r_+}}\xi^2(g^{t^* r})^2\partial_r\psi\cdot\slashed\Delta\psi\cdot r^2\mathrm{d}\omega\bigg|\nonumber\\
     \geq&C\bigg|\int_{\Sigma_{t^*}}\bigg[r^4|\slashed\nabla\partial_r\psi|^2+\bigg|\slashed\nabla\bigg(\Big(g^{rr}+\xi^2(g^{t^* r})^2\Big)r^2\bigg)\partial_r\psi\cdot\slashed\nabla\partial_r\psi\bigg|-r^3|\partial_r\psi\cdot\slashed\Delta\psi|-\frac{1}{2\varepsilon}|\slashed\nabla\psi|^2r^2\bigg]\mathrm{d}r\mathrm{d}\omega\bigg|\label{DropBdyTerm}\\
     \geq&C\bigg(\int_{\Sigma_{t^*}}\bigg[r^4|\slashed\nabla\partial_r\psi|^2-\frac{\varepsilon}{2}r^2(\slashed\Delta\psi)^2\bigg]\mathrm{d}r\mathrm{d}\omega-\frac{1}{2\varepsilon}||\psi||^2_{H^1_{KAdS}(\Sigma_{t^*})}\bigg)\label{CSTime}
\end{align}
for some $C>0$. Here, as before, $\varepsilon>0$ is chosen sufficiently small. Inequality (\ref{DropBdyTerm}) follows from applying the asymptotic behaviour of the inverse metric components (\ref{InverseAsymptotics}) and absorbing the small contribution of the boundary term on $r_+$ by other terms. Line (\ref{CSTime}) follows from the Cauchy-Schwarz inequality. The second term on the right-hand side of (\ref{CSTime}) can be absorbed on the left-hand side of (\ref{MultElliptic}), whilst the final term can be absorbed by the $(\slashed\Delta\psi)^2$ term on the right-hand side of (\ref{MultElliptic}).\par

The mass term appearing on the right-hand side of (\ref{MultElliptic}) satisfies
\begin{align*}
    \bigg|\int_{\Sigma_{t^*}}\alpha\psi\cdot\slashed\Delta\psi\cdot r^{2}\mathrm{d}r\mathrm{d}\omega\bigg|
    \leq&\frac{1}{2\varepsilon}\int_{\Sigma_{t^*}}(\alpha\psi)^2 r^{2}\mathrm{d}r\mathrm{d}\omega+\frac{\varepsilon}{2}\int_{\Sigma_{t^*}}(\slashed\Delta\psi)^2 r^2\mathrm{d}r\mathrm{d}\omega\\
    \leq& \frac{\alpha^2}{2\varepsilon}||\psi||^2_{H^1_{KAdS}(\Sigma_{t^*})}+\frac{\varepsilon}{2}\int_{\Sigma_{t^*}}(\slashed\Delta\psi)^2 r^2\mathrm{d}r\mathrm{d}\omega,
\end{align*}
and so can be absorbed by other terms.\par

This leaves one terms of potentially bad sign on the right-hand side of (\ref{MultElliptic}). This satisfies
\begin{align}
    &\bigg|\int_{\Sigma_{t^*}}2\Big(g^{\phi r}+\xi^2 g^{t^*r}g^{t^*\phi}\Big)\partial_{\phi}\partial_r\psi\cdot\slashed\Delta\psi\cdot r^{2}\mathrm{d}r\mathrm{d}\omega\bigg|\nonumber\\
    \leq&C\bigg[\int_{\Sigma_{t^*}}\frac{\varepsilon}{2}(\slashed\Delta\psi)^2 r^2\mathrm{d}r\mathrm{d}\omega +\frac{1}{2\varepsilon}||\Phi\psi||^2_{H^1_{KAdS}(\Sigma_{t^*})}\bigg],\label{PhiThing}
\end{align}
for some $C>0$. The first term on the right-hand side of (\ref{PhiThing}) can be absorbed by the $(\slashed\Delta\psi)^2$ term on the right-hand side of (\ref{MultElliptic}), whilst the second can be moved to the left-hand side of (\ref{MultElliptic}).\par

For the Laplacian term on the right-hand side of (\ref{MultElliptic}), we appeal to Corollary \ref{Cor:IntegratedEllipticSphere}. The Gauss curvature term of bad sign on the right-hand side of (\ref{IntegratedEllipticSphere}) is of first-order, so easily controlled via an application of the Cauchy-Schwarz inequality by $||\psi||_{H^1_{KAdS(\Sigma_{t^*})}}$ and absorbed by the left-hand side of (\ref{MultElliptic}). The remaining $\partial_{\phi}^2\psi$ term on the right-hand side of (\ref{MultElliptic}) is of good sign. Finally, the second-order $r$-derivative of $\psi$ is controlled on the near region by $||N\psi||_{H^1_{KAdS}(\Sigma_{t^*})}$ and on the far region via the equation, yielding estimate (\ref{SpacelikeElliptic}).\par
\end{proof}

In order to generalise the estimate of Proposition \ref{prop:SpacelikeElliptic} to higher order derivatives $\Gamma^{\sigma}D_x^2\psi$, one must commute with $N$ sufficiently many times and apply $\slashed\Delta^m$ as a multiplier for $m>0$ sufficiently large. By doing so, one recovers the following result.

\begin{proposition}[A non-degenerate $k$\textsuperscript{th}-order elliptic estimate for derivatives where only $2$ of $k$ are arbitrary]
\label{prop:2ofk}
 Let $\psi$ be a function in $H^k_{KAdS}(\Sigma_{t^*})$ for $k\geq 2$ satisfying Dirichlet boundary conditions (\ref{Dirichlet}). Then,
    \begin{align}
        &\int_{\Sigma_{t^*}}\sum_{|\sigma|= k-2}\bigg[|\slashed\nabla^2\Gamma^{\sigma}\psi|^2+r^2|\slashed\nabla\Gamma^{\sigma}\partial_r\psi|^2+r^4(\Gamma^{\sigma}\partial_r^2\psi)^2\bigg]r^2\mathrm{d}r\mathrm{d}\omega\nonumber\\
        \leq&C\sum_{\substack{|\alpha|\leq k-1,\\ |\beta|\leq k-2}}\bigg[||\Gamma^{\alpha}\psi||_{H^1_{KAdS}(\Sigma_{t^*})}^2+\int_{\Sigma_{t^*}}\Big(\Gamma^{\beta}(\Box_g\psi+\alpha\psi)\Big)^2 r^2\mathrm{d}r\mathrm{d}\omega\bigg]\label{EstOf2ofk}
    \end{align}
    for some $C>0$.
\end{proposition}

\begin{proof}
In order to control third order derivatives of the form $\Gamma D_x^2\psi$, we successively commute (\ref{ExpandBox}) with $T$, $N$ and $\Phi$. The process of proving the second order estimate can then be repeated, applying higher-order multipliers $\slashed\Delta^{\sigma}$ as required to successfully integrate by parts. Commuting with $T$ and $\Phi$ generates no additional terms, so the $TD_x^2\psi$, $\Phi D_x^2\psi$ estimates are trivial. Commutator terms are generated in the $ND_x^2\psi$ estimate, however, these are of at most second order (in general derivatives) and vanishing near infinity so controlled by the estimate of Proposition \ref{prop:SpacelikeElliptic}:
   \begin{align}
       N\Box_g\psi =& \Box_g N\psi + N\bigg(\frac{1}{\sqrt{|\det g|}}\bigg)\partial_{\mu}(g^{\mu\nu}\sqrt{|\det g|}\partial_{\nu}\psi)+\frac{1}{\sqrt{|\det g|}}N\Big(\partial_{\mu}(g^{\mu\nu}\sqrt{|\det g|})\Big)\partial_{\nu}\psi\nonumber\\
       &+\frac{1}{\sqrt{|\det g|}} N(g^{\mu\nu}\sqrt{|\det g|})\partial_{\mu}\partial_{\nu}\psi.\label{NCommutator}
   \end{align}
Further commutations with $N$ generate additional commutator terms which are also of at most second order in general derivatives and of lower order in $N$. These can always be controlled by estimates with fewer commutations.
\end{proof}
Given Propsition \ref{prop:2ofk}, we can next estimate general derivtives on the ``near" region.

\subsubsection{Estimating general derivatives on the ``near" region (\textbf{III})}
\label{sec:GenNear}
We would like to prove the following estimate, which gives control of general $k$\textsuperscript{th}-order derivatives of $\psi$ on the ``near" region. This, combined with the analogous ``far" region result (Proposition \ref{prop:FarKthOrder}) will allow us to control general $k$\textsuperscript{th}-order derivatives of $\psi$ on the entire exterior region, non-degenerately.
\begin{theorem}[A non-degenerate $k$\textsuperscript{th}-order ``near" region estimate]
\label{thm:kNear}
     Let $\psi$ be a function in $H^k_{KAdS}(\Sigma_{t^*})$ for $k\geq 2$ satisfying Dirichlet boundary conditions (\ref{Dirichlet}) and let $\xi(r)$ be as in the definition of $N$ (\ref{RedshiftVF}). Then,
    \begin{align}
        &\sum_{\substack{|\sigma|=k,\\\sigma_2\leq k-1}}\int_{\Sigma_{t^*}}\xi(r)r^{2\sigma_2}|\slashed\nabla^{\sigma_1}\partial_{t^*}^{\sigma_2}\partial_r^{\sigma_3}\psi|^2 r^{2}\mathrm{d}r\mathrm{d}\omega\nonumber\\
        \leq &C\sum_{\substack{|\alpha|\leq k-1,\\ |\beta|\leq k- 2,\\|\gamma|\leq\big\lceil\frac{k-2}{2}\big\rceil}}\bigg[||\Gamma^{\alpha}\psi||^2_{H^1_{KAdS}(\Sigma_{t^*})}+\int_{\Sigma_{t^*}}\bigg(\Big(\Gamma^{\beta}(\Box_g\psi+\alpha\psi)\Big)^2+\Big(\slashed\Delta^{\gamma}(\Box_g\psi+\alpha\psi)\Big)^2\bigg) r^2\mathrm{d}r\mathrm{d}\omega\bigg]\nonumber
    \end{align}
    for some $C>0$.
\end{theorem}
In order to prove Theorem \ref{thm:kNear}, we require the following result which is immediately implied by Proposition \ref{prop:2ofk}.
\begin{proposition}[A non-degenerate $k$\textsuperscript{th}-order ``near" region estimate for derivatives where at most 2 of $k$ are $\partial_{\theta}$]
    \label{prop:2theta}
     Let $\psi$ be a function in $H^k_{KAdS}(\Sigma_{t^*})$ for $k\geq 2$ satisfying Dirichlet boundary conditions (\ref{Dirichlet}) and let $\xi(r)$ be as in the definition of $N$ (\ref{RedshiftVF}). Then,
    \begin{align}
        &\int_{\Sigma_{t^*}}\sum_{\substack{|\sigma|= k-2,\\\sigma_4\leq 2}}\xi(r)(\partial_{t^*}^{\sigma_1}\partial_r^{\sigma_2}\partial_{\phi}^{\sigma_3}\partial_{\theta}^{\sigma_4}\psi)^2r^2\mathrm{d}r\mathrm{d}\omega\nonumber\\
        \leq&C\sum_{\substack{|\alpha|\leq k-1,\\ |\beta|\leq k-2}}\bigg[||\Gamma^{\alpha}\psi||_{H^1_{KAdS}(\Sigma_{t^*})}^2+\int_{\Sigma_{t^*}}\Big(\Gamma^{\beta}(\Box_g\psi+\alpha\psi)\Big)^2 r^2\mathrm{d}r\mathrm{d}\omega\bigg]
    \end{align}
    for some $C>0$.
\end{proposition}
We can now prove Theorem \ref{thm:kNear}. We demonstrate how to obtain control of the third and fourth-order $\partial_{\theta}$ derivatives. The general higher-order statement follows by iterating the same steps sufficiently many times and commuting with $\Gamma$.

\begin{proof}
Rearranging (\ref{ExpandBox}) so that only terms involving $\partial_{\theta}$ derivatives of $\psi$ appear on the right-hand side gives
    \begin{align}
        &\Box_g\psi+\alpha\psi-g^{t^*t^*}\partial_{t^*}^2\psi-g^{rr}\partial_r^2\psi-2g^{t^*\phi}\partial_{t^*}\partial_{\phi}\psi-2g^{t^*r}\partial_{t^*}\partial_r\psi-\partial_r(g^{rr})\partial_r\psi\nonumber\\
        &-\frac{g^{rr}}{\sqrt{|\det g|}}\partial_r(\sqrt{|\det g|})\partial_r\psi
        -2g^{\phi r}\partial_{\phi}\partial_r\psi-\partial_r(g^{t^*r})\partial_{t^*}\psi-\frac{g^{t^*r}}{\sqrt{|\det g|}}\partial_r(\sqrt{|\det g|})\partial_{t^*}\psi\nonumber\\
        &-\frac{g^{r\phi}}{\sqrt{|\det g|}}\partial_r(\sqrt{|\det g|})\partial_{\phi}\psi-\partial_r(g^{r\phi})\partial_{\phi}\psi-\alpha\psi\nonumber\\
        =&\frac{g^{\theta\theta}}{\sqrt{|\det g|}}\partial_{\theta}(\sqrt{|\det g|})\partial_{\theta}\psi+\partial_{\theta}(g^{\theta\theta})\partial_{\theta}\psi+\slashed\Delta\psi.\label{OnlyTheta}
    \end{align}
Applying the $\slashed\Delta$ operator to both sides of expression (\ref{OnlyTheta}), multiplying by $\xi(r)\slashed\Delta^p\psi$ for $p=1,2$ and integrating over $\Sigma_{t^*}$ then yields
\begin{align}
    &\bigg|\int_{\Sigma_{t^*}}\sum_{\substack{|\sigma|\leq 4,\\\sigma_4\leq 2}}\bigg[\slashed\Delta(\Box_g\psi+\alpha\psi)+\partial_{t^*}^{\sigma_1}\partial_r^{\sigma_2}\partial_{\phi}^{\sigma_3}\partial_{\theta}^{\sigma_4}\psi\bigg]\xi(r)\slashed\Delta^p\psi\cdot r^2\mathrm{d}r\mathrm{d}\omega\bigg|\nonumber\\
    \geq& C\bigg|\int_{\Sigma_{t^*}}\bigg[\frac{g^{\theta\theta}}{\sqrt{|\det g|}}\partial_{\theta}(\sqrt{|\det g|})\slashed\Delta\partial_{\theta}\psi+\partial_{\theta}(g^{\theta\theta})\slashed\Delta\partial_{\theta}\psi+\slashed\Delta^2\psi\bigg]\xi(r)\slashed\Delta^p\psi\cdot r^2\mathrm{d}r\mathrm{d}\omega\bigg|,\label{UsePrevNear}
\end{align}
for some $C>0$.\par 

Set $p=1$. By Cauchy-Schwarz and Proposition \ref{prop:2theta}, this implies
\begin{align}
    &\sum_{\substack{|\alpha|\leq 3,\\ |\beta|\leq 2}}\bigg[||\Gamma^{\alpha}\psi||_{H^1_{KAdS}(\Sigma_{t^*})}^2+\int_{\Sigma_{t^*}}\bigg(\Big(\Gamma^{\beta}(\Box_g\psi+\alpha\psi)\Big)^2+\Big(\slashed\Delta(\Box_g\psi+\alpha\psi)\Big)^2\bigg) r^2\mathrm{d}r\mathrm{d}\omega\bigg]\nonumber\\
    \geq& C\int_{\Sigma_{t^*}}\bigg[\frac{g^{\theta\theta}}{\sqrt{|\det g|}}\partial_{\theta}(\sqrt{|\det g|})\slashed\Delta\partial_{\theta}\psi+\partial_{\theta}(g^{\theta\theta})\slashed\Delta\partial_{\theta}\psi+\slashed\Delta^2\psi\bigg]\xi(r)\slashed\Delta\psi\cdot r^2\mathrm{d}r\mathrm{d}\omega,\label{IntTheta}
\end{align}
for some $C>0$. Control of third-order angular derivatives (in particular, the third-order $\partial_{\theta}$ derivative) arises from integrating the final term on the right-hand side of (\ref{IntTheta}) by parts once in the angular directions:
\begin{align}
    \bigg|\int_{\Sigma_{t^*}}\xi(r)\cdot\slashed\Delta^2\psi\cdot\slashed\Delta\psi\cdot r^2\mathrm{d}r\mathrm{d}\omega\bigg|\geq\int_{\Sigma_{t^*}}\xi(r)|\slashed\nabla^3\psi|^2r^2\mathrm{d}r\mathrm{d}\omega.\label{3Theta}
\end{align}
The remaining terms of potentially bad sign are controlled via Cauchy-Schwarz,
\begin{align}
    &\bigg|\int_{\Sigma_{t^*}}\bigg[\frac{g^{\theta\theta}}{\sqrt{|\det g|}}\partial_{\theta}(\sqrt{|\det g|})\slashed\Delta\partial_{\theta}\psi\cdot\slashed\Delta\psi+\partial_{\theta}(g^{\theta\theta})\slashed\Delta\partial_{\theta}\psi\cdot\slashed\Delta\psi\bigg]r^2\mathrm{d}r\mathrm{d}\omega\bigg|\nonumber\\
    &\leq C\int_{\Sigma_{t^*}}\bigg[\varepsilon|\slashed\nabla^3\psi|^2+\frac{1}{\varepsilon}(\slashed\Delta\psi)^2\bigg]r^2\mathrm{d}r\mathrm{d}\omega,\label{ThetaAbsorb}
\end{align}
for some $C>0$, where $\varepsilon>0$ is chosen to be sufficiently small. The second term on the right-hand side of (\ref{ThetaAbsorb}) can be absorbed on the left-hand side of (\ref{IntTheta}).\par

Combining (\ref{3Theta}), (\ref{ThetaAbsorb}) and (\ref{IntTheta}) (with $p=1$) gives
\begin{align}
  &\sum_{\substack{|\alpha|\leq 3,\\ |\beta|\leq 2}}\bigg[||\Gamma^{\alpha}\psi||_{H^1_{KAdS}(\Sigma_{t^*})}^2+\int_{\Sigma_{t^*}}\bigg(\Big(\Gamma^{\beta}(\Box_g\psi+\alpha\psi)\Big)^2+\Big(\slashed\Delta(\Box_g\psi+\alpha\psi)\Big)^2\bigg) r^2\mathrm{d}r\mathrm{d}\omega\bigg]\nonumber\\
    \geq& C\int_{\Sigma_{t^*}}\xi(r)|\slashed\nabla^3\psi|^2 r^2\mathrm{d}r\mathrm{d}\omega,\label{3rdOrderNear} 
\end{align}
for some $C>0$.\par

Now, set $p=2$. Applying Cauchy-Schwarz with a sufficiently small $\varepsilon>0$ weight to the left-hand side of (\ref{UsePrevNear}) and applying Proposition \ref{prop:2theta} gives
\begin{align}
    &\sum_{\substack{|\alpha|\leq 3,\\ |\beta|\leq 2}}\bigg[||\Gamma^{\alpha}\psi||_{H^1_{KAdS}(\Sigma_{t^*})}^2+\int_{\Sigma_{t^*}}\bigg(\Big(\Gamma^{\beta}(\Box_g\psi+\alpha\psi)\Big)^2+\Big(\slashed\Delta(\Box_g\psi+\alpha\psi)\Big)^2\bigg) r^2\mathrm{d}r\mathrm{d}\omega\bigg]\nonumber\\
    \geq& C\bigg|\int_{\Sigma_{t^*}}\bigg[\frac{g^{\theta\theta}}{\sqrt{|\det g|}}\partial_{\theta}(\sqrt{|\det g|})\slashed\Delta\partial_{\theta}\psi+\partial_{\theta}(g^{\theta\theta})\slashed\Delta\partial_{\theta}\psi+\slashed\Delta^2\psi\nonumber\\
    &\qquad\qquad-\varepsilon(\slashed\Delta^2\psi)^2\bigg]\xi(r)\slashed\Delta^2\psi\cdot r^2\mathrm{d}r\mathrm{d}\omega\bigg|\label{LastTheta}
\end{align}
The first two terms on the right-hand side of (\ref{LastTheta}) are controlled by (\ref{3rdOrderNear}). The third term controls $\slashed\nabla^4\psi$ via an application of Corollary \ref{Cor:IntegratedEllipticSphere}, after absorbing the Gauss curvature term of bad sign as before. This gives 
\begin{align*}
    &\sum_{\substack{|\alpha|\leq 3,\\ |\beta|\leq 2}}\bigg[||\Gamma^{\alpha}\psi||_{H^1_{KAdS}(\Sigma_{t^*})}^2+\int_{\Sigma_{t^*}}\bigg(\Big(\Gamma^{\beta}(\Box_g\psi+\alpha\psi)\Big)^2+\Big(\slashed\Delta(\Box_g\psi+\alpha\psi)\Big)^2\bigg) r^2\mathrm{d}r\mathrm{d}\omega\bigg]\nonumber\\
    \geq& C\int_{\Sigma_{t^*}}\xi(r)\bigg[|\slashed\nabla^3\psi|^2+|\slashed\nabla^4\psi|^2\bigg]r^2\mathrm{d}r\mathrm{d}\omega,
\end{align*}
for some $C>0$.
\end{proof}

\subsubsection{Estimating general derivatives on the entire exterior region (\textbf{IV})}
\label{sec:Level2}
Given Proposition \ref{prop:FarKthOrder} and Theorem \ref{thm:kNear}, we now have the following estimate for general $k$\textsuperscript{th}-order spatial derivatives on the entire slice $\Sigma_{t^*}$.
\begin{theorem}[A non-degenerate $k$\textsuperscript{th}-order $L^2$ estimate on spacelike slices]
\label{thm:KElliptic}
   Let $\psi$ be a function in $H^k_{KAdS}(\Sigma_{t^*})$ for $k\geq 2$ satisfying Dirichlet boundary conditions (\ref{Dirichlet}). Then,
    \begin{align}
        &\sum_{\substack{|\sigma|=k,\\\sigma_2\leq k-1}}\int_{\Sigma_{t^*}}r^{2\sigma_2}|\slashed\nabla^{\sigma_1}\partial_{t^*}^{\sigma_2}\partial_r^{\sigma_3}\psi|^2 r^{2}\mathrm{d}r\mathrm{d}\omega\nonumber\\
        \leq &C\sum_{\substack{|\alpha|\leq k-1,\\ |\beta|\leq k- 2,\\|\gamma|\leq\big\lceil\frac{k-2}{2}\big\rceil}}\bigg[||\Gamma^{\alpha}\psi||^2_{H^1_{KAdS}(\Sigma_{t^*})}+\int_{\Sigma_{t^*}}\bigg(\Big(\Gamma^{\beta}(\Box_g\psi+\alpha\psi)\Big)^2+\Big(\Box_g^{\gamma}(\Box_g\psi+\alpha\psi)\Big)^2\bigg) r^2\mathrm{d}r\mathrm{d}\omega\bigg]\label{EstOfThm3.10}
    \end{align}
    for some $C>0$.
\end{theorem}

\subsection{Sobolev embedding in spacelike slices}
\label{sec:Sobolev}
We recall the Sobolev embedding result for a general complete Riemannian manifold with positive injectivity radius and Ricci curvature bounded below \cite{Riemannian}, stating it for the spheres $S^2_{t^*,r}\subset\mathcal{M}$. Theorem \ref{SliceSobolev}, an embedding on spacelike slices $\Sigma_{t^*}$ will follow as a corollary.

\begin{theorem}[$L^{\infty}$ Sobolev embedding on spheres $S^2_{t^*,r}$ {\cite[Thm 3.4]{Riemannian}}]
\label{SphericalSobolev}
    Fix $(t^*,r)$ and consider the Riemannian $2$-submanifold $S^2_{t^*,r}$ of $\mathcal{M}$ equipped with the induced metric. Let $u:S^2_{t^*,r}\to\mathbb{R}$ be a function in the Sobolev space \begin{align*}
        H^2_{S^2_{t^*,r}}=\bigg\{f:S^2_{t^*,r}\to\mathbb{R}\text{ }\bigg|\text{ }\int_{S^2_{t^*,r}}\sum_{j=0}^2|\slashed{\nabla}^j f|^2 r^2\mathrm{d}\omega<\infty\bigg\}.
    \end{align*}
    Then one can bound $u$ uniformly,
    \begin{align*}
        \sup_{S^2_{t^*,r}}|u|^2 &\leq \tilde{c}\int_{S^2_{t^*,r}}\sum_{j=0}^2|\slashed{\nabla}^j u|^2 r^2\mathrm{d}\omega,    
    \end{align*}
    for some $\tilde{c}>0$.
\end{theorem}
Applying Theorem \ref{thm:KElliptic} and Theorem \ref{SphericalSobolev} allows us to derive the following embeddings.

\begin{theorem}[$L^{\infty}$ Sobolev embedding for normalised derivatives on spacelike slices $\Sigma_{t^*}$]
\label{SliceSobolev}
    Fix $t^*$ and let $\psi$ be a function in $H^k_{KAdS}(\Sigma_{t^*})$ for $k\geq 3$ satisfying Dirichlet boundary conditions (\ref{Dirichlet}). Then one can bound $\psi$ and its $D$ derivatives (\ref{UnitDerivatives}) up to order $|\sigma|\leq k-3$ uniformly as
    \begin{align}
    &||D^{\sigma}\psi||_{L^{\infty}(\Sigma_{t^*})}\nonumber\\
     \leq&\frac{C}{r^{\frac{1}{2}}}\Bigg[\sum_{\substack{|\alpha|\leq |\sigma|+2,\\|\beta|\leq|\sigma|+1 ,\\|\gamma|\leq\big\lceil\frac{|\sigma|+1}{2}\big\rceil}}\bigg[||\Gamma^{\alpha}\psi||^2_{H^1_{KAdS}(\Sigma_{t^*})}+\int_{\Sigma_{t^*}}\bigg(\Big(\Gamma^{\beta}(\Box_g\psi+\alpha\psi)\Big)^2+\Big(\Box_g^{\gamma}(\Box_g\psi+\alpha\psi)\Big)^2\bigg) r^2\mathrm{d}r\mathrm{d}\omega\bigg]\Bigg]^{\frac{1}{2}}\label{SobolevEst}
    \end{align}
    for some $C>0$. Up to order $|\sigma|=k-4$, one has estimate (\ref{SobolevEst}) for the $\overline{D}$ derivatives (\ref{UnitDerivatives}) of $\psi$.
\end{theorem}

\begin{proof}
Given that $\psi|_{\mathcal{I}}=0$, we can apply the Fundamental Theorem of Calculus to obtain
    \begin{align}
    |D^{\sigma}\psi|&\leq \int_r^{\infty}|\partial_rD^{\sigma}\psi|\mathrm{d}r\nonumber\\
    &\leq\bigg(\int_{r}^{\infty}\frac{1}{r^2}\mathrm{d}r\bigg)^{\frac{1}{2}}\bigg(\int_{r}^{\infty}|\partial_rD^{\sigma}\psi|^2 r^2\mathrm{d}r\bigg)^{\frac{1}{2}}\nonumber\\
    &\leq\frac{1}{r^{\frac{1}{2}}}\bigg(\int_r^{\infty}|\partial_rD^{\sigma}\psi|^2 r^2\mathrm{d}r\bigg)^{\frac{1}{2}}.
\end{align}
Then, by Theorem \ref{SphericalSobolev}, we have that
\begin{align}
    &\int_r^{\infty}|\partial_rD^{\sigma}\psi|^2 r^2\mathrm{d}r\nonumber\\ \leq&\tilde{c}\sum_{j= 0}^2\int_{\Sigma_{t^*}}|\slashed\nabla^{j}\partial_rD^{\sigma}\psi|^2 r^4\mathrm{d}r\mathrm{d}\omega\nonumber\\
    \leq&C\sum_{\substack{|\alpha|\leq |\sigma|+2,\\|\beta|\leq|\sigma|+1,\\|\gamma|\leq\big\lceil\frac{|\sigma|+1}{2}\big\rceil}}\bigg[||\Gamma^{\alpha}\psi||^2_{H^1_{KAdS}(\Sigma_{t^*})}+\int_{\Sigma_{t^*}}\bigg(\Big(\Gamma^{\beta}(\Box_g\psi+\alpha\psi)\Big)^2+\Big(\Box_g^{\gamma}(\Box_g\psi+\alpha\psi)\Big)^2\bigg) r^2\mathrm{d}r\mathrm{d}\omega\bigg]\label{ApplyKthElliptic}
\end{align}
for some $C>0$, where (\ref{ApplyKthElliptic}) follows from Theorem \ref{thm:KElliptic}. Thus,
\begin{align*}
     &|D^{\sigma}\psi|\nonumber\\
     \leq&\frac{C}{r^{\frac{1}{2}}}\Bigg[\sum_{\substack{|\alpha|\leq |\sigma|+2,\\|\beta|\leq|\sigma|+1,\\|\gamma|\leq\big\lceil\frac{|\sigma|+1}{2}\big\rceil}}\bigg[||\Gamma^{\alpha}\psi||^2_{H^1_{KAdS}(\Sigma_{t^*})}+\int_{\Sigma_{t^*}}\bigg(\Big(\Gamma^{\beta}(\Box_g\psi+\alpha\psi)\Big)^2+\Big(\Box_g^{\gamma}(\Box_g\psi+\alpha\psi)\Big)^2\bigg) r^2\mathrm{d}r\mathrm{d}\omega\bigg]\Bigg]^{\frac{1}{2}}
\end{align*}
for some $C>0$, as required. The estimate for $\overline{D}^{\sigma}\psi$ with $|\sigma|\leq k-4$ follows similarly, where the top-order $T$ derivative is controlled with a stronger weight by the zeroth-order term of (\ref{kthEnergy}) after commuting $|\sigma|$ times with $T$.
\end{proof}

We also note the following Sobolev embedding for derivatives where only 1 is arbitrary.

\begin{theorem}[$L^{\infty}$ Sobolev embedding for derivatives where only 1 is arbitrary]
\label{thm:SobolevOne}
    Fix $t^*$ and let $\psi$ be a function in $H^k_{KAdS}(\Sigma_{t^*})$ for $k\geq 3$ satisfying Dirichlet boundary conditions (\ref{Dirichlet}). Then, for each $|\sigma|\leq k-4$, one has
    \begin{align}
    &||N^{\sigma_1}T^{\sigma_2}\Phi^{\sigma_3}D\psi||_{L^{\infty}(\Sigma_{t^*})}\nonumber\\
     \leq&\frac{C}{r^{\frac{1}{2}}}\Bigg[\sum_{\substack{|\alpha|\leq |\sigma|+3,\\|\beta|\leq|\sigma|+2 ,\\|\gamma|\leq\big\lceil\frac{|\sigma|+2}{2}\big\rceil}}\bigg[||\Gamma^{\alpha}\psi||^2_{H^1_{KAdS}(\Sigma_{t^*})}+\int_{\Sigma_{t^*}}\bigg(\Big(\Gamma^{\beta}(\Box_g\psi+\alpha\psi)\Big)^2+\Big(\Box_g^{\gamma}(\Box_g\psi+\alpha\psi)\Big)^2\bigg)r^2\mathrm{d}r\mathrm{d}\omega\bigg]\Bigg]^{\frac{1}{2}}\label{SobolevOne}
    \end{align}
\end{theorem}

\subsection{Energy estimates for the linear inhomogeneous problem}
\label{sec:EgEst}
We now prove general $n$\textsuperscript{th}-order energy estimates for a solution $\psi$ of (\ref{FullProblem}). We begin with a result for the problem (\ref{FullProblem}) with general inhomogeneity $\tilde{\mathcal{F}}$.
\begin{proposition}[An $n$\textsuperscript{th}-order energy estimate for solutions of the linear inhomogeneous problem]
    \label{prop:LeaveFAlone}
    Suppose $\psi$ is a $CH^k_{KAdS}$ solution of (\ref{FullProblem}), with inhomogeneity $\tilde{\mathcal{F}}$ satisfying
    \begin{gather*}
        \sum_{|\sigma|\leq n-1}||N^{\sigma_1}T^{\sigma_2}\Phi^{\sigma_3}\tilde{\mathcal{F}}||_{L^2(\Sigma_{t^*})}<\infty,\\
        \sum_{\substack{|\beta|\leq n- 2,\\|\gamma|\leq\big\lceil\frac{n-2}{2}\big\rceil}}\int_{\Sigma_{t^*}}\bigg((\Gamma^{\beta}\tilde{\mathcal{F}})^2+(\Box_g^{\gamma}\tilde{\mathcal{F}})^2\bigg) r^2\mathrm{d}r\mathrm{d}\omega<\infty
    \end{gather*}
    for all $n\leq k$, $t^*$ in $[\tau,\tau']$. Then $\psi$ satisfies the general $n$\textsuperscript{th} order energy estimate
    \begin{align}
    \label{GenFEstimate}
        &||\psi||^2_{H^n_{KAdS}(\Sigma_{t^*})}\nonumber\\
        \leq& C\bigg[F^n_{\mathcal{H}^+\cap[\tau,\tau']}[\psi]+||\psi||^2_{H^n_{KAdS}(\Sigma_{\tau'})}+\sum_{\substack{|\beta|\leq n- 2,\\|\gamma|\leq\big\lceil\frac{n-2}{2}\big\rceil}}\int_{\Sigma_{t^*}}\bigg((\Gamma^{\beta}\tilde{\mathcal{F}})^2+(\Box_g^{\gamma}\tilde{\mathcal{F}})^2\bigg) r^2\mathrm{d}r\mathrm{d}\omega\nonumber\\
        &\quad+\int_{\tau}^{\tau'}\bigg(C_{\kappa,n}||\psi||^2_{H^n(\Sigma_{t^*})}+\sum_{|\beta|\leq n-3}||\Gamma^{\beta}\tilde{\mathcal{F}}||^2_{L^2(\Sigma_{t^*})}\bigg)\mathrm{d}t^*\nonumber\\
        &\quad+\bigg(\int_{\tau}^{\tau'}\sum_{|\sigma|\leq n-1}||N^{\sigma_1}T^{\sigma_2}\Phi^{\sigma_3}\tilde{\mathcal{F}}||_{L^2(\Sigma_{t^*})}\mathrm{d}t^*\bigg)^2\bigg],
    \end{align}
    for some $C>0$ and all $n\leq k$ where $F^n_{\mathcal{H}^+\cap[\tau,\tau']}$ is a flux term through the event horizon, involving derivatives of the scattering data $h_{\mathcal{H}^+}$ of up to $n$\textsuperscript{th}-order. The constant $C_{\kappa,n}$ depends on the spacetime surface gravity $\kappa$ and $n$.
\end{proposition}

\begin{proof}
We apply the vector field $N$ (\ref{RedshiftVF}) as a multiplier and integrate the equation over the spacetime region $\mathcal{D}_{[\tau,\tau']}=\mathcal{M}\cap[\tau,\tau']$. By direct calculation and an application of the Divergence Theorem, one obtains that
\begin{align*}
    \int_{\mathcal{D}_{[\tau,\tau']}}|N\psi\cdot \tilde{\mathcal{F}}| r^2\mathrm{d}t^*\mathrm{d}r\mathrm{d}\omega &=\int_{\mathcal{D}_{[\tau,\tau']}}|N\psi(\Box_g\psi+\alpha\psi) |r^2\mathrm{d}t^*\mathrm{d}r\mathrm{d}\omega\\
    &\geq C\int_{\partial\mathcal{D}_{[\tau,\tau']}}J^N_{\mu}[\psi]\cdot n_{\Sigma_{t^*}}^{\mu}\cdot\mathrm{vol}_{\Sigma_{t^*}}
\end{align*}
for some $C>0$. Hence, by integrating (\ref{CurrentBulk}) over $\mathcal{D}$ and applying Proposition \ref{prop:NCoercive}, it is clear that $\psi$ satisfies
\begin{align*}
    &\int_{\Sigma_{\tau}}J^N_{\mu}[\psi]\cdot n_{\Sigma_{t^*}}^{\mu}\cdot\mathrm{vol}_{\Sigma_{t^*}}+\tilde{C}_{\alpha}\bigg[\int_{\Sigma_{\tau'}\cap\{r<r_1\}}\psi^2\frac{\Sigma}{\Xi}\mathrm{d}r\mathrm{d}\omega+\int_{\mathcal{D}_{[\tau,\tau']}\cap\{r<r_1\}}\Big(\psi^2+(\partial_{t^*}\psi)^2\Big)\frac{\Sigma}{\Xi}\mathrm{d}r\mathrm{d}\omega\mathrm{d}t^*\bigg]\\
    \leq &C\bigg(\int_{\mathcal{H}^+\cap[\tau,\tau']}J_{\mu}^N[\psi]\cdot n_{\mathcal{H}^+}^{\mu}\cdot\mathrm{vol}_{\mathcal{H}^+}+\int_{\Sigma_{\tau'}}J^N_{\mu}[\psi]\cdot n_{\Sigma_{t^*}}^{\mu}\cdot\mathrm{vol}_{\Sigma_{t^*}}\\
    &+\int_{\mathcal{D}_{[\tau,\tau']}}\bigg(|N\psi\cdot \tilde{\mathcal{F}}|r^2\mathrm{d}r\mathrm{d}\omega+|K^N[\psi]|\mathrm{vol}_{\Sigma_{t^*}}\bigg)\mathrm{d}t^*\\
    &+\tilde{C}_{\alpha}\bigg[\int_{\Sigma_{\tau'}\cap\{r<r_1\}}\psi^2\frac{\Sigma}{\Xi}\mathrm{d}r\mathrm{d}\omega+\int_{\mathcal{D}_{[\tau,\tau']}\cap\{r<r_1\}}\Big(\psi^2+(\partial_{t^*}\psi)^2\Big)\frac{\Sigma}{\Xi}\mathrm{d}r\mathrm{d}\omega\mathrm{d}t^*\bigg]\bigg).
\end{align*}
where we have used the Dirichlet boundary condition to drop the boundary term on $\mathcal{I}$. Applying Proposition \ref{prop:NCoercive} to the integrals on the spacelike slices gives the energy estimate
\begin{align*}
    ||\psi||^2_{H^1_{KAdS(\Sigma_{\tau})}}\leq& C\bigg[F^1_{\mathcal{H}^+\cap[\tau,\tau']}[\psi]+||\psi||^2_{H^1_{KAdS(\Sigma_{\tau'})}}\\
    &\quad+\int_{\mathcal{D}_{[\tau,\tau']}}\bigg(||\psi||^2_{H^1_{KAdS}(\Sigma_{t^*})}+||\partial_{t^*}\psi||^2_{H^1_{KAdS}(\Sigma_{t^*})}+|N\psi\cdot \tilde{\mathcal{F}}| r^2\mathrm{d}r\mathrm{d}\omega\\
    &\qquad\qquad\quad+|K^N[\psi]|\mathrm{vol}_{\Sigma_{t^*}}\bigg)\mathrm{d}t^*\bigg],
\end{align*}
where
\begin{align*}
    F^1_{\mathcal{H}^+\cap[\tau,\tau']}[\psi]=\int_{\mathcal{H}^+\cap[\tau,\tau']}J_{\mu}^N[\psi]\cdot n_{\mathcal{H}^+}^{\mu}\cdot\mathrm{vol}_{\mathcal{H}^+}.
\end{align*}
We shall refer to the higher-order analogues of this horizon flux term as
\begin{align*}
    F^m_{\mathcal{H}^+\cap[\tau,\tau']}[\psi]=\sum_{|\sigma|\leq m-1}\int_{\mathcal{H}^+\cap[\tau,\tau']}J_{\mu}^N[\Gamma^{\sigma}\psi]\cdot n_{\mathcal{H}^+}^{\mu}\cdot\mathrm{vol}_{\mathcal{H}^+}.
\end{align*}
Since we can perform the same estimate over smaller time intervals, we in fact have the following $L^{\infty}$ in time estimate:
\begin{align*}
    &\sup_{t^*\in[\tau,\tau']}||\psi||^2_{H^1_{KAdS(\Sigma_{t^*})}}\\\leq &C\bigg[F^1_{\mathcal{H}^+[\tau,\tau']}[\psi]+||\psi||^2_{H^1_{KAdS(\Sigma_{\tau'})}}+\int_{\mathcal{D}_{[\tau,\tau']}}\bigg(||\psi||^2_{H^1_{KAdS}(\Sigma_{t^*})}+||\partial_{t^*}\psi||^2_{H^1_{KAdS}(\Sigma_{t^*})}\\
    &\qquad\qquad\qquad\qquad\qquad\qquad\qquad\qquad\qquad+|N\psi\cdot\tilde{\mathcal{F}}| r^2\mathrm{d}r\mathrm{d}\omega+|K^N[\psi]|\mathrm{vol}_{\Sigma_{t^*}}\bigg)\mathrm{d}t^*\bigg],
\end{align*}
In order to generate the general $k$\textsuperscript{th}-order estimate appearing in the Theorem, we commute the equation with $\Gamma^{\sigma}$:
\begin{align}
\label{Commuted}
    &\sum_{|\sigma|\leq n-1}\sup_{t^*\in[\tau,\tau']}||\Gamma^{\sigma}\psi||^2_{H^1_{KAdS(\Sigma_{t^*})}}\nonumber\\
    \leq&C\bigg[F^n_{\mathcal{H}^+\cap[\tau,\tau']}[\psi]+||\psi||^2_{H^n_{KAdS(\Sigma_{\tau'})}}+\int_{\mathcal{D}_{[\tau,\tau']}}\sum_{\substack{|\sigma|\leq n-1,\\|\tilde{\sigma}|\leq n-1}}\bigg(||\psi||^2_{H^n_{KAdS}(\Sigma_{t^*})}
    +|K^N[N^{\sigma_1}T^{\sigma_2}\Phi^{\sigma_3}\psi]|\mathrm{vol}_{\Sigma_{t^*}}\nonumber\\
    &\qquad\qquad\qquad\qquad\qquad\qquad\qquad\qquad+\Big(|\Gamma^{\tilde{\sigma}}N\psi\cdot \Gamma^{\sigma}\tilde{\mathcal{F}}|+|[\Box_g,N^{\sigma_1}](\psi)|\Big) r^2\mathrm{d}r\mathrm{d}\omega\bigg)\mathrm{d}t^*\bigg].
\end{align}
The term in (\ref{Commuted}) arising from the inhomogeneity $\tilde{\mathcal{F}}$ satisfies
\begin{align}
    &\int_{\tau}^{\tau'}\int_{\Sigma_{t^*}}\sum_{\substack{|\sigma|\leq n-1,\\|\tilde{\sigma}\leq n-1}}|N^{\tilde{\sigma}_1+1}T^{\tilde{\sigma_2}}\Phi^{\tilde{\sigma}_3}\psi\cdot N^{\sigma_1}T^{\sigma_2}\Phi^{\sigma_3}\tilde{\mathcal{F}}|r^2\mathrm{d}r\mathrm{d}\omega\mathrm{d}t^*\nonumber\\
    \leq &\int_{\tau}^{\tau'}\sum_{\substack{|\sigma|\leq n-1,\\|\tilde{\sigma}|\leq n-1}}||N^{\tilde{\sigma}_1+1}T^{\tilde{\sigma_2}}\Phi^{\tilde{\sigma}_3}\psi||_{L^2(\Sigma_{t^*})}||N^{\sigma_1}T^{\sigma_2}\Phi^{\sigma_3}\tilde{\mathcal{F}}||_{L^2(\Sigma_{t^*})}\mathrm{d}t^*\label{alphabetagamma}\\
    \leq &\frac{\varepsilon}{2}\sup_{t^*\in[\tau,\tau']}||\psi||^2_{H^n_{KAdS}(\Sigma_{t^*})}+\frac{1}{2\varepsilon}\bigg(\int_{\tau}^{\tau'}\sum_{|\sigma|\leq n-1}||N^{\sigma_1}T^{\sigma_2}\Phi^{\sigma_3}\tilde{\mathcal{F}}||_{L^2(\Sigma_{t^*})}\mathrm{d}t^*\bigg)^2\label{Last}
\end{align}
Estimate (\ref{alphabetagamma}) follows from an application of the Cauchy-Schwarz inequality, and (\ref{Last}) from the definition of the $n$\textsuperscript{th}-order energy of $\psi$. We choose $\varepsilon>0$ sufficiently small to allow the first term of (\ref{Last}) to eventually be absorbed by the left-hand side of estimate (\ref{Commuted}).\par 

Returning to (\ref{Commuted}), we use that the redshift vector field $N$ is Killing on $r\geq r_1$, so that bulk terms generated by commuting with $N$ are non-vanishing only on the compact $r$-region $r_+\leq r<r_1$. This ensures no large-$r$ bulk contribution, giving the bound
\begin{align}
    \int_{\Sigma_{t^*}}\sum_{|\sigma|\leq n-1}|K^N[N^{\sigma_1}T^{\sigma_2}\Phi^{\sigma_3}\psi]|\mathrm{vol}_{\Sigma_{t^*}}&\leq C_{\kappa,n}\int_{\Sigma_{t^*}}\sum_{|\sigma|\leq n-1}J^N_{\mu}[N^{\sigma_1}T^{\sigma_2}\Phi^{\sigma_3}\psi]\cdot n^{\mu}_{\Sigma_{t^*}}\cdot\mathrm{vol}_{\Sigma_{t^*}}\nonumber\\
    &\leq C_{\kappa,n}||\psi||^2_{H^n_{KAdS}(\Sigma_{t^*})},\label{BulkBound}
\end{align}
where $C_{\kappa,n}$ depends on the spacetime surface gravity $\kappa$ and the number of commutations. This leaves only the commutator term on the right-hand side of (\ref{Commuted}) to be treated. By direct calculation (for example, at first-order (\ref{NCommutator})) one finds that $[\Box_g,N^{\sigma_1}]$ is an operator of order $\sigma_1$. Clearly, its terms involve at most two general derivatives. Moreover, since $N=T$ near infinity (and $[\Box_g,T]=0$), one need not worry about $r$-weights. As such, Proposition \ref{prop:2ofk} gives
\begin{align}
    \int_{\Sigma_{t^*}}\sum_{\sigma_1\leq n-1}|[\Box_g,N^{\sigma_1}](\psi)|\leq &C\int_{\Sigma_{t^*}}\sum_{|\sigma|\leq n-3}\bigg[|\Gamma^{\sigma}\slashed\nabla^2\psi|^2+r^2|\Gamma^{\sigma}\slashed\nabla\partial_r\psi|^2+r^4(\Gamma^{\sigma}\partial_r^2\psi)^2\bigg]r^2\mathrm{d}r\mathrm{d}\omega\nonumber\\
        \leq&C\sum_{\substack{|\alpha|\leq n-2,\\ |\beta|\leq n-3}}\bigg[||\Gamma^{\alpha}\psi||_{H^1_{KAdS}(\Sigma_{t^*})}^2+\int_{\Sigma_{t^*}}\Big(\Gamma^{\beta}(\Box_g\psi+\alpha\psi)\Big)^2 r^2\mathrm{d}r\mathrm{d}\omega\bigg]\nonumber\\
        =&C\sum_{\substack{|\alpha|\leq n-2,\\ |\beta|\leq n-3}}\bigg[||\Gamma^{\alpha}\psi||_{H^1_{KAdS}(\Sigma_{t^*})}^2+\int_{\Sigma_{t^*}}(\Gamma^{\beta}\tilde{\mathcal{F}})^2 r^2\mathrm{d}r\mathrm{d}\omega\bigg].\label{TreatedCommutator}
\end{align}
Combining estimates (\ref{Commuted}), (\ref{Last}), (\ref{BulkBound}) and (\ref{TreatedCommutator}) gives
\begin{align}
     &\sum_{|\sigma|\leq n-1}\sup_{t^*\in[\tau,\tau']}||\Gamma^{\sigma}\psi||^2_{H^1_{KAdS(\Sigma_{t^*})}}\nonumber\\
     \leq&C\bigg[F^n_{\mathcal{H}^+\cap[\tau,\tau']}[\psi]+||\psi||^2_{H^n_{KAdS}(\Sigma_{\tau'})}+\bigg(\int_{\tau}^{\tau'}\sum_{|\sigma|\leq n-1}||N^{\sigma_1}T^{\sigma_2}\Phi^{\sigma_3}\tilde{\mathcal{F}}||_{L^2(\Sigma_{t^*})}\mathrm{d}t^*\bigg)^2\nonumber\\
        &+\int_{\tau}^{\tau'}\bigg(C_{\kappa,n}||\psi||^2_{H^n(\Sigma_{t^*})}+\sum_{|\beta|\leq n-3}||\Gamma^{\beta}\tilde{\mathcal{F}}||^2_{L^2(\Sigma_{t^*})}\bigg)\mathrm{d}t^*\bigg],\label{SecondStepCommuted}
\end{align}
for some $C>0$, where the first term on the right-hand side of (\ref{TreatedCommutator}) has been absorbed by the higher order $C_{\kappa,n}$-weighted bulk term. The left-hand side of estimate (\ref{SecondStepCommuted}) does not control all $n$\textsuperscript{th}-order deriatives of $\psi$. In particular, the purely angular, angular-radial and purely radial derivatives of order $n\geq 2$ are missing.\par

By Theorem \ref{thm:KElliptic}, we have that
\begin{align}
    &\sum_{|\sigma|\leq n}\int_{\Sigma_{t^*}}r^{2\sigma_2}|\slashed\nabla^{\sigma_1}\partial_r^{\sigma_2}\psi|^2 r^{2}\mathrm{d}r\mathrm{d}\omega\nonumber\\
        \leq   &C\sum_{\substack{|\alpha|\leq n-1,\\ |\beta|\leq n- 2,\\|\gamma|\leq\big\lceil\frac{n-2}{2}\big\rceil}}\bigg[||\Gamma^{\alpha}\psi||^2_{H^1_{KAdS}(\Sigma_{t^*})}+\int_{\Sigma_{t^*}}\bigg(\Big(\Gamma^{\beta}(\Box_g\psi+\alpha\psi)\Big)^2+\Big(\Box_g^{\gamma}(\Box_g\psi+\alpha\psi)\Big)^2\bigg) r^2\mathrm{d}r\mathrm{d}\omega\bigg]\nonumber\\
        =&C\sum_{\substack{|\alpha|\leq n-1,\\ |\beta|\leq n- 2,\\|\gamma|\leq\big\lceil\frac{n-2}{2}\big\rceil}}\bigg[||\Gamma^{\alpha}\psi||^2_{H^1_{KAdS}(\Sigma_{t^*})}+\int_{\Sigma_{t^*}}\bigg((\Gamma^{\beta}\tilde{\mathcal{F}})^2+(\Box_g^{\gamma}\tilde{\mathcal{F}})^2\bigg) r^2\mathrm{d}r\mathrm{d}\omega\bigg],\label{FReplacement}
\end{align}
Finally, combining estimates (\ref{SecondStepCommuted}) and (\ref{FReplacement}) produces estimate (\ref{GenFEstimate}) for all derivatives of $\psi$ up to $n$\textsuperscript{th}-order.
\end{proof}

\subsection{Estimating the nonlinearity}
\label{sec:EstTheNonlin}
Next, we estimate the nonlinearity $\mathcal{F}$ appearing in problem (\ref{FullProblem}) of the form (\ref{Quadratic}). This will eventually allow us to treat the terms involving $\mathcal{F}$ on the right-hand side of estimate (\ref{GenFEstimate}). We begin with the following proposition which provides an estimate for the top-order $\mathcal{F}$ terms on the right-hand side of (\ref{GenFEstimate}). Then we will comment on treating the other (lower-order) $\mathcal{F}$-dependent terms on the right-hand side of (\ref{GenFEstimate}). We note that Proposition \ref{prop:AbstractEnergyEst} confirms that the nonlinearity (\ref{Quadratic}) indeed satisfies Condition \ref{NonlinAssump}. 

\begin{proposition}[Estimates for the nonlinearity]
    \label{prop:AbstractEnergyEst}
    Let $\psi$ be as in Proposition \ref{prop:LeaveFAlone} with $k\geq 9$ and suppose $\mathcal{F}$ is a nonlinearity of the form (\ref{Quadratic}). Then $\mathcal{F}$ satisfies the estimates 
\begin{equation}
    \label{THEEstimate}
    \begin{gathered}
                \sum_{|\sigma|\leq k-1}||N^{\sigma_1}T^{\sigma_2}\Phi^{\sigma_3}\mathcal{F}||_{L^2(\Sigma_{t^*})}\leq C||\psi||_{H^k_{KAdS}(\Sigma_{t^*})}||\psi||_{H^{k-1}_{KAdS}(\Sigma_{t^*})}\\
                \sum_{|\gamma|\leq\lceil\frac{k-2}{2}\rceil}||\Box_g^{\gamma}\mathcal{F}||^2_{L^2(\Sigma_{t^*})}\leq \tilde{C}||\psi||^4_{H^{k-1}_{KAdS}(\Sigma_{t^*})}
    \end{gathered}
\end{equation}
    for some $C,\tilde{C}>0$.
\end{proposition}

\begin{proof}
By the chain rule and assumption (\ref{Quadratic}) on $\mathcal{F}$, we have that
\begin{align}
    \sum_{|\sigma|\leq k-1}||N^{\sigma_1}T^{\sigma_2}\Phi^{\sigma_3}\mathcal{F}||_{L^2(\Sigma_{t^*})}
    =&\sum_{|\sigma|\leq k-1}||N^{\sigma_1}T^{\sigma_2}\Phi^{\sigma_3}(F^{\mu\nu}(t^*,x)\cdot e_{\mu}\psi e_{\nu}\psi)||_{L^2(\Sigma_{t^*})}\nonumber\\
    \leq&C_F\sum_{\substack{|\sigma|\leq k-1\\ \delta\leq\sigma_1,\gamma\leq\sigma_2,\rho\leq\sigma_3\\0\leq\mu,\nu\leq 3}}||N^{\sigma_1-\delta}T^{\sigma_2-\gamma}\Phi^{\sigma_3-\rho}e_{\mu}\psi\cdot N^{\delta}T^{\gamma}\Phi^{\rho}e_{\nu}\psi||_{L^2(\Sigma_{t^*})}\label{ExpandF},
\end{align}
where $C_F$ is a constant depending on the functions $F^{\mu\nu}(t^*,x)$ appearing in $\mathcal{F}$ (\ref{Quadratic}). Assuming, without loss of generality, that $\delta+\gamma\leq\frac{k-1}{2}$ yields
\begin{align}
    &\sum_{\substack{|\sigma|\leq k-1\\\delta\leq\sigma_1,\gamma\leq\sigma_2,\rho\leq\sigma_3\\0\leq\mu,\nu\leq 3}}||N^{\sigma_1-\delta}T^{\sigma_2-\gamma}\Phi^{\sigma_3-\rho}e_{\mu}\psi\cdot N^{\delta}T^{\gamma}\Phi^{\rho}e_{\nu}\psi||_{L^2(\Sigma_{t^*})}\nonumber\\
    \leq&\sum_{\substack{|\sigma|\leq k-1\\|\delta|\leq\frac{k-1}{2},\\0\leq\mu,\nu\leq 3}}||N^{\sigma_1}T^{\sigma_2}\Phi^{\sigma_3}e_{\mu}\psi||_{L^2(\Sigma_{t^*})} ||N^{\delta_1}T^{\delta_2}\Phi^{\delta_3}e_{\nu}\psi||_{L^{\infty}(\Sigma_{t^*})}\nonumber\\
    \leq& C||\psi||_{H^k_{KAdS}(\Sigma_{t^*})}\bigg[||\psi||_{H^{k-1}_{KAdS}(\Sigma_{t^*})}+\sum_{|\gamma|\leq\lceil\frac{|\delta|+2}{2}\rceil}||\Box_g^{\gamma}\mathcal{F}||_{L^2(\Sigma_{t^*})}\bigg].\label{FirstBox}
\end{align}
Here, the last inequality follows from Theorem \ref{thm:SobolevOne} and the assumption $k\geq 9$. By the assumption $k\geq 9$, the $\Box$ terms on the right-hand side of (\ref{FirstBox}) are of lower-order than those  appearing in (\ref{THEEstimate}), so we demonstrate how to treat those first, as follows.\par

We expand the $\Box_g$ term as follows, treating the second-order terms in each $\Box$ explicitly. The other terms are dealt with similarly.
\begin{align*}
    &\sum_{|\gamma|\leq\lceil\frac{k-2}{2}\rceil}||\Box_g^{\gamma}\mathcal{F}||^2_{L^2(\Sigma_{t^*})}\nonumber\\
    \leq& C\sum_{|\gamma|\leq\lceil\frac{k-2}{2}\rceil}\bigg|\bigg|(r^2\partial_r^2)^{\gamma_1}\bigg(\frac{1}{r^2}\partial_{t^*}^2\bigg)^{\gamma_2}\bigg(\frac{1}{r^2}\partial_{\theta}^2\bigg)^{\gamma_3}\bigg(\frac{1}{r^2}\partial_{\phi}^2\bigg)^{\gamma_4}(F^{\mu\nu}(t^*,x)\cdot e_{\mu}\psi e_{\nu}\psi)\bigg|\bigg|^2_{L^2(\Sigma_{t^*})}+\ldots\nonumber\\
    \leq& C_F\sum_{\substack{|\gamma|\leq\lceil\frac{k-2}{2}\rceil,\\\delta\leq\gamma_1,\rho\leq\gamma_2,\sigma\leq\gamma_3,\beta\leq\gamma_4\\0\leq \mu,\nu\leq 3}}\bigg|\bigg|(r^2\partial_r^2)^{\gamma_1-\delta}\bigg(\frac{1}{r^2}\partial_{t^*}^2\bigg)^{\gamma_2-\rho}\cdots e_{\mu}\psi\cdot(r^2\partial_r^2)^{\delta}\bigg(\frac{1}{r^2}\partial_{t^*}^2\bigg)^{\rho}\cdots e_{\nu}\psi\bigg|\bigg|^2_{L^2(\Sigma_{t^*})},
\end{align*}
where $C_F$ is a constant depending on the functions $F^{\mu\nu}(t^*,x)$ appearing in $\mathcal{F}$ (\ref{Quadratic}). Assuming, without loss of generality, that $\delta+\rho+\sigma+\beta\leq\frac{1}{2}\lceil\frac{k-2}{2}\rceil\leq\frac{k-2}{4}$ yields
\begin{align*}
    &\sum_{\substack{|\gamma|\leq\lceil\frac{k-2}{2}\rceil,\\\delta\leq\gamma_1,\rho\leq\gamma_2,\sigma\leq\gamma_3,\beta\leq\gamma_4\\0\leq \mu,\nu\leq 3}}\bigg|\bigg|(r^2\partial_r^2)^{\gamma_1-\delta}\bigg(\frac{1}{r^2}\partial_{t^*}^2\bigg)^{\gamma_2-\rho}\cdots e_{\mu}\psi\cdot(r^2\partial_r^2)^{\delta}\bigg(\frac{1}{r^2}\partial_{t^*}^2\bigg)^{\rho}\cdots e_{\nu}\psi\bigg|\bigg|^2_{L^2(\Sigma_{t^*})}\nonumber\\
    \leq&\sum_{\substack{|\gamma|\leq\lceil\frac{k-2}{2}\rceil,\\|\sigma|\leq\frac{k-2}{4},\\0\leq \mu,\nu\leq 3}}\bigg|\bigg|(r^2\partial_r^2)^{\gamma_1}\bigg(\frac{1}{r^2}\partial_{t^*}^2\bigg)^{\gamma_2}\cdots e_{\mu}\psi\bigg|\bigg|^2_{L^2(\Sigma_{t^*})}\bigg|\bigg|(r^2\partial_r^2)^{\sigma_1}\bigg(\frac{1}{r^2}\partial_{t^*}^2\bigg)^{\sigma_2}\cdots e_{\nu}\psi\bigg|\bigg|^2_{L^{\infty}(\Sigma_{t^*})}\nonumber\\
    \leq& C||\psi||^4_{H^{k-1}_{KAdS}(\Sigma_{t^*})},
\end{align*}
for some $C>0$. Here, we have applied Theorem \ref{SliceSobolev} (treating lower order terms on the right-hand side of (\ref{SobolevEst}) via the same argument) and the assumption $k\geq 9$.\par

Applying the same argument to the $\Box$ terms on the right-hand side of (\ref{FirstBox}), we have that
\begin{align*}
    \sum_{|\sigma|\leq k-1}||N^{\sigma_1}T^{\sigma_2}\Phi^{\sigma_3}\mathcal{F}||_{L^2(\Sigma_{t^*})}\leq C||\psi||_{H^k_{KAdS}(\Sigma_{t^*})}||\psi||_{H^{k-1}_{KAdS}(\Sigma_{t^*})}.
\end{align*}
\end{proof}
By the same method, we treat the other terms (of strictly lower order) on the right-hand side of (\ref{GenFEstimate}):
\begin{align}
    &\sum_{\substack{|\beta|\leq k- 2,\\|\gamma|\leq\big\lceil\frac{k-2}{2}\big\rceil}}\int_{\Sigma_{t^*}}(\Gamma^{\beta}\mathcal{F})^2 r^2\mathrm{d}r\mathrm{d}\omega+\int_{\tau}^{\tau'}\sum_{|\beta|\leq k-3}||\Gamma^{\beta}\mathcal{F}||^2_{L^2(\Sigma_{t^*})}\mathrm{d}t^*\nonumber\\
    \leq&C\int_{\tau}^{\tau'}||\psi||^2_{H^{k-2}_{KAdS}(\Sigma_{t^*})}||\psi||^2_{H^{k-3}_{KAdS}(\Sigma_{t^*})},\label{EstOtherFTerms}
\end{align}
for some $C>0$. We now have all the ingredients required to proceed to stating and proving the main result.

\section{The main result}
\label{sec:Result}
In order to precisely formulate the main result, we define the following $k$\textsuperscript{th}-order quantities $\mathfrak{D}_k$ which depend on the scattering data $h_{\mathcal{H}^+}$.

\begin{definition}[Weighted $k$\textsuperscript{th}-order data quantities]
\label{def:data}
Given scattering data $h_{\mathcal{H}^+}:\{t^*\in[0,\infty)\}\times S^2_{t^*,r_+}\to\mathbb{R}$ on $\mathcal{H}^+$, we define the quantity $\mathfrak{D}_k$ (which depends on derivatives of the scattering data $h_{\mathcal{H}}$ up to $k$\textsuperscript{th}-order) as
\begin{align*}
    \mathfrak{D}_{k}= \sum_{|\sigma|\leq k}\int_{S^2_{t^*,r_+}}\int_0^{\infty}|\exp(\Upsilon_{k}\cdot t^*)\slashed\nabla^{\sigma_1}\partial_{t^*}^{\sigma_2} h_{\mathcal{H}^+}\rvert^2\mathrm{d}t^*\mathrm{d}\omega,
\end{align*}
where $\Upsilon_{k}$ is a constant depending on $k$.
\end{definition}
Given this, we now state the main result.

\begin{theorem}[Exponentially decaying nonlinear waves on Kerr-AdS]
\label{MainResult}
Let $t_0^*>0$ be fixed, $k\geq 9$ and let ${h_{\mathcal{H}^+}}:\mathbb{R}^+\times S^2_{t^*,r_+}\to\mathbb{R}$ be sufficiently smooth scattering data with
\begin{align}
    \mathfrak{D}_{k}<\infty
\end{align}
for associated constant (see Definition \ref{def:data}) $\Upsilon_{\kappa, k}\geq B_k$, with $B_k>0$ sufficiently large depending on $k$ and the spacetime surface gravity $\kappa$. Then there exists a solution $\psi:\{t^*\geq t_0^*\}\times\{r>r_+\}\times S^2_{t^*,r}\to\mathbb{R}$ in $CH^k_{KAds}$ of problem (\ref{FullProblem}) satisfying
\begin{align}
\label{MainL2}
    ||\psi||^2_{H^k_{KAdS}(\Sigma_{t^*})}&\leq \overline{C}\exp(-B_k\cdot t^*)\mathfrak{D}_{k}.
\end{align}
for some $\overline{C}>0$. Moreover, the solution satisfies
\begin{equation}
\begin{aligned}
\label{MainEstimate2}
    \sum_{|\sigma|\leq k-3}\lvert\lvert r^{\frac{1}{2}}D^{\sigma}\psi\rvert\rvert^2_{L^{\infty}(\Sigma_{t^*})}&\leq C\exp(-B_{k}\cdot t^*)\mathfrak{D}_{k},\\
    \sum_{|\sigma|\leq k-4}\lvert\lvert r^{\frac{1}{2}}\overline{D}^{\sigma}\psi\rvert\rvert^2_{L^{\infty}(\Sigma_{t^*})}&\leq C\exp(-B_{k}\cdot t^*)\mathfrak{D}_{k}.
\end{aligned}
\end{equation}
for some $C>0$.
\end{theorem}

\begin{remark}[Uniqueness of solutions]
Note that Theorem \ref{MainResult} makes no claims regarding uniqueness of the solution $\psi$. However, by applying the estimates appearing in sections \ref{sec:EgEst} and \ref{sec:EstTheNonlin} to the difference of two solutions $\psi_1'$, $\psi_2'$, one can derive uniqueness in the class of exponentially decaying solutions. One can see \cite{SchwzScattering} for an analogous discussion of the problem of uniqueness in the setting of a scattering problem for exponentially decaying perturbations of the asymptotically flat Kerr black hole exterior.
\end{remark}

\section{The proof}
\label{sec:Proof}
This section is concerned with the proof of Theorem \ref{MainResult}. Firstly, a sequence of finite-in-time problems approximating the full problem are constructed (Section \ref{sec:FiniteProblems}). The bootstrap argument of Section \ref{sec:Bootstrap} is used to show existence of uniformly bounded, exponentially decaying solutions of these finite problems via the bootstrap improvement argument appearing in Section \ref{sec:FiniteEstimates}. Finally, a standard argument is used to ensure the sequence of solutions of the finite problems converges to the target solution of the global towards the future nonlinear problem. This appears in Section \ref{sec:Convergence}.

\subsection{An approximating sequence of nonlinear waves}
\label{sec:FiniteProblems}
Given a sequence $t_i^*\to\infty$ of times to the future of a fixed $t_0^*>0$, we consider the finite-in-time problems
\begin{align}
    \label{FiniteProblem}
    \begin{cases}
        \Box_g\psi_i + \alpha\psi_i = \mathcal{F}_{i} = \mathcal{F}(D\psi_i,D\psi_i),\\
        \psi_i|_{\mathcal{H}^+}={h_{\mathcal{H}^+}}_i,\quad \psi_i|_{\Sigma_{t_i^*}}=0,\quad r^{\frac{3}{2}-s}\psi_i|_{\mathcal{I}}=0,\quad s=\sqrt{\frac{9}{4}-\alpha}
    \end{cases}
\end{align}
for $\psi_i$ $\in$ $CH^k_{KAdS}$. Here ${h_{\mathcal{H}^+}}_i$ is the appropriately truncated scattering data, satisfying
\begin{align}
    {h_{\mathcal{H}^+}}_i(t^*,r=r_+,\theta,\phi)=\begin{cases}
        h_{\mathcal{H}^+}(t^*,\theta,\phi)\text{  for  } t^*\leq t_i^*-\vartheta_i,\quad\vartheta_i=\frac{1}{10} t_i^*\\
        0\text{  for  }t^*=t_i^*
    \end{cases} \label{Truncation}
\end{align}
and interpolating smoothly between its two cases.

\begin{figure}[h]
    \centering
    \includegraphics[width=0.25\textwidth]{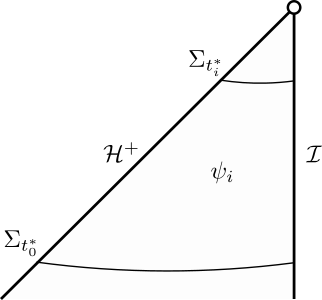}
    \caption{The spacetime region on which the finite-in-time problems are constructed.}
\end{figure}

We prove the following estimates on the $\psi_i$ and their derivatives.

\begin{theorem}[Uniformly bounded solutions of the finite problems]
\label{FiniteUniformThm}
Let $t_0^*>0$ be fixed, $k\geq 9$ and let ${h_{\mathcal{H}^+}}_i:\mathbb{R}^+\times S^2_{t^*,r_+}\to\mathbb{R}$ be sufficiently smooth truncated scattering data with
\begin{align}
\label{DataDecayAssumption}
    \mathfrak{D}_{k}<\infty
\end{align}
for associated constant (see Definition \ref{def:data}) $\Upsilon_{\kappa, k}\geq B_k$, with $B_k>0$ sufficiently large depending on $k$ and the spacetime surface gravity $\kappa$. Then there exists a solution $\psi_i:\{t^*\text{ }\in\text{ }[t_0^*,t_i^*]\}\times\{r\geq r_+\}\times S^2_{t^*,r}\to\mathbb{R}$ $\in$ $CH^k_{KAdS}$ of the nonlinear scattering problem (\ref{FiniteProblem}) satisfying
\begin{align}
\label{BootstrapAssumption}
    ||\psi_i||^2_{H^k_{KAdS}(\Sigma_{t^*})}&\leq\overline{C}\exp(-B_{k}\cdot t^*)\mathfrak{D}_{k}
\end{align}
for some $\overline{C}>0$. Moreover, the solution satisfies
\begin{equation}
\begin{aligned}
\label{4.1Uniform}
    \sum_{|\sigma|\leq k-3}\lvert\lvert r^{\frac{1}{2}}D^{\sigma}\psi_i\rvert\rvert^2_{L^{\infty}(\Sigma_{t^*})}&\leq C\exp(-B_{k}\cdot t^*)\mathfrak{D}_{k},\\
    \sum_{|\sigma|\leq k-4}\lvert\lvert r^{\frac{1}{2}}\overline{D}^{\sigma}\psi_i\rvert\rvert^2_{L^{\infty}(\Sigma_{t^*})}&\leq C\exp(-B_{k}\cdot t^*)\mathfrak{D}_{k}.
\end{aligned}
\end{equation}
for some $C>0$.
\end{theorem}

\subsection{The bootstrap argument}
\label{sec:Bootstrap}
We prove Theorem \ref{FiniteUniformThm} via the following bootstrap argument. Openness of the bootstrap set $\mathcal{B}$ is given by Theorem \ref{BootstrapThm}, which is stated and proven in Section \ref{sec:FiniteEstimates}. For the duration of this section and Section \ref{sec:FiniteEstimates}, $i$ is fixed.

\begin{proof}[Proof of Theorem \ref{FiniteUniformThm}]
Consider the set
\begin{align}
\label{BootstrapRegion}
    \mathcal{B}=\bigg\{t^*\in[t_0^*,t_i^*]\text{ }\bigg|\text{ }\text{ }&\exists\text{ a solution $\psi_i$ in $CH^k_{KAdS}$ of (\ref{FiniteProblem}) on $[t^*,t_i^*]$ satisfying (\ref{BootstrapAssumption}) $\forall$ $\tau$ $\in$ $[t^*,t_i^*]$}\bigg\}.
\end{align}
We will show that $\mathcal{B}=[t_0^*,t_i^*]$.\par

\begin{itemize}
    \item Firstly, by local existence of solutions of (\ref{FiniteProblem}), $\mathcal{B}$ is non-empty.
    \item  Furthermore, let $\tilde{\tau}_j$ be a sequence of times in $\mathcal{B}$ converging to a limit $\tilde{\tau}$. Then by local existence and the fact that the time of existence depends only on (\ref{BootstrapAssumption}) which is uniform, the limit $\tilde{\tau}$ must also lie in $\mathcal{B}$. As such, $\mathcal{B}$ is closed.
    \item Now, let $t^*$ $\in$ $\mathcal{B}$. Then, by continuity of the solution $\psi_i$ and a further application of local existence, there exists a small open ball $\mathbf{B}_{\epsilon}(t^*)=\mathcal(t^*-\epsilon,t^*+\epsilon)$ of radius $\epsilon>0$ on which $\psi_i$ solves (\ref{FiniteProblem}) and satisfies
\begin{align}
\label{3/4}
||\psi_i(\tau)||_{H_{\mathrm{KAdS}}^k(\Sigma_{\tau})}^2\leq 2\overline{C}\exp(-B_k\cdot\tau)\mathfrak{D}_{k}.
\end{align}
Theorem \ref{BootstrapThm} below implies that the bound (\ref{3/4}) can be improved
\begin{align*}
||\psi_i(\tau)||_{H_{\mathrm{KAdS}}^k(\Sigma_{\tau})}^2\leq \overline{C}\exp(-B_k\cdot\tau)\mathfrak{D}_{k},
\end{align*}
so that $\mathbf{B}_{\epsilon}(t^*)\subset\mathcal{B}$ and $\mathcal{B}$ is open.
\end{itemize}
Hence, $\mathcal{B}$ is open, closed and non-empty, implying $\mathcal{B}=[t_0^*,t_i^*]$.\par

The uniform bound (\ref{4.2Uniform}) follows from Theorem \ref{BootstrapThm}, completing the proof.
\end{proof}

\subsection{Improving the bootstrap assumption}
\label{sec:FiniteEstimates}
The following bootstrap assumption improvement result gives openness of the bootstrap set $\mathcal{B}$ and the uniform bound (\ref{4.2Uniform}). These are required for the proof of Theorem \ref{FiniteUniformThm} appearing in the previous section. The proof follows from the energy estimate for the linear inhomogeneous problem (Proposition \ref{prop:LeaveFAlone}) and that for the nonlinearity (Proposition \ref{prop:AbstractEnergyEst}).

\begin{theorem}[Improving the bootstrap assumption]
\label{BootstrapThm}
Let $t_0^*>0$ be fixed, $k\geq 9$ and let $h_{\mathcal{H}^+}$ be as in Theorem \ref{FiniteUniformThm}. Furthermore, let $t_0^*<T<t_i^*$ and suppose there exists a solution $\psi_i$ $\in CH^k_{KAdS}$ of problem (\ref{FiniteProblem}) satisfying (\ref{BootstrapAssumption}) for some $\overline{C}>0$ sufficiently large at $\tau$ $\in$ $[T,t_i^*]$. Then $\psi_i$ decays in $t^*$ with an improved constant
\begin{align}
\label{ImprovedConstant}
    ||\psi_i||^2_{H^k_{KAdS(\Sigma_{\tau})}}\leq\frac{\overline{C}}{2}\exp(-B_k\cdot\tau)\mathfrak{D}_{k}.
\end{align}
and satisfies
\begin{equation}
\begin{aligned}
\label{4.2Uniform}
    \sum_{|\sigma|\leq k-3}\lvert\lvert r^{\frac{1}{2}}D^{\sigma}\psi_i\rvert\rvert^2_{L^{\infty}(\Sigma_{t^*})}&\leq C\exp(-B_{k}\cdot t^*)\mathfrak{D}_{k},\\
    \sum_{|\sigma|\leq k-4}\lvert\lvert r^{\frac{1}{2}}\overline{D}^{\sigma}\psi_i\rvert\rvert^2_{L^{\infty}(\Sigma_{t^*})}&\leq C\exp(-B_{k}\cdot t^*)\mathfrak{D}_{k}.
\end{aligned}
\end{equation}
for $B_{n}=C_{\kappa,n}$ and some $C>0$.
\end{theorem}

\begin{proof}
By Propositions \ref{prop:LeaveFAlone} and \ref{prop:AbstractEnergyEst}, as well as estimate (\ref{EstOtherFTerms}), $\psi_i$ satisfies
\begin{align}
    &||\psi_i||^2_{H^k_{KAdS}(\Sigma_{\tau})}\nonumber\\
    \leq& C\bigg[F^k_{\mathcal{H}^+\cap[\tau,t_i^*]}[\psi_i]+||\psi_i||^2_{H^k_{KAdS}(\Sigma_{t_i^*})}+||\psi_i||^4_{H^{k-1}_{KAdS}(\Sigma_{\tau})}+\int_{\tau}^{t_i^*}||\psi_i||^2_{H^{k-2}_{KAdS}(\Sigma_{t^*})}||\psi_i||^2_{H^{k-3}_{KAdS}(\Sigma_{t^*})}\mathrm{d}t^*\nonumber\\
    &\quad+\int_{\tau}^{t_i^*}C_{\kappa,k}||\psi_i||^2_{H^k(\Sigma_{t^*})}\mathrm{d}t^*+\bigg(\int_{\tau}^{t_i^*}||\psi_i||_{H^k_{KAdS}(\Sigma_{t^*})}||\psi_i||_{H^{k-1}_{KAdS}(\Sigma_{t^*})}\mathrm{d}t^*\bigg)^2\bigg]\nonumber\\     
    \leq&C\bigg[F^k_{\mathcal{H}^+\cap[\tau,t_i^*]}[\psi_i]+||\psi_i||^2_{H^k_{KAdS}(\Sigma_{t_i^*})}+\bigg(\int_{\tau}^{t_i^*}||\psi_i||_{H^k_{KAdS}(\Sigma_{t^*})}||\psi_i||_{H^{k-1}_{KAdS}(\Sigma_{t^*})}\mathrm{d}t^*\bigg)^2\nonumber\\
    &\quad+\int_{\tau}^{t_i^*}\bigg(||\psi_i||^2_{H^{k-1}_{KAdS}(\Sigma_{t^*})}+||\psi_i||^2_{H^{k-3}_{KAdS}(\Sigma_{t^*})}+C_{\kappa,k}\bigg)||\psi_i||^2_{H^k(\Sigma_{t^*})}\mathrm{d}t^*\bigg]\label{ComboEstimate}   
\end{align}
where
\begin{align*}
   F^k_{\mathcal{H}^+\cap[\tau,t_i^*]}[\psi_i]\leq C\exp(-2\Upsilon_{\kappa,k}\cdot t^*)\mathfrak{D}_{k}
\end{align*}
for some $C>0$, by truncation of the scattering data (\ref{Truncation}).\par

By the decay assumption (\ref{DataDecayAssumption}) on the approximate scattering data ${h_{\mathcal{H}^+}}_i$, the assumption that $\psi_i|_{\Sigma_{t_i^*}}=0$, and applying the bootstrap assumption (\ref{BootstrapAssumption}) to the lower order norm in the final term on the right-hand side, (\ref{ComboEstimate}) becomes
\begin{align}
    ||\psi_i||^2_{H^k_{KAdS}(\Sigma_{\tau})}\leq&C\bigg[\exp(-2\Upsilon_{\kappa,k}\cdot t^*)\mathfrak{D}_{k}+\int_{\tau}^{t_i^*}\bigg(2\overline{C}\exp(-B_k\cdot\tau)\mathfrak{D}_k+C_{\kappa,k}\bigg)||\psi_i||^2_{H^k(\Sigma_{t^*})}\mathrm{d}t^*\nonumber\\
        &\qquad+\bigg(\int_{\tau}^{t_i^*}||\psi_i||_{H^k_{KAdS}(\Sigma_{t^*})}\Big(\overline{C}\exp(-B_k\cdot t^*)\mathfrak{D}_{k}\Big)^{\frac{1}{2}}\mathrm{d}t^*\bigg)^2\bigg]\nonumber\\
        \leq&C\bigg[\exp(-2\Upsilon_{\kappa,k}\cdot t^*)\mathfrak{D}_{k}+\int_{\tau}^{t_i^*}\bigg(3\overline{C}\exp(-B_k\cdot\tau)\mathfrak{D}_k+C_{\kappa,k}\bigg)||\psi_i||^2_{H^k(\Sigma_{t^*})}\mathrm{d}t^*\bigg]\label{BootstrapApplied}
\end{align}
Let
\begin{gather*}
    a(C_{\kappa,k})=C\bigg(3\overline{C}\exp(-B_k\cdot\tau)\mathfrak{D}_k+C_{\kappa,k}\bigg),\nonumber\\
    b(t^*)=C\exp(-2\Upsilon_{\kappa,k}\cdot t^*)\mathfrak{D}_{k}.
\end{gather*}
Then a Gr\"onwall inequality \cite[Proposition 2.6]{Gronwall} yields
\begin{align}
     &||\psi_i||^2_{H^k_{KAdS}(\Sigma_{\tau})}\nonumber\\
     \leq&b(t^*)+\int_{\tau}^{t_i^*}a(C_{\kappa,k})b(t^*)\exp\bigg(\int_{\tau}^{t^*}a(C_{\kappa,k})\mathrm{d}s\bigg)\mathrm{d}t^*\nonumber\\
     =&b(t^*)+\int_{\tau}^{t_i^*}a(C_{\kappa,k})b(t^*)\exp\Big((t^*-\tau)a(C_{\kappa,k})\Big)\mathrm{d}t^*\nonumber\\
     =&b(t^*)+C\exp(-a(C_{\kappa,k})\cdot\tau)\int_{\tau}^{t_i*}a(C_{\kappa,k})\mathfrak{D}_k\exp\Big((a(C_{\kappa,k})-2\Upsilon_{\kappa,k})t^*\Big)\mathrm{d}t^*\nonumber\\
     =&b(t^*)-C\exp(-a(C_{\kappa,k})\cdot\tau)\int_{\tau}^{t_i*}\frac{a(C_{\kappa,k})\mathfrak{D}_k}{2\Upsilon_{\kappa,k}-a(C_{\kappa,k})}\partial_{t^*}\bigg(\exp\Big((a(C_{\kappa,k})-2\Upsilon_{\kappa,k})t^*\Big)\bigg)\mathrm{d}t^*\nonumber\\
     \leq& b(t^*)+C\exp(-a(C_{\kappa,k})\cdot\tau)\frac{a(C_{\kappa,k})\mathfrak{D}_k}{2\Upsilon_{\kappa,k}-a(C_{\kappa,k})}\exp\Big((a(C_{\kappa,k})-2\Upsilon_{\kappa,k})\tau\Big)\nonumber\\
     \leq&C\exp(-2\Upsilon_{\kappa,k}\cdot\tau)\mathfrak{D}_k\bigg(1+\frac{a(C_{\kappa,k})}{2\Upsilon_{\kappa,k}-a(C_{\kappa,k})}\bigg)\label{PreG}
\end{align}
The desired improved estimate
\begin{align}
    \label{FinalGronwall}
    ||\psi_i||^2_{H^k_{KAdS(\Sigma_{\tau})}}\leq\frac{\overline{C}}{2}\exp(-B_k\cdot\tau)\mathfrak{D}_{k}
\end{align}
follows, provided $2\Upsilon_{\kappa,k}\geq B_k$ and
\begin{align}
    \label{Requirements}
    \frac{\overline{C}}{2}\bigg[\frac{2C}{\overline{C}}+\frac{2Ca(C_{\kappa,k})}{\overline{C}(2\Upsilon_{\kappa,k}-a(C_{\kappa,k}))}\bigg]\leq\frac{\overline{C}}{2}.
\end{align}
Inequality (\ref{Requirements}) holds, provided $2\Upsilon_{\kappa,k}>a(C_{\kappa,k})$ and $\overline{C}$ and $B_k$ are sufficiently large, depending on $C_{\kappa,k}$, $C$ and $\mathfrak{D}_k$.\par

To obtain the uniform bounds (\ref{4.2Uniform}), we apply Theorem \ref{SliceSobolev}. For each $|\sigma|<k-3$, this gives
\begin{align*}
    &||\overline{D}^{\sigma}\psi_i||_{L^{\infty}(\Sigma_{\tau})}\nonumber\\
     \leq&\frac{C}{r^{\frac{1}{2}}}\Bigg[\sum_{\substack{|\alpha|\leq |\sigma|+2,\\|\beta|\leq|\sigma|+1,\\|\gamma|\leq\big\lceil\frac{|\sigma|+1}{2}\big\rceil}}\bigg[||\Gamma^{\alpha}\psi||^2_{H^1_{KAdS}(\Sigma_{\tau})}+\int_{\Sigma_{\tau}}\bigg(\Big(\Gamma^{\beta}(\Box_g\psi+\alpha\psi)\Big)^2+\Big(\Box_g^{\gamma}(\Box_g\psi+\alpha\psi)\Big)^2\bigg) r^2\mathrm{d}r\mathrm{d}\omega\bigg]\Bigg]^{\frac{1}{2}}\nonumber\\
     =&\frac{C}{r^{\frac{1}{2}}}\Bigg[\sum_{\substack{|\alpha|\leq|\sigma|+2,\\|\beta|\leq|\sigma|+1,\\|\gamma|\leq\big\lceil\frac{|\sigma|+1}{2}\big\rceil}}\bigg[||\Gamma^{\alpha}\psi||^2_{H^1_{KAdS}(\Sigma_{\tau})}+\int_{\Sigma_{\tau}}\bigg((\Gamma^{\beta}\mathcal{F})^2+(\Box_g^{\gamma}\mathcal{F})^2\bigg) r^2\mathrm{d}r\mathrm{d}\omega\bigg]\Bigg]^{\frac{1}{2}}\nonumber\nonumber\\
     \leq&\frac{C}{r^{\frac{1}{2}}}\Bigg[||\psi||^2_{H^{|\sigma|+3}_{KAdS}(\Sigma_{\tau})}+||\psi||^2_{H^{|\sigma|+2}_{KAdS}(\Sigma_{\tau})}||\psi||^2_{H^{|\sigma|+1}_{KAdS}(\Sigma_{\tau})}\Bigg]^{\frac{1}{2}}\nonumber\\
     \leq&\frac{C}{r^{\frac{1}{2}}}\bigg(\frac{\overline{C}}{2}\bigg)^{\frac{1}{2}}\exp\bigg(-\frac{B_k}{2}\cdot\tau\bigg)\mathfrak{D}_{|\sigma|+3}^{\frac{1}{2}},
\end{align*}
where the third inequality follows from estimate (\ref{EstOtherFTerms}) and the last from (\ref{FinalGronwall}). For $|\sigma|=k-3$, one obtains the same estimate for $D^{\sigma}$.
\end{proof}

\subsection{Convergence of the solution}
\label{sec:Convergence}
All that remains is to prove Theorem \ref{MainResult} as a consequence of Theorem \ref{FiniteUniformThm}.
\begin{proof}[Proof of Theorem \ref{MainResult}]
By Theorem \ref{FiniteUniformThm}, we have uniform boundedness of the solutions $\psi_i$ to the finite problems (\ref{FiniteProblem}) and their first-order derivatives $D\psi_i$. Since we have set $\psi_i|_{\Sigma_{t_i^*}}=0$, each solution $\psi_i$ trivially extends to a function which is identically zero on $[t_i^*,\infty)$. Thus, we may apply the \textit{Arz\`ela-Ascoli Theorem} on the domain $[t_0^*,\infty)\times \{r\geq r_{\mathcal{H}}\}\times S^2_{t^*,r}$ to obtain a uniformly convergent subsequence ${\psi_i}_j\to\psi$. We may also apply \textit{Arz\`ela-Ascoli} to the derivatives $\{D{\psi_i}_j,D^2{{\psi_i}_j}_k,\dots\}$ up to sufficiently high order to obtain uniformly convergent subsequences ${D\psi_i}_{j}\to D\psi$, ${{D^2\psi_i}_{j}}_k\to D^2\psi$. Thus, the limiting function $\psi$ solves (\ref{EqnOfChoice}). Finally, the uniform bound (\ref{4.1Uniform}) gives the estimates (\ref{MainEstimate2}).
\end{proof}
\section*{Acknowledgements}
I would like to thank Martin Taylor, Gustav Holzegel and Claude Warnick for their insightful comments during the preparation of this paper, as well as the Engineering and Physical Sciences Research Council (EPSRC) for funding this project.

\end{document}